\documentclass[11pt,reqno]{amsart}

\usepackage{graphics,graphicx,fontenc,mathabx,wasysym,cancel}
\usepackage{epsfig,psfrag,manfnt}
\usepackage{bbm,subfigure,rotating,bm,float,mathdots,mathrsfs,slashed}
\usepackage[all,cmtip]{xy}
\usepackage{amssymb,amsmath,amsfonts,amsthm,color}

\DeclareFontFamily{U}{MnSymbolC}{}
\DeclareSymbolFont{MnSyC}{U}{MnSymbolC}{m}{n}
\DeclareFontShape{U}{MnSymbolC}{m}{n}{
    <-6>  MnSymbolC5
   <6-7>  MnSymbolC6
   <7-8>  MnSymbolC7
   <8-9>  MnSymbolC8
   <9-10> MnSymbolC9
  <10-12> MnSymbolC10
  <12->   MnSymbolC12}{}
\DeclareMathSymbol{\intprod}{\mathbin}{MnSyC}{'270}

\makeatletter
\newcommand*\bigcdot{\mathpalette\bigcdot@{.5}}
\newcommand*\bigcdot@[2]{\mathbin{\vcenter{\hbox{\scalebox{#2}{$\m@th#1\bullet$}}}}}
\makeatother

\newcommand{\calC}{\mathcal{C}}
\newcommand{\calD}{\mathcal{D}}
\newcommand{\calE}{\mathcal{E}}

\newcommand{\calH}{\mathcal{H}}

\newcommand{\calL}{\mathcal{L}}

\newcommand{\calR}{\mathcal{R}}
\newcommand{\calS}{\mathcal{S}}
\newcommand{\calT}{\mathcal{T}}

\newcommand{\mN}{\mathbb{N}}

\newcommand{\mR}{\mathbb{R}}

\newcommand{\mT}{\mathbb{T}}

\newcommand{\mZ}{\mathbb{Z}}

\newcommand{\bbk}{\mathbf{k}}

\newcommand{\bbn}{\mathbf{n}}

\newcommand{\bbx}{\mathbf{x}}

\newcommand{\balpha}{{\bm{\alpha}}}

\newcommand{\bmu}{\bm{\mu}}

\newcommand{\bxi}{\bm{\xi}}

\newcommand{\cnab}{{\mathring{\slashed{\nabla}}}}
\newcommand{\cLap}{{\mathring{\slashed{\Delta}}}}
\newcommand{\cg}{{\mathring{\slashed{g}}}}

\newtheorem{theorem}{Theorem}[section]
\newtheorem{lemma}[theorem]{Lemma}

\theoremstyle{definition}

\newtheorem{remark}[theorem]{Remark}

\theoremstyle{definition}

\theoremstyle{definition}

\theoremstyle{definition}

\begin{document}

\keywords{Klein-Gordon equation, general relativity, FLRW models,
  Schwarzschild-de Sitter, Reissner-Nordstr\"om-de Sitter}

\subjclass[2010]{Primary 58J45; Secondary 35B40, 35L15, 83C57, 83F05}

\title[Decay of Klein-Gordon equation solutions]{Decay of solutions  \\
to the Klein-Gordon equation\\
  on some expanding cosmological spacetimes}

\vspace{-0.3in}
  
\author[J. Nat\'ario]{Jos\'e Nat\'ario}
\address{CAMGSD, Departamento de Matem\'atica\\ Instituto Superior T\'ecnico\\
Universidade de Lisboa\\ Portugal}

\email{jnatar@math.ist.utl.pt}

\author[A. Sasane]{Amol Sasane}
\address{Department of Mathematics \\London School of Economics\\
     Houghton Street\\ London WC2A 2AE\\ United Kingdom}
\email{A.J.Sasane@lse.ac.uk}

\maketitle 

\vspace{-0.36in}

\begin{abstract}
  The decay of solutions to the Klein-Gordon equation is studied in
  two expanding cosmological spacetimes, namely
\begin{itemize}
\item the de Sitter universe in flat Friedmann-Lema\^{i}tre-Robertson-Walker (FLRW) form, and
\item the cosmological region of the Reissner-Nordstr\"om-de Sitter
  (RNdS) model.
\end{itemize}
Using energy methods, for initial data with finite higher order
energies, decay rates for the solution are obtained. 
Also, a previously established decay rate of the time derivative 
of the solution to the wave equation, in an expanding de Sitter universe in flat FLRW form,  
is improved, proving Rendall's conjecture. A similar improvement is also given for the 
wave equation in the cosmological region of the RNdS spacetime. 
\end{abstract}

\vspace{-0.21in}

\tableofcontents

\section{Introduction}

\noindent 
The aim of this article is to obtain exact decay rates for solutions
to the Klein-Gordon equation in a fixed background of some expanding
cosmological spacetimes.  The two spacetimes we
will consider are the de Sitter universe in flat Friedmann-Lema\^{i}tre-Robertson-Walker (FLRW)  form, and  the cosmological region of the
Reissner-Nordstr\"om-de Sitter (RNdS) model. 
The problem we consider is
linear, in that the background is fixed. This constitutes a 
first step towards understanding the more complicated nonlinear coupled problem, 
where one also considers the effect of the energy-momentum tensor of the solution to the 
Klein-Gordon equation on the Einstein equation. This 
nonlinear coupled problem is much more complicated, and usually requires, 
as a first step, a detailed understanding of our simpler linear problem.

There are several motivations behind the interest in this question.
Firstly, one may consider the linear wave equations as a proxy for the
Einstein equations, with the ultimate goal of understanding the
qualitative behaviour of solutions to the Einstein equations by
slicing spacetime into spacelike hypersurfaces.  After this first step, 
one may then proceed to consider linearised Einstein equations
(which can be reduced to tensor wave-like linear equations), and
finally, the full nonlinear Einstein equations. With the addition of a
positive cosmological constant to the Einstein field equations, the
expectation is that the resulting accelerated expansion has a
dominating effect on the decay of solutions. Precise estimates on
solutions may then prove useful in formulating and proving cosmic
no-hair theorems (e.g. \cite{AR}, \cite{CNO0}).

The wave equation $\square_g \phi=0$ in expanding cosmological spacetimes
$(M,g)$ has been amply studied in the literature, see for example
\cite{BGP},  \cite{CNO}, \cite{DR}, \cite{S},  and the references therein.  It is a
natural question to also study the Klein-Gordon equation
$\square_g\phi-m^2\phi=0$, the degenerate version of which, when $m^2=0$,
is the wave equation.  For example, in \cite[\S6]{S}, also the case of
the Klein-Gordon equation in the Schwarzschild-de Sitter
spacetime is considered. In \cite{Rin}, the asymptotic behaviour of the solutions to the 
Klein-Gordon equation near the Big Bang singularity is studied, while
we investigate the asymptotics of the Klein-Gordon equation in the far
future in the case of the de Sitter universe in flat FLRW form, and in the cosmological
region of the Reissner-Nordstr\"om-de Sitter solution. 
Recently, in \cite{ER}, among other things,  decay estimates for the solutions to the Klein-Gordon equation 
 were obtained in de Sitter models (see in particular, Corollary 2.1 
and the less obvious Proposition 3.1). However, these results are proved 
via Fourier transformation (reminiscent of our mode calculation in Appendix A) 
and do not seem to be as sharp as our Theorem~\ref{theorem_3}.

The wave equation in the de Sitter spacetime having flat
$3$-dimensional spatial sections was considered in Rendall \cite{Ren}.
There it was shown that the time derivative
$\partial_t \phi=:\dot{\phi}$ decays at least as
$e^{-Ht}=(a(t))^{-1}$, where $H=\sqrt{\Lambda/3}$ is the Hubble
constant, and $\Lambda>0$ is the cosmological constant. Moreover, it was
conjectured that the decay is of the order
$e^{-2Ht}=(a(t))^{-2}$. The almost-exact conjectured decay rate 
of $|\dot{\phi}|\lesssim (a(t))^{-2+\delta}$ (where $\delta>0$ can be chosen arbitrarily at the outset) 
follows as a corollary of a result shown recently \cite[Remark~1.1]{CNO}. We
improve this result, to obtaining full conformity with Rendall's conjecture, in
our result Theorem~\ref{theorem_2} below.

Finally, from the pure mathematical perspective, analysis of linear
wave equations on Lorentzian manifolds is a natural topic of study
within the realm of hyperbolic partial differential equations and 
differential geometry; see for example \cite{A}, \cite[\S7, Chap.2]{Tay}.

A naive heuristic indication of the effect of the accelerated
expansion on the decay of the solution, based on physical energy
considerations, can be obtained as follows. Considering an expanding
FLRW model with flat $n$-dimensional spatial sections of radius
$a(t)$, we have on the one hand that the energy density of a solution
$\phi$ of the Klein-Gordon equation is of the order of $m^2 \phi^2$.
On the other hand, if the wavelength of the particles associated with
$\phi$ follows the expansion, then it is proportional to $a(t)$, and
so the energy varies as
$$
E^2\sim m^2+p^2\sim A+\frac{B}{(a(t))^2},
$$
where $A,B>0$ are constants. Thus  
$\displaystyle 
m^2 \phi^2  (a(t))^n \;\propto\; \left(A +\frac{B}{(a(t))^2}\right), 
$  
giving
$$
m^2 \phi^2  \sim (a(t))^{-n} \left(A+\frac{B}{(a(t))^2} \right).
$$
As $\dot{a}\geq 0$ (expanding FLRW spacetime), the term $
A+\frac{B}{(a(t))^2}
$ 
approaches a finite positive value, and so one may expect
$$
\phi\sim (a(t))^{-\frac{n}{2}}.
$$
We will find out that in fact things are much more complicated: this
decay rate is valid only for $|m|\geq \frac{n}{2}$.  In order to
obtain precise conjectures on the expected decay, we will consider
Fourier modes for spatially-periodic solutions to the Klein-Gordon equation, or
equivalently, consider the expanding de Sitter universe in flat FLRW form   with toroidal
spatial sections. This exercise already demonstrates that the
underlying decay mechanism is the cosmological expansion, as opposed
to dispersion. The Fourier mode analysis, which is
peripheral to the rest of the paper, is relegated to Appendix~A.

In the cosmological region of the Reissner-Nordstr\"om-de Sitter spacetimes, 
the expanding region is foliated by spacelike hypersurfaces of `constant $r$'. 
One expects the decay rate with respect to $r$, for the solution to the Klein-Gordon equation, 
in the cosmological region of the Reissner-Nordstr\"om-de Sitter spacetime, 
to be the same as the one for the de Sitter universe in flat FLRW form, when $e^t$ is replaced by $r$. We show that this 
expectation is correct, and a suitable modification of the technique used in the 
case of the de Sitter universe in flat FLRW form, does enable one to obtain the expected decay rates also for the 
case of the Reissner-Nordstr\"om-de Sitter spacetime. 

\goodbreak 

\medskip 

\noindent 
Our main results are as follows:
\begin{itemize}
\item Theorem~\ref{theorem_2} considers the $m=0$ case (wave equation), and we obtain a decay estimate on $\partial_t{\phi}$,   improving a corollary of
  \cite[Theorem~1]{CNO}, and  proving the aforementioned Rendall's
  conjecture. 

\item Theorem~\ref{theorem_2b} improves \cite[Theorem~2]{CNO}, and we obtain a decay estimate on $\partial_r \phi$, using a similar method to the 
one we use for proving Rendall's conjecture.

\item Theorem~\ref{theorem_3} gives the decay rate of the
  solutions $\phi$  to the Klein-Gordon equation in the de Sitter universe in flat FLRW form.
\item Theorem~\ref{theorem_4} gives the decay rate of the
  solutions $\phi$ to the Klein-Gordon equation in the
  cosmological region of the RNdS model.
\end{itemize}
Theorems~\ref{theorem_2},  \ref{theorem_3}, \ref{theorem_4}, \ref{theorem_2b} are stated
and proved in Sections \ref{Section_Rendall},  \ref{Section_FLRW},
\ref{Section_RNdS}, \ref{Section_thm_2b}, respectively.  The Fourier mode analysis for
spatially-periodic solutions to the Klein-Gordon equation is given in Appendix A,
while Appendix B contains a technical lemma which is needed in the proof 
of Theorem~\ref{theorem_3}. 
Finally, in Appendix C, we establish the sharpness of the bound of the $|m|=\frac{n}{2}$ case of 
Theorem~\ref{theorem_3}.

\subsection{Relation of our results to previous work} 
Our decay rates for the Klein-Gordon equation solutions in the case of the de Sitter universe can be retrieved from the article \cite{Vas} by setting $x=e^{-t}$, $Y=\mathbb{R}^n$ therein. However, the methods used are entirely different: our proof in this case is more explicit, and more elementary (relying on energy methods, rather than technical tools from microlocal analysis of partial differential operators). 

In the article \cite{Gaj}, the Klein-Gordon equation is studied in the Nariai spacetime using energy methods, and en route it is also established that solutions of the Klein-Gordon equation decay exponentially in the de Sitter case (with spherical spatial sections). However, the decay rates are not given explicitly. 

The article \cite{DafRod} contains a general discussion of redshift estimates,  
which we use to prove our results in the context of the Reissner-Nordstr\"om-de Sitter spacetime. 
Similar estimates are used in the article \cite{S} to study the wave equation in the Schwarzschild-de Sitter spacetime, of which the Reissner-Nordstr\"om-de Sitter spacetime is a perturbation for large radius. Nevertheless, we do not appeal to these results, and instead of extracting what we need from these sources, we give a less technical, self-contained derivation for the convenience of the reader in \S\ref{Subsection_redshift_estimates}. Here we follow \cite{CNO} (where a 
similar derivation was given for the wave equation). 

\medskip 

\noindent {\bf Acknowledgements:} We thank Pedro Gir\~{a}o for suggesting the idea behind the proofs of Theorems 2.2 and 5.3. 
JN was partially supported by FCT/Portugal through UID/MAT/04459/2013 and grant (GPSEinstein) PTDC/MAT-ANA/1275/2014.

\section{Decay in the de Sitter universe in flat FLRW form; $m=0$}
\label{Section_Rendall}

\noindent 
In \cite[Theorem~1]{CNO}, the following result was shown:

\begin{theorem}$\;$\label{theorem_1}

\noindent 
Suppose that 
\begin{itemize}
\item $\delta>0$,
\item $I\subset \mR$ is an open interval of the form $(t_*,+\infty)$, $t_0\in I$, 
\item $a(\cdot) \in C^1(I)$ with $\dot{a}(t)\geq 0$ for $t\geq t_0$, and
   $\epsilon>0$ is such that
 
 \smallskip 
 
 \noindent $
 \phantom{aaaaaaaaaaaaaaaa}\displaystyle \int_{t_0}^\infty \frac{1}{(a(t))^\epsilon} dt <+\infty,
 $ 
\item $n\geq 2$, 
\item $(M,g)$ is an expanding {\em FLRW} spacetime with flat
  $n$-dimensional 
  
  \noindent sections, given by $I\times \mR^n$, with the metric
   \begin{equation}
   \label{equation_10_august_2020_16:53}
   g=-dt^2 +(a(t))^2 \left( (dx^1)^2+\cdots+(dx^n)^2\right),
   \end{equation}
\item $k>\frac{n}{2}+2$, $\;\;\phi_0\in H^k(\mR^n)$,
   $\;\;\phi_1\in H^{k-1}(\mR^n)$, and
\item $\phi$ is a smooth solution to the Cauchy problem 
 $$
 \left\{ \begin{array}{rcll}
          \square_g \phi&=&0,& \;\;\;\;(t\geq t_0,\;\bbx\in \mR^n), \\
          \phi(t_0,\bbx)&=& \phi_0(\bbx) &\;\;\;\;(\bbx\in \mR^n),\\
          \partial_t \phi(t_0,\bbx)&=& \phi_1(\bbx) & \;\;\;\;(\bbx\in \mR^n).
         \end{array}\right.
 $$
\end{itemize}
Then 
$$
\forall t\geq t_0,\;\;\;\;\|\partial_t \phi(t,\cdot)\|_{L^\infty(\mR^n)}\lesssim \left(a(t)\right)^{-2+\epsilon+\delta}.
$$
\end{theorem}

\noindent 
Here, the symbol $\lesssim $ is used to mean that there exists
a constant $C(\delta)$, independent of $\epsilon$, such that
$$
\|\partial_t \phi(t,\cdot)\|_{L^\infty(\mR^n)}\leq C(\delta) \left(a(t)\right)^{-2+\epsilon+\delta}.
$$
We also use the standard notation $H^k(\mR^n)$ for the Sobolev space, 
$$
\|\phi\|^2_{H^k(\mR^n)} := \int_{\mR^n} \sum_{|\balpha|\leq k} (\partial_\balpha\phi)^2 d^n\bbx<+\infty \;\;\textrm{ for } \phi \in H^k(\mR^n),
$$
where 
\begin{eqnarray*}
 \balpha&=&(\alpha_1,\cdots,\alpha_n)\in \mN_0^n, \quad \mN_0=\{0,1,2,3,\cdots\},\\ 
 |\balpha |&:=&\alpha_1+\cdots+\alpha_n, \textrm{ and}\\ 
 \partial_\balpha&= &(\partial_{x_1})^{\alpha_1}\cdots (\partial_{x_n})^{\alpha_n};
\end{eqnarray*}
see for example \cite[p.249]{Wal} or \cite[Chap. 4]{Tay}.

\begin{remark}[Smoothness assumption on the solution $\phi$] $\;$

\smallskip

\noindent 
In Theorem~\ref{theorem_1} (and later also in Theorems~\ref{theorem_2}, \ref{theorem_3}, \ref{theorem_4}, \ref{theorem_CNO_2}, \ref{theorem_2b}), we will assume, for the 
sake of simplicity of exposition, that the solution $\phi$ to the wave/Klein-Gordon equation is smooth. However, these theorems are also true without this assumption. 
To see this, we note that  for non-smooth solutions with initial data in $H^k \times H^{k-1}$, we can approximate the initial data 
by smooth functions in $H^k \times H^{k-1}$, prove the bounds for the $H^k$ norms of the corresponding solutions, 
and then take limits. Since the solution of the problem with rough initial 
data is in $C^0(I, H^k) \cap \;\!C^1(I, H^{k-1})$, these bounds will continue to be true in the limit, 
and we can then use the Sobolev embedding theorem. This enables one to drop the smoothness assumption. 
\end{remark}

\bigskip 

\noindent In Theorem~\ref{theorem_1} above, in particular, if 
$$
a(t)=e^{Ht},
$$
where $H$ is the Hubble constant,
then since $\epsilon>0$ can be taken to be arbitrarily small, we
obtain
$$
\|\partial_t \phi(t,\cdot)\|_{L^\infty(\mR^n)} \lesssim (a(t))^{-2+\delta}= e^{-(2-\delta)Ht},
$$
and this is in agreement with Rendall's conjecture up to the small
quantity $\delta>0$.  We will show below that in fact one gets the
exact rate $(a(t))^{-2}$ when $n>2$.  There is no loss of generality
in assuming that $H=1$.  Our result is the following.

\goodbreak

\begin{theorem}$\;$\label{theorem_2}

\noindent 
Suppose that 
\begin{itemize}
\item $I\subset \mR$ is an open interval of the form $(t_*,+\infty)$, $t_0\in I$, 
\item $n> 2$, 
\item $(M,g)$ is the expanding de Sitter universe in flat FLRW form, with flat
  $n$-dimensional sections, given by $I\times \mR^n$, with the metric
  $$
   g=-dt^2 +e^{2t} \left( (dx^1)^2+\cdots+(dx^n)^2\right),
  $$
\item $k>\frac{n}{2}+2$, $\;\;\phi_0\in H^k(\mR^n)$, $\;\;\phi_1\in H^{k-1}(\mR^n)$, and 
\item $\phi$ is a smooth solution to the Cauchy problem 
 $$
 \left\{ \begin{array}{rcll}
          \square_g \phi&=&0,&\;\; \;\;(t\geq t_0,\;\bbx\in \mR^n), \\
          \phi(t_0,\bbx)&=& \phi_0(\bbx) &\;\;\;\;(\bbx\in \mR^n),\\
          \partial_t \phi(t_0,\bbx)&=& \phi_1(\bbx) &\;\;\;\; (\bbx\in \mR^n).
         \end{array}\right.
 $$
\end{itemize}
Then 
$$
\forall t\geq t_0,\;\;\;\;\|\partial_t \phi(t,\cdot)\|_{L^\infty(\mR^n)}\lesssim \left(a(t)\right)^{-2}.
$$
\end{theorem}
\begin{proof} $\;$ We proceed in several steps. 
 
\smallskip 
 
\noindent {\bf Step 1: Bound on $\Delta \phi$.} 

\smallskip 

\noindent We will follow the preliminary steps of the proof of
\cite[Theorem~1]{CNO} in order to obtain a bound on $\Delta \phi$,
which will be needed in the proof of our Theorem~\ref{theorem_2}. 
We repeat this preliminary step here from \cite[\S2.2]{CNO} for the 
sake of completeness and for the convenience of the reader.

\smallskip 

\noindent For a vector field $X=X^\mu \partial_\mu$, it can be shown that 
$$
\nabla_\mu X^\mu =\frac{1}{\sqrt{-g}}\partial_\mu (\sqrt{-g} \;\!X^\mu),
$$
where $g:=\det [g_{\mu \nu}]$ is the determinant of the matrix
$[g_{\mu\nu}]$ describing the metric in the chart.  Then it follows
that
$$
\square_g  \phi=\nabla_\mu (\partial^\mu \phi)
=
\frac{1}{\sqrt{-g}} \partial_\mu(\sqrt{-g} \;\!\partial^\mu \phi).
$$
Thus $\square_g \phi=0$ can be rewritten as
$\partial_\mu(\sqrt{-g} \;\!\partial^\mu \phi)=0$.  With the metric for the de Sitter universe in flat FLRW form
 given by
$$
g=-dt^2+(a(t))^2 \left( (dx^1)^2+\cdots +(dx^n)^2\right),
$$
the wave equation can be rewritten as 
$$
\partial_\mu(a^n \partial^\mu \phi)=0,
$$
that is,
$$
-\ddot{\phi}-\frac{n\dot{a}}{a}\dot{\phi}
+\frac{1}{a^2} \delta^{ij} \partial_i \partial_j \phi=0.
$$
We recall (see e.g. \cite[Appendix~E]{Wal}) that the
energy-momentum tensor for the wave equation is
\begin{equation}
\label{eq_14_August_12:16}
T_{\mu\nu}=\partial_\mu \phi \partial _\nu \phi -
\frac{g_{\mu\nu}}{2} \partial_\alpha \phi \partial^\alpha \phi .
\end{equation}
Then it can be shown that 
$$
\nabla_\mu T^{\mu \nu}=0.
$$
From \eqref{eq_14_August_12:16}, have in particular that 
$$
T_{00}=\frac{1}{2} \left(\dot{\phi}^2 +a^{-2} \delta^{ij} \partial_i \phi \partial_j \phi\right).
$$
Define the vector field
$$
X=a^{2-n} \frac{\partial }{\partial t}.
$$
Then $X$ is future-pointing ($g(X,\partial_t)<0$) and causal ($X$ is
time-like since $g(X,X)<0$).  We form the current $J$, given by
$$
J_\mu=T_{\mu \nu} X^\nu.
$$
Then it can be shown that 
$$
J=(X\cdot \phi)\textrm{grad }\phi-\frac{1}{2} g( \textrm{grad } \phi,\textrm{grad }\phi)X.
$$
(Here $X\cdot \phi$ means the application of the vector field $X$ on $\phi$.) 
It follows that $g(J,J)\leq 0$, so that $J$ is causal. Also,  $J$ is 
past-pointing. To see this, we choose $E_1,\cdots,E_n$ orthogonal and spacelike 
such that $\{X,E_1,\cdots, E_n\}$ forms an orthogonal basis in each tangent space. 
Then expressing 
$$
\textrm{grad } \phi=c^0X+c^1E_1+\cdots c^n E_n,
$$
we obtain 
{\small 
\begin{eqnarray*}
g(J,X)\!\!&\!\!=\!\!&\!\!\!\left(g(X,\textrm{grad } \phi)\right)^2
-\frac{1}{2} g(\textrm{grad } \phi,\textrm{grad }\phi)\cdot g(X,X)\\
&\!\!=\!\!& \!\!\!\frac{(c^0)^2}{2} \!\left(g(X,X)\right)^2\!
-\!\frac{1}{2} \left((c^1)^2g(E_1,E_1)\!
+\!\cdots\!+\!(c^n)^2 g(E_n,E_n)\right) \cdot g(X,X)\!\geq \!0.
\end{eqnarray*}}

\noindent Set 
$$
N=\frac{\partial}{\partial t},
$$
the future unit normal vector field.  We define the energy $E$ by
$$
E(t)=\int_{\{t\}\times \mR^n} J_\mu N^\mu 
=\int_{\mR^n} a^2 T_{00} d^n \bbx
=\int_{\mR^n} \frac{1}{2} \left( a^2 \dot{\phi}^2 
+\delta^{ij} \partial_i \partial_j \phi\right) d^n \bbx.
$$
The deformation tensor $\Pi$ associated with the multiplier $X$ is 
$$
\Pi =\frac{1}{2} \calL_X g=-dt \calL_X dt +\dot{a} a^{3-n} \delta_{ij} dx^i dx^j.
$$
It can be shown that 
$$
\calL_X dt =(2-n)\dot{a} a^{1-n} dt.
$$
Thus 
$$
\Pi= (n-2) \dot{a} a^{1-n} dt^2 +\dot{a} a^{3-n} \delta_{ij} dx^i dx^j.
$$

\medskip 

\noindent {\bf Claim:} $\nabla_\mu J^\mu=T^{\mu\nu}\Pi_{\mu\nu}.$

\medskip

\noindent We have 
\begin{eqnarray*}
\Pi_{\mu \nu}
&=&\frac{1}{2} (\calL_X g)(\partial_\mu,\partial_\nu)
=\frac{1}{2} \left(\calL_X(g_{\mu\nu}) 
 -g(\calL_X \partial_\mu, \partial_\nu)-g(\partial_\mu, \calL_X\partial_\nu)\right)\\
&=& \frac{1}{2} \left( X(g_{\mu\nu})-g([X,\partial_\mu],\partial_\nu)-g(\partial_\mu,[X,\partial_\nu])\right).
\end{eqnarray*}
But by the definition of the Levi-Civita connection $\nabla$, 
\begin{eqnarray*}
 g(\nabla_{\partial_\mu} \partial_\nu,X)&=& 
 \frac{1}{2}\left( \partial_\mu (g(\partial_\nu,X))+\partial_\nu(g(\partial_\mu,X))-X(g(\partial_\mu,\partial_\nu))\right.\\
&&\left.+g([X,\partial_\mu],\partial_\nu)+g([X,\partial_\nu],\partial_\mu)+g(X,[\partial_\mu,\partial_\nu])\right)\phantom{\frac{1}{2}}\\
&=& \frac{1}{2} \left(\partial_\mu (g(\partial_\nu,X))+\partial_\nu (g(\partial_\mu,X))\right) 
\\
&&-\frac{1}{2}\left( X(g_{\mu\nu})-g([X,\partial_\mu],\partial_\nu)-g(\partial_\mu,[X,\partial_\nu])\right).
\end{eqnarray*}
So 
\begin{eqnarray*}
 \Pi_{\mu\nu}&=& \frac{1}{2}  \left( X(g_{\mu\nu})-g([X,\partial_\mu],\partial_\nu)-g(\partial_\mu,[X,\partial_\nu])\right)\\
 &=& \frac{1}{2}\left( \partial_\mu (g(\partial_\nu,X))+\partial_\nu (g(\partial_\mu,X))\right)-g(\nabla_\mu \partial_\nu,X).
\end{eqnarray*}
Writing $X=k\partial_0$, where $k(t):=(a(t))^{2-n}$, we have 
$$
 \Pi_{\mu\nu}= \frac{1}{2}  \left( \partial_\mu(kg_{\nu 0})+\partial _\nu(kg_{\mu 0})\right)-g(\nabla_\mu \partial_\nu,k\partial_0)
 = -\dot{k}\delta_{\mu 0}\delta_{\nu 0}+k\Gamma_{\mu \nu}^0.
 $$
 We have 
  $$
 \nabla_\mu X_\nu =\nabla _\mu (-k \delta_{\nu 0})=-\dot{k} \delta_{\mu 0} \delta_{\nu 0}+k\Gamma_{\mu \nu}^0.
 $$ 
 So $\Pi_{\mu\nu}=\nabla_\mu X_\nu$, and consequently, 
  $$
 \nabla_\mu J^\mu =\nabla_\mu (T^{\mu \nu}X_\nu)=0+T^{\mu \nu}\nabla_\mu X_\nu=T^{\mu \nu}\Pi_{\mu\nu}.
 $$ 
 This completes the proof of our claim that $\nabla_\mu J^\mu=T^{\mu\nu}\Pi_{\mu\nu}.$

 \medskip 
 
\noindent So the `bulk term' is
\begin{eqnarray*}
 \nabla_\mu J^\mu =T^{\mu\nu}\Pi_{\mu \nu} 
 &=&
 (n-2)\dot{a}a^{1-n} \dot{\phi}^2+\frac{n-2}{2} \dot{a} a^{1-n} \partial_\alpha \phi \partial^\alpha \phi 
 \\
 &&+ \dot{a}a^{-1-n} \delta^{ij} \partial_i \phi \partial_j\phi-\frac{n}{2} \dot{a}a^{1-n} \partial_\alpha  \phi \partial^\alpha \phi\\
 &=& (n-1)\dot{a} a^{1-n}\dot{\phi}^2 \geq 0.\phantom{\frac{1}{2}}
 \end{eqnarray*}
For each $R>0$, define the set 
$$
B_0:=\{(t_0,\bbx)\in I\times \mR^n: \langle \bbx,\bbx \rangle_{\scriptscriptstyle \mR^n} \leq R^2\}.
$$
The future domain of dependence of $B_{0}$ is the set
$$
D^+(B_{0}):=\left\{ p\in M\;\Big|\; \begin{array}{ll} \textrm{Every past inextendible causal curve}\\ \textrm{through }p \textrm{ intersects } B_{0}.\end{array}\right\}.
$$ 
(Here by a {\em causal curve}, we mean one whose tangent vector at
each point is a causal vector.  A curve $c:(a,b)\rightarrow M$ which
is smooth and future directed\footnote{That is, $\dot{c}$ is
  future-pointing.} is called {\em past inextendible} if
$\lim\limits_{t\rightarrow a}c(t)$ does not exist.)

\smallskip 

\noindent 
Let $t_1>t_0$. We will now apply the divergence theorem to the region 
$$
\calR:=D^+(B_{0})\;\!\cap\;\! \{(t,\bbx)\in M: t\leq t_1\}.
$$
For preliminaries on the divergence theorem in the context of a
time-oriented Lorentzian manifold, we refer the reader to
\cite[Appendix~B]{Wal}. We have
$$
\int_\calR (\nabla_\mu J^\mu) \epsilon =\int_{\partial \calR} J\intprod \epsilon,
$$
where $\partial\calR$ denotes the boundary of $\calR$, $\epsilon$ is
the volume form on $M$ induced by $g$, and $\intprod$ denotes
contraction in the first index.
%

\begin{figure}[h]
      \includegraphics[width=8.1 cm]{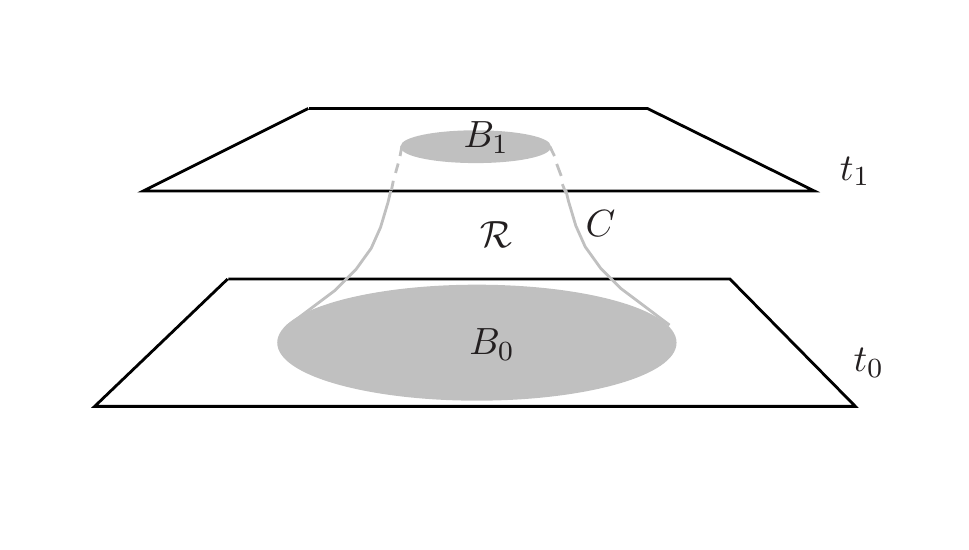}
\end{figure}

\noindent 
Since $J$ is past-pointing, the boundary integral over the null
portion $C$ of the boundary $\partial \calR$ is nonpositive. Also,
because $\nabla_\mu J^\mu$ is nonnegative, we have that the volume
integral over $\calR$ is nonnegative. This gives an inequality on the
two boundary integrals, one over $B_0$, and the other over
$B_1:=D^+(B_0) \cap \{t=t_1\}$, as follows:
$$
\int_{B_0} \frac{1}{2} (a^2 \dot{\phi}^2 +\delta^{ij} \partial_i \phi \partial_j \phi)d^n \bbx 
\geq 
\int_{B_1} \frac{1}{2} (a^2 \dot{\phi}^2 +\delta^{ij} \partial_i \phi \partial_j \phi)d^n \bbx .
$$
Passing the limit $R\rightarrow \infty$ yields 
 $
E(t_0)\geq E(t_1).
$ 
As the choice of $t_1>t_0$ was arbitrary, we have 
$$
\forall t\geq t_0, \;\;E(t)\leq E(t_0)<\infty.
$$
The finiteness of $E(t_0)$ follows from our assumption that
$\phi_0\in H^k(\mR^n)$ and $\phi_1\in H^{k-1}( \mR^n)$ for a $k$
satisfying $k>\frac{n}{2}+2\geq 1$.  From here, it follows that for all
$t \geq t_0$,
\begin{eqnarray*}
&& \int_{\mR^n} \dot{\phi}^2 d^n \bbx \lesssim \frac{1}{a^2},\;\;\textrm{ and}\\
&& \int_{\mR^n} \delta^{ij} \partial_i \phi \partial_j \phi d^n \bbx \lesssim  1.
\end{eqnarray*}
But since each partial derivative $\partial_i \phi$ is also a solution
of the wave equation, and as $k\geq 2$, we obtain, by applying the
above to the partial derivatives $\partial_i \phi$, that also
$$
\int_{\mR^n} (\Delta \phi)^2 d^n \bbx \lesssim 1.
$$
In fact, since $k>\frac{n}{2}+2$, we also obtain that for a
$k'>\frac{n}{2}$,
$$
\|\Delta \phi\|_{H^{k'} (\mR^n)}\lesssim 1.
$$
Finally, by the Sobolev inequality (see e.g. \cite[(7.30),
p.158]{GT}), we obtain
\begin{equation}
\label{30_Aug_9:48}
\|\Delta \phi\|_{L^\infty(\mR^n)}\lesssim 1.
\end{equation}
This completes Step 1 of the proof of Theorem~\ref{theorem_2}.

\medskip 

\noindent {\bf Step 2: The wave equation in conformal coordinates.} 
$\;$

\smallskip 

\noindent The key point of departure from the earlier derivation of the estimates from \cite{CNO} is 
the usage of `conformal coordinates', which renders the wave equation in a form where it becomes possible to 
integrate, leaving essentially just the time derivative of $\phi$ with other terms (e.g. $\Delta\phi$) for which we have a known bound. An 
application of the triangle inequality will then deliver the desired bound.

Define 
$$
\tau=\int_{t_0}^t \frac{1}{a(s)} ds.
$$
Then 
$$
\displaystyle 
\frac{d\tau}{dt}=\frac{1}{a(t)}\;\;\textrm{ and } \;\; a(t) \displaystyle \frac{d}{dt}=\frac{d}{d\tau}.
$$
With a slight abuse of notation, we write $a(\tau):=a(t(\tau))$.
Then $dt=a(\tau) d\tau$. So
$$
g= -dt^2 +(a(t))^2 \left( (dx^1)^2+\cdots+ (dx^n)^2\right)
= (a(\tau))^2 \left(-d\tau^2 +\delta_{ij} dx^i dx^j\right).
$$
The wave equation $\square_g \phi=0$ can be rewritten as
$\partial_\mu(\sqrt{-g} \;\!\partial^\mu \phi)=0$, which becomes
$$
\partial_\mu (a^{n+1} \partial^\mu \phi)=0.
$$
Separating the partial derivative operators with respect to the $\tau$ and $\bbx$ coordinates, we obtain the
wave equation in conformal coordinates
$$
\partial_\tau (a^{n-1} \partial_\tau \phi)=a^{n-1} \Delta \phi,
$$
where $\Delta$ is the usual Laplacian on $\mR^n$. This completes Step
2 of the proof of Theorem~\ref{theorem_2}.

\medskip 

\goodbreak

\noindent {\bf Step 3: $n>2$ and $a(t)=e^t$.}

\smallskip 

\noindent 
We have 
\begin{equation}
 \label{Aug_30_15:18}
\tau=\int_{t_0}^t \frac{1}{e^s} ds=e^{-t_0}-\frac{1}{e^t}=e^{-t_0}-\frac{1}{a},
\end{equation}
and so
$$
a(\tau)=\displaystyle \frac{1}{e^{-t_0}-\tau}.
$$
We note that $\tau \in [0,e^{-t_0})$. 
Also, 
$$
\displaystyle 
a(\tau=0)=\frac{1}{e^{-t_0}}=e^{t_0}=a(t=t_0).
$$ 
Integrating 
$$
\partial_\tau(a^{n-1}\partial_\tau\phi)=a^{n-1}\Delta \phi
$$
from $\tau=0$ to $\tau$, we obtain 
$$
a^{n-1} \partial_\tau \phi -a(t_0)^{n-1} \left. \partial_\tau \phi\right|_{\tau=0}
=
\int_0^\tau \Delta \phi \frac{1}{(e^{-t_0}-\tau)^{n-1} }d\tau,
$$
and so 
$$
a^{n-1} a\partial_t \phi =a(t_0)^{n-1} a(t_0) \left. \partial_t\phi\right|_{t=t_0}+ \int_0^\tau \Delta \phi \frac{1}{(e^{-t_0}-\tau)^{n-1} }d\tau,
$$
that is, 
$$
\partial_t \phi =(a(t))^{-n} \left(a(t_0)^n \phi_1 +\int_0^\tau \Delta \phi \frac{1}{(e^{-t_0}-\tau)^{n-1} }d\tau \right).
$$
Hence, using the bound from \eqref{30_Aug_9:48}, namely $\|\Delta\phi(t,\cdot)\|_{L^\infty(\mR^n)}\leq C$ for all $t\geq t_0$, we obtain 

\vspace{-0.3cm}

{\small 
\begin{eqnarray*}
 &&\|\partial_t \phi(t,\cdot)\|_{L^\infty(\mR^n)} \phantom{\Big(\frac{a(t_0)}{a(t)}\Big)^{n-2}}
 \\
 &\leq &\!\!\!\! (a(t))^{-n}\!\left(a(t_0)^n \|\phi_1\|_{L^\infty(\mR^n)} +\int_0^\tau \|\Delta\phi(t,\cdot)\|_{L^\infty(\mR^n)} 
 \frac{1}{(e^{-t_0}-\tau)^{n-1} }d\tau \right)\\
 &=& \!\!\!\!(a(t))^{-n}\!\left(a(t_0)^n \|\phi_1\|_{L^\infty(\mR^n)} \!+\!
 \frac{C }{n-2} \left( (e^{-t_0}\!-\!\tau)^{2-n}\! -(e^{-t_0})^{2-n}\right) \!\right)\\
 &=& \!\!\!\!(a(t))^{-n}\!\left(a(t_0)^n \|\phi_1\|_{L^\infty(\mR^n)} \!+\!
 \frac{C }{n-2} \left( (a(t))^{n-2}\! -(a(t_0))^{n-2}\right) \!\right)\\
 &\leq& \!\!\!\!(a(t))^{-n}(a(t))^{n-2}\left(\!\frac{a(t_0)^n \|\phi_1\|_{L^\infty(\mR^n)}}{(a(t))^{n-2}} \!+\!
 \frac{C }{n-2} \left(1\!-\!\Big(\frac{a(t_0)}{a(t)}\Big)^{n-2}\right) \!\right)\\
 &\leq & \frac{1}{(a(t))^2} \left( \frac{a(t_0)^n \|\phi_1\|_{L^\infty(\mR^n)}}{(a(t_0))^{n-2}} \!+\!
 \frac{C }{n-2} \left( 1\!-\!0\right) \!\right).
\end{eqnarray*}}

\vspace{-0.45cm}

\noindent 
\noindent Hence
$$
\|\partial_t \phi(t,\cdot)\|_{L^\infty(\mR^n)}\leq \frac{1}{(a(t))^2} \left( (a(t_0))^2 \|\phi_1\|_{L^\infty(\mR^n)}\!+\!
 \frac{C }{n-2}  \!\right),
 $$
 and so 
 $$
 \|\partial_t \phi(t,\cdot)\|_{L^\infty(\mR^n)}\lesssim (a(t))^{-2}.
 $$
 This completes the proof of Theorem~\ref{theorem_2}.
\end{proof}

\medskip 

\begin{remark} {\bf The case when $n=2$ and $a(t)=e^t$:} 

\smallskip 

\noindent Integrating 
$$
\partial_\tau(a\partial_\tau\phi)=a\Delta \phi
$$
from $\tau=0$ to $\tau$, we obtain 
$$
a \partial_\tau \phi -a(t_0) \left. \partial_\tau \phi\right|_{\tau=0}
=
\int_0^\tau \Delta \phi \frac{1}{e^{-t_0}-\tau }d\tau,
$$
and so 
$$
\partial_t \phi =(a(t))^{-2} \left(a(t_0)^2 \phi_1 +\int_0^\tau \Delta \phi \frac{1}{e^{-t_0}-\tau }d\tau \right).
$$
Hence 
\begin{eqnarray*}
 &&\|\partial_t \phi(t,\cdot)\|_{L^\infty(\mR^2)} \phantom{\frac{a(t_0)}{a(t)}\bigg)^{n-2}}
 \\
 &\leq &\!\!\!\! (a(t))^{-2}\!\left(\!a(t_0)^2 \|\phi_1\|_{L^\infty(\mR^2)} +\int_0^\tau \|\Delta\phi(t,\cdot)\|_{L^\infty(\mR^2)} 
 \frac{1}{e^{-t_0}-\tau}d\tau \right)\phantom{\frac{a(t_0)}{a(t)}\bigg)^{n-2}}\\
 &=& \!\!\!\!(a(t))^{-2}\!\left(\!a(t_0)^2 \|\phi_1\|_{L^\infty(\mR^2)} \!+\!
 C\left( -\log (e^{-t_0}\!-\!\tau)\Big|_0^\tau \right) \!\right)\phantom{\frac{a(t_0)}{a(t)}\bigg)^{n-2}}\\
 &=& \!\!\!\!(a(t))^{-2}(\log a(t))\!\bigg(\!\frac{a(t_0)^2 \|\phi_1\|_{L^\infty(\mR^2)}}{\log a(t)}  \!+\!
 C \Big( 1-\frac{\log a(t_0)}{\log a(t)} \Big) \!\bigg)\\
 &\leq &
 \!\!\!\!(a(t))^{-2}(\log a(t))\!\bigg(\!\frac{a(t_0)^2 \|\phi_1\|_{L^\infty(\mR^2)}}{t_0 }  \!+\!
 C \!\bigg),
\end{eqnarray*}
 and so 
 $$
 \|\partial_t \phi(t,\cdot)\|_{L^\infty(\mR^n)}\lesssim (a(t))^{-2}\log a(t).
 $$
 This can be viewed as an improvement to \cite[Theorem~1]{CNO} in the
 special case when $a(t)=e^t$ and $n=2$, since
 $$
 \log a(t)=t \lesssim e^{\delta t}= 1+ \delta t+\cdots .
 $$
\end{remark}

\medskip 

\begin{remark} {\bf The case when $a(t)=t^p$, $p\geq 1$:}
 
 \smallskip 
 
 \noindent One can prove an analogue of Theorem~\ref{theorem_2} when
 $a(t)=t^p$ as well.  In this case, the $\epsilon$ from
 Theorem~\ref{theorem_1} can be chosen to be any number satisfying
  $$
  \epsilon>\frac{1}{p},
  $$
  and so Theorem~\ref{theorem_1} gives the decay estimate
 $$
 \|\partial_t \phi(t,\cdot)\|_{L^\infty(\mR^n)}\lesssim   \left(a(t)\right)^{-2+\frac{1}{p}+\delta}=t^{-(2p-1-\delta')},
 $$
 where $\delta'>0$ can be chosen arbitrarily. We can improve this to the following:
 $$
 \|\partial_t \phi(t,\cdot)\|_{L^\infty(\mR^n)}\lesssim   \left(a(t)\right)^{-2+\frac{1}{p}}= t^{-(2p-1)}.
 $$
 The proof is the same, mutatis mutandis, as that of Theorem~\ref{theorem_2}.
\end{remark}

\begin{remark}
Using a similar method, one can also obtain an improvement to \cite[Theorem~2]{CNO}. 
But we will postpone this discussion until after Section~\ref{Section_RNdS}, since we will 
need some preliminaries about the RNdS spacetime, which will be established in Section~\ref{Section_RNdS}. 
\end{remark}

\section{Decay in the de Sitter universe in flat FLRW form}
\label{Section_FLRW}

\noindent 
The Klein-Gordon equation is $
\square_g\phi -m^2\phi=0$, 
that is, 
$$
\frac{1}{\sqrt{-g}} \partial_\mu (\sqrt{-g}\; \partial^\mu \phi)-m^2\phi=0.
$$
In the case of the de Sitter universe in flat FLRW form, we obtain 
\begin{equation}
 \label{KG_FLRW}
-\ddot{\phi}-\frac{n\dot{a}}{a}\dot\phi+\frac{1}{a^2} \delta^{ij} \partial_i \partial_j \phi-m^2\phi=0.
\end{equation}

\noindent 
In this section, we will prove Theorem~\ref{theorem_3}.  We arrive at
the guesses for the specific estimates given in Theorem~\ref{theorem_3} below,  
based on an analysis using Fourier modes, assuming  spatially periodic solutions. This 
Fourier mode analysis is given in Appendix A.

\begin{theorem}$\;$\label{theorem_3}

\noindent 
Suppose that  
\begin{itemize}
 \item $I\subset \mR$ is an open interval of the form $(t_*,+\infty)$, $t_0\in I$, 
 \item $m\in \mR$, 
 \item $n> 2$, 
 \item $(M,g)$ is the expanding de Sitter universe in flat FLRW form, with flat $n$-dimensional sections, 
  given by $I\times \mR^n$, with the metric 
   $$
   g=-dt^2 +e^{2t} \left( (dx^1)^2+\cdots+(dx^n)^2\right),
   $$
 \item $k>\frac{n}{2}+2$, $\;\;\phi_0\in H^k(\mR^n)$, $\;\;\phi_1\in H^{k-1}(\mR^n)$, and 
 \item $\phi$ is a smooth solution to the Cauchy problem 
 $$
 \left\{ \begin{array}{rcll}
          \square_g \phi -m^2 \phi &=&0,& \quad (t\geq t_0,\;\bbx\in \mR^n), \\
          \phi(t_0,\bbx)&=& \phi_0(\bbx) &\quad (\bbx\in \mR^n),\\
          \partial_t \phi(t_0,\bbx)&=& \phi_1(\bbx) & \quad (\bbx\in \mR^n).
         \end{array}\right.
         $$
\end{itemize}
Then for all $t\geq t_0$, we have 
$$
\| \phi(t,\cdot)\|_{L^\infty(\mR^n)}\lesssim \left\{ \begin{array}{ll} 
a^{-\frac{n}{2}} & \textrm{if }\;\; |m|>  \frac{n}{2},\phantom{a^{\sqrt{\frac{n^2}{4}-m^2}}}\\
a^{-\frac{n}{2}}\log a & \textrm{if }\;\;|m|=\frac{n}{2},\phantom{a^{\sqrt{\frac{n^2}{4}-m^2}}}\\
a^{-\frac{n}{2}+\sqrt{\frac{n^2}{4}-m^2}} &\textrm{if }\;\;|m|<\frac{n}{2}.
\end{array}\right.
$$
\end{theorem}

\smallskip 

\begin{remark} We recall that the conformally invariant wave equation in $n+1$ dimensions 
is 
$$
\left(\square_g -\frac{n-1}{4n} R_g\right)\phi=0,
$$
where $R_g$ is the scalar curvature of the metric $g$; see for instance \cite{Wal}. 
If $g$ is a FLRW metric with flat $n$-dimensional spatial sections, having the form given by \eqref{equation_10_august_2020_16:53}, then 
$$
R_g=\frac{2n\ddot{a}}{a}+\frac{n(n-1)\dot{a}^2}{a^2}.
$$
Thus in de Sitter space in flat FLRW form, the conformally invariant wave equation can be interpreted as a Klein-Gordon equation, with the mass parameter satisfying $m^2 = \frac{n^2-1}{4}$. From \cite[Appendix~B]{CNO}, 
we have 
\begin{equation}
\label{equation_10_august_2020_16:40}
\|\phi(t,\cdot)\|_{L^\infty(\mR^n)}\lesssim a^{\frac{1-n}{2}},
\end{equation}
which follows from  using the fact that the $L^\infty$-norm of $\psi(\cdot, t)$, 
defined by $\phi=a^{1-\frac{n+1}{2}}\psi$ (see \cite[eq. (178)]{CNO}), is uniformly bounded with respect to $t$. 
The estimate \eqref{equation_10_august_2020_16:40} is in complete agreement with the result of our Theorem~\ref{theorem_3} above, since the relation $\frac{n^2}{4} - m^2 = \frac{1}{4}$ reduces 
our bound $a^{-\frac{n}{2}+\sqrt{\frac{n^2}{4}-m^2}}$ precisely to $a^{\frac{1-n}{2}}$. 
\end{remark}

\subsection{Preliminary energy function and estimates}$\;$

\smallskip 

\noindent 
Define the energy-momentum tensor $T$ by 
$$
T_{\mu \nu}=\partial_\mu \phi \partial_\nu \phi -\frac{1}{2} g_{\mu\nu}(\partial_\alpha \phi \partial^\alpha \phi+m^2\phi^2).
$$
Then $\nabla_\mu T^{\mu\nu}=0$. Also, in particular, 
$$
T_{00}=\frac{1}{2}\left(\dot{\phi}^2+\frac{1}{a^2}|\nabla \phi|^2+m^2\phi^2\right)=T^{00}.
$$
Set 
$$
X=a^{-n}\frac{\partial}{\partial t}.
$$
Then $X$ is time-like and hence causal, and $X$ is future pointing. 

\medskip 

\noindent 
Define $J$ by 
$$
J^\mu=T^{\mu \nu}X_\nu.
$$
Then $J$ is causal and past-pointing. 

\medskip 

\noindent 
Let $N=\partial_t$. Define the energy $E$ by 
$$
E(t)=\int_{\{t\}\times \mR^n} J_\mu N^\mu 
=\int_{\mR^n}\frac{1}{2}\left( \dot{\phi}^2 +\frac{1}{a^2} |\nabla \phi|^2+m^2\phi^2\right)d^n\bbx.
$$
Define
$$
\Pi=\frac{1}{2}\calL_X g=-dt\calL_X dt 
+a^{-n+1}\dot{a}\left((dx^1)^2+\cdots+(dx^n)^2\right).
$$
As 
$$
\calL_X dt=-na^{-n-1}\dot{a} dt,
$$
we have 
$$
\Pi=na^{-n-1} \dot{a} dt^2+a^{-n+1}\dot{a} \left((dx^1)^2+\cdots+(dx^n)^2\right).
$$
Hence 
$$
\nabla_\mu J^\mu=T^{\mu \nu}\Pi_{\mu \nu}
=\frac{a^{-n-1}\dot{a}}{2} \left(2n\dot{\phi}^2 +\frac{2}{a^2}|\nabla \phi|^2\right) \geq 0.
$$
For $R>0$, define 
$$
B_0:=\{(t_0,\bbx)\in I\times \mR^n: \langle \bbx,\bbx\rangle_{\scriptscriptstyle \mR^n} \leq R^2\}.
$$
The future domain of dependence of $B_{0}$ is denoted by $D^+(B_0)$.  

\noindent 
Let $t_1>t_0$. We will now apply the divergence theorem to the region 
$$
\calR:=D^+(B_{0})\;\!\cap \;\!\{(t,\bbx)\in M: t\leq t_1\}.
$$
We have 
$$
\int_\calR (\nabla_\mu J^\mu) \epsilon =\int_{\partial \calR} J\intprod \epsilon.
$$
Using 
\begin{itemize}
 \item $\nabla_\mu J^\mu\geq 0$, and 
 \item the fact that the boundary contribution on $C$, the null portion of $\partial \calR$, is nonpositive 
(since $J$ is causal and past-pointing), 
\end{itemize}
we obtain the inequality 
$$
\int_{B_0} \!\frac{1}{2} \!\left( \dot{\phi}^2 \!+\!\frac{1}{a^2}|\nabla \phi|^2\!+\!m^2\phi^2\right)\!d^n \bbx 
\;\!\geq \;\!
\int_{B_1}\! \frac{1}{2}\! \left( \dot{\phi}^2 \!+\!\frac{1}{a^2}|\nabla \phi|^2\!+\!m^2\phi^2\right)\!d^n \bbx .
$$
Passing the limit $R\rightarrow \infty$ yields $E(t_1)\leq E(t_0)<+\infty$. As $t_1>t_0$ was arbitrary, we obtain 
$$
\forall t\geq t_0,\;\;E(t)= 
\int_{\mR^n} \frac{1}{2} \left( \dot{\phi}^2 +\frac{1}{a^2}|\nabla \phi|^2+m^2\phi^2\right)d^n \bbx
\leq E(t_0)\leq \infty.
$$
 From here, it follows that for all $t \geq t_0$, 
\begin{eqnarray*}
&&
\int_{\mR^n} \dot{\phi}^2 d^n \bbx \lesssim  1,\\
&& 
\int_{\mR^n} |\nabla\phi|^2 d^n \bbx \lesssim  a^2,\;\;\textrm{ and }\\
&& 
\int_{\mR^n} \phi^2 d^n \bbx \lesssim  1 \;\;\;\;(\textrm{if } m\neq 0).
\end{eqnarray*}
But since each partial derivative $(\partial_{x^1})^{i_1}\cdots (\partial_{x^n})^{i_n} \phi$ is also a solution of the Klein-Gordon equation,
it follows from $\phi_0 \in H^k(\mR^n)$  and $\phi_1 \in H^{k-1}(R^n)$ for a $k > \frac{n}{2 }+ 2$, that also $\phi(t,\cdot) \in H^k(\mR^n)$  and 
$\partial_t \phi(t,\cdot) \in H^{k-1}(\mR^n)$, and moreover
\begin{eqnarray*}
\|\dot{\phi}\|_{H^{k'}(\mR^n)} &\lesssim & 1,\\
\|\partial_i \phi\|_{H^{k'}(\mR^n)}   &\lesssim & a,\;\;\textrm{ and }\\
\|\phi\|_{H^{k'}(\mR^n)} &\lesssim&  1 \;\;\;\;(\textrm{if } m\neq 0),
\end{eqnarray*}
where $k':=k-1$.

\subsection{The auxiliary function $\psi$ and its PDE}$\;$

\smallskip 

\noindent Motivated by the decay rate we anticipate for $\phi$, we define the auxiliary function $\psi$ by 
$$
\psi:= a^{\kappa}\phi,
$$
where 
$$
\kappa:=\left\{\begin{array}{ll}
               \frac{n}{2} & \textrm{if } \;|m|\geq \frac{n}{2},\\
               \frac{n}{2}-\sqrt{\frac{n^2}{4}-m^2} & \textrm{if }\; |m|\leq \frac{n}{2}.
              \end{array}\right.
$$
Then, using \eqref{KG_FLRW}, it can be shown that $\psi$ satisfies the equation 
\begin{equation}
 \label{KG_FLRW_for_psi}
 \ddot{\psi}+(\kappa^2-n\kappa+m^2)\psi+(n-2\kappa)\dot{\psi}-\frac{1}{a^2} \Delta \psi=0.
\end{equation}

\subsection{The case $|m|>\frac{n}{2}$} $\;$

\smallskip 

\noindent 
We have  $\kappa=\frac{n}{2}$,  so that $n-2\kappa=0$, while 
$$
\kappa^2-n\kappa+m^2=m^2-\frac{n^2}{4},
$$
and thus \eqref{KG_FLRW_for_psi} becomes 
$$
\ddot{\psi}-\frac{1}{a^2} \Delta \psi+\left(m^2-\frac{n^2}{4}\right)\psi=0.
$$
We note that if $\phi\in H^\ell(\mR^n)$ and $\dot{\phi}\in H^{\ell-1}(\mR^n)$ for some $\ell$, then 
$\psi\in H^\ell(\mR^n)$ too, and also
$$
\dot{\psi}=\frac{n}{2}a^{\frac{n}{2}-1}\dot{a}\phi +a^{\frac{n}{2}}\dot{\phi}\in H^{\ell-1}(\mR^n).
$$
Define the new energy $\calE$, associated with the $\psi$-evolution, by 
$$
\calE(t):=\frac{1}{2}\int_{\mR^n} \left( \dot{\psi}^2+\frac{1}{a^2} |\nabla \psi|^2+\Big(m^2-\frac{n^2}{4}\Big)\psi^2 \right) d^n\bbx.
$$
Then  using the fact that $a=e^t=\dot{a}>0$, and also equation \eqref{KG_FLRW_for_psi}, we obtain
{\small 
\begin{eqnarray*}
 \calE'(t)&=& \!\!\!\int_{\mR^n}\left(\dot{\psi}\ddot{ \psi}-\frac{a\dot{a}}{a^4}|\nabla \psi|^2 
 +\frac{1}{a^2}\langle \nabla \psi,\nabla\dot{\psi}\rangle +\Big(m^2-\frac{n^2}{4}\Big)\psi \dot{\psi}\right)d^n \bbx\\
 &\leq & \!\!\! \int_{\mR^n}\left(\dot{\psi}\ddot{ \psi} 
 +\frac{1}{a^2}\langle \nabla \psi,\nabla\dot{\psi}\rangle +\Big(m^2-\frac{n^2}{4}\Big)\psi\dot{\psi}\right)d^n \bbx\\
 &\leq &\!\!\! \int_{\mR^n}\left(\dot{\psi}\Big(\frac{1}{a^2}\Delta\psi -\Big(m^2-\frac{n^2}{4}\Big)\psi\Big)
 +\frac{1}{a^2}\langle \nabla \psi,\nabla\dot{\psi}\rangle +\Big(m^2-\frac{n^2}{4}\Big)\psi \dot{\psi}\right)d^n \bbx\\
 &=&\!\!\! \frac{1}{a^2} \int_{\mR^n}\left(\dot{\psi} \Delta\psi + \langle \nabla \psi,\nabla\dot{\psi}\rangle \right)d^n \bbx= \frac{1}{a^2} \int_{\mR^n} \nabla \cdot (\dot{\psi} \nabla \psi) d^n\bbx.
\end{eqnarray*}}

\noindent 
For a fixed $t$, and for a ball $B(\mathbf{0},r)\subset \mR^n$, where $r>0$, it follows  from the divergence theorem 
(since $\dot{\psi}$ and $\nabla{\psi}$ are smooth),  that 
$$
\int_{B(\mathbf{0},r)}  \nabla \cdot (\dot{\psi} \nabla \psi) \;\!d^n\bbx=\int_{\partial B(\mathbf{0},r)}  \dot{\psi} \;\!\langle \nabla \psi,\bbn\rangle \;\! d\sigma_r,
$$
where $d\sigma_r$ is the surface area measure on the sphere $S_r=\partial B(\mathbf{0},r)$, and $\bbn$ is the outward-pointing unit normal. 
The right hand side surface integral tends to $0$ as $r\rightarrow+\infty$, 
 by an application of Lemma~\ref{technical_lemma_surface_integral}, given in Appendix~B.  

So for $t\geq t_0$, we have $\calE'(t)\leq 0$, which yields $\calE(t)\leq \calE(t_0)$. 
In particular, for all $t\geq t_0$, $\|\psi(t,\cdot)\|_{L^2(\mR^n)}\lesssim C$, that is, 
$\|a^{\frac{n}{2}}\phi(t,\cdot)\|_{L^2(\mR^n)}\lesssim C$, and so\footnote{We note that to reach this conclusion, we used Lemma~\ref{technical_lemma_surface_integral}, 
for which we need $\dot{\psi}(t,\cdot), \nabla \psi(t,\cdot)\in  H^1(\mR^n)$, which means that the initial conditions
for $\phi$ must be such that  $\phi_0\in  H^2(\mR^n)$ and $\phi_1 \in H^1(\mR^n)$.} 
\begin{equation}
\label{15:37_9_september}
\|\phi(t,\cdot)\|_{L^2(\mR^n)} \lesssim a^{-\frac{n}{2}}.
\end{equation}
Then with enough regularity on $\phi_0,\phi_1$ at the outset, that is, if 
$\phi_0\in H^k(\mR^n)$ and $\phi_1\in H^{k-1}(\mR^n)$ for a $k>\frac{n}{2}+2$, 
and by considering $(\partial_{x^1})^{i_1}\cdots (\partial_{x^n})^{i_{n}}\phi$ as a solution to the 
Klein-Gordon equation, we arrive at\footnote{Note that in order to use the estimate \eqref{15:37_9_september}, 
for $D\phi:=(\partial_{x^1})^{i_1}\cdots (\partial_{x^n})^{i_{n}}\phi$ replacing $\phi$, where $|(i_1,\cdots ,i_n)|=:k'$,  
we must ensure that the initial conditions  for $D\phi$, namely $(D \phi(t_0,\cdot),D\dot{\phi}(t_0,\cdot))$ is in $ (H^2(\mR^n), H^1(\mR^n))$,  
which is guaranteed if the initial condition for $\phi$, namely $(\phi_0,\phi_1)$ is in $ (H^k(\mR^n), H^{k-1}(\mR^n))$, 
with $k-k'=2$.}
$$
\|\phi(t,\cdot)\|_{H^{k'}(\mR^n)} \lesssim a^{-\frac{n}{2}},
$$
where $k':=k-2$. As $k'=k-2>\frac{n}{2}$, we have, using the Sobolev inequality, that 
$$
\forall t\geq t_0,\;\;\;\; \|\phi(t,\cdot)\|_{L^\infty(\mR^n)} \lesssim a^{-\frac{n}{2}}.
$$
This completes the proof of Theorem~\ref{theorem_3} in the case when $|m|>\frac{n}{2}$.

\medskip 

\subsection{The case $|m|<\frac{n}{2}$} $\;$

\smallskip 

\noindent 
We have 
$$
\kappa=\displaystyle \frac{n}{2}-\sqrt{\frac{n^2}{4}-m^2},\quad\;\;  n-2\kappa=2\displaystyle \sqrt{\frac{n^2}{4}-m^2}\;>0,
$$
and $\kappa^2-n\kappa+m^2=0$. Equation \eqref{KG_FLRW_for_psi} becomes 
$$
\ddot{\psi} +2\Big(\sqrt{\frac{n^2}{4}-m^2}\Big)\;\dot{\psi}-\frac{1}{a^2} \Delta \psi=0.
$$
Defining 
 $\displaystyle 
\widetilde{\calE}(t):=\frac{1}{2}\int_{\mR^n} \left(\dot{\psi}^2+\frac{1}{a^2}|\nabla \psi|^2 \right) d^n\bbx,
$ 
we obtain 
\begin{small} 
\begin{eqnarray*}
\!\!\!\!\!\!\!\!\!\!\!\! \!\!\!\!\!\widetilde{\calE}'(t)\!\!\!\!\!&=&\!\!\!\!\!
 \int_{\mR^n}\left( \dot{\psi}\ddot{\psi}-\frac{\dot{a}}{a^3} |\nabla \psi|^2 +\frac{1}{a^2} \langle \nabla \psi, \nabla \dot{\psi}\rangle \right)d^n\bbx \phantom{aa}
 \\
 &=& \!\!\!\!
 \int_{\mR^n}\!\!\left( \dot{\psi}\Big(\!-\!2\Big(\sqrt{\frac{n^2}{4}\!-\!m^2}\Big)\;\dot{\psi}\!+\!
 \frac{1}{a^2} \Delta \psi\Big) \!-\!\frac{\dot{a}}{a^3} |\nabla \psi|^2 \!+\!\frac{1}{a^2} \langle \nabla \psi, \!\nabla \dot{\psi}\rangle \!\right)d^n\bbx \phantom{aaa}
 \\
 &=& \!\!\!\!
 \!-2\Big(\sqrt{\frac{n^2}{4}\!-\!m^2} \Big)
 \int_{\mR^n} \dot{\psi}^2 d^n\bbx \!-\!\frac{\dot{a}}{a^3} \int_{\mR^n} |\nabla \psi|^2 d^n\bbx.
\end{eqnarray*}
\end{small}

\noindent 
Using $a=e^t=\dot{a}$, we obtain 
\begin{eqnarray*}
 \widetilde{\calE}'(t)&=&
 -4\Big(\sqrt{\frac{n^2}{4}-m^2}\Big) \;
 \frac{1}{2}\int_{\mR^n} \dot{\psi}^2 d^n\bbx -2\frac{1}{2} \int_{\mR^n} \frac{1}{a^2}|\nabla \psi|^2 d^n\bbx\\
 &\leq& 
 -\min\left\{ 4\Big(\sqrt{\frac{n^2}{4}-m^2}\Big) , \;2\right\}\cdot \frac{1}{2} \int_{\mR^n} \left(\dot{\psi}^2+\frac{1}{a^2} |\nabla \psi|^2 \right) d^n\bbx
 \phantom{aaa}
 \\
 &=&-\;\theta \cdot \widetilde{\calE}(t),
 \phantom{\frac{1}{2} \int_{\mR^n} \left(\frac{1}{a^2} |\nabla \psi|^2 \right)} 
\end{eqnarray*}
\noindent where 
$$
\theta:=\min\left\{ 4\Big(\sqrt{\frac{n^2}{4}-m^2} \Big)\;,\; 2\right\}>0.
$$
So $\widetilde{\calE}'(t)+\theta \cdot \widetilde{\calE}(t)\leq 0$. Multiplying throughout by $e^{\theta t}>0$, we obtain 
$$
\frac{d}{dt}\left(e^{\theta t}\cdot \widetilde{\calE}(t)\right)\leq 0.
$$
Integrating from $t_0$ to $t$ yields 
$$
e^{\theta t}\cdot \widetilde{ \calE}(t)\leq e^{\theta t_0}\cdot \widetilde{ \calE}(t_0),
$$
that is, $\widetilde{\calE}(t)\lesssim  e^{-\theta t}$. In particular, 
$$
\|\dot{\psi}(t,\cdot)\|_{L^2(\mR^n)} \leq \sqrt{2\;\!\widetilde{\calE}(t)} \lesssim e^{-\frac{\theta}{2}t}.
$$
We have 
$\displaystyle 
\psi(t,\bbx)=\psi(t_0,\bbx)+ \int_{t_0}^t (\partial_t \psi)(s,\bbx) ds,
$ 
and so 
 \begin{eqnarray*}
 \|\psi(t,\cdot)\|_{L^2(\mR^n)}
 &\leq & 
 \|\psi(t_0,\cdot)\|_{L^2(\mR^n)} + \int_{t_0}^t \|(\partial_t \psi)(s,\cdot)\|_{L^2(\mR^n)} ds,
 \\
 &\lesssim & 
 A+\int_{t_0}^t Be^{-\frac{\theta}{2}s} ds 
 = A+B\frac{e^{-\frac{\theta}{2}t_0}-e^{-\frac{\theta}{2}t}}{\theta/2} \lesssim C.
 \end{eqnarray*}
 Thus for all $t\geq t_0$, we have 
 $$
 \|\phi(t,\cdot)\|_{L^2(\mR^n)} =a^{-\kappa}\|\psi(t,\cdot)\|_{L^2(\mR^n)}\lesssim a^{-\kappa}.
 $$
 By considering $(\partial_{x^1})^{i_1}\cdots (\partial_{x^n})^{i_n}\phi$ and using the Sobolev inequality, we have
 $$
\forall t\geq t_0,\;\; \;\; \|\phi(t,\cdot)\|_{L^\infty(\mR^n)} \lesssim a^{-\kappa}=a^{-(\frac{n}{2}-\sqrt{\frac{n^2}{4}-m^2}\;\!)}.
$$
This completes the proof of  Theorem~\ref{theorem_3} in the case when $|m|<\frac{n}{2}$.

\subsection{The case $|m|=\frac{n}{2}$}  $\;$

\smallskip 

\noindent We have $\kappa= \frac{n}{2}$, and equation \eqref{KG_FLRW_for_psi} becomes 
$
\displaystyle \ddot{\psi} -\frac{1}{a^2} \Delta \psi=0.
$ 

\smallskip 

\noindent 
Defining the same energy as we used earlier in the case when $|m|<\frac{n}{2}$, 
$$
\widetilde{\calE}(t):=\frac{1}{2}\int_{\mR^n} \left(\dot{\psi}^2+\frac{1}{a^2}|\nabla \psi|^2 \right) d^n\bbx,
$$
we obtain  
\begin{eqnarray*}
 \widetilde{\calE}'(t)\!\!\!&=&\!\!\!\!
 \int_{\mR^n}\left( \dot{\psi}\ddot{\psi}-\frac{\dot{a}}{a^3} |\nabla \psi|^2 +\frac{1}{a^2} \langle \nabla \psi, \nabla \dot{\psi}\rangle \right)d^n\bbx 
 \\
 &=& \!\!\!\!
 \int_{\mR^n}\!\!\left( \dot{\psi}
 \frac{1}{a^2} \Delta \psi \!-\!\frac{\dot{a}}{a^3} |\nabla \psi|^2 \!+\!\frac{1}{a^2} \langle \nabla \psi, \!\nabla \dot{\psi}\rangle \!\right)d^n\bbx 
 \!=\!-\frac{\dot{a}}{a^3} \int_{\mR^n}\!\! |\nabla \psi|^2 d^n\bbx\leq 0.
\end{eqnarray*}
So $\widetilde{\calE}(t)\leq \widetilde{\calE}(t_0)$ for $t\geq t_0$. In particular, $\|\dot{\psi}(t,\cdot)\|_{L^2(\mR^n)}\lesssim B$ for $t\geq t_0$. 
Again, 
$$
\displaystyle 
\psi(t,\bbx)=\psi(t_0,\bbx)+ \int_{t_0}^t (\partial_t \psi)(s,\bbx) ds,
$$ 
gives  
 \begin{eqnarray*}
 \|\psi(t,\cdot)\|_{L^2(\mR^n)}
 &\leq & 
 \|\psi(t_0,\cdot)\|_{L^2(\mR^n)} + \int_{t_0}^t \|(\partial_t \psi)(s,\cdot)\|_{L^2(\mR^n)} ds,
 \\
 &\lesssim & 
 A'+\int_{t_0}^t B ds 
  \lesssim A+Bt\lesssim \log a.
 \end{eqnarray*}
 Thus for all $t\geq t_0$, we have 
  $
 \|\phi(t,\cdot)\|_{L^2(\mR^n)} =a^{-\kappa}\|\psi(t,\cdot)\|_{L^2(\mR^n)}\lesssim a^{-\kappa} \log a.
 $ 
 Hence (by considering $(\partial_{x^1})^{i_1}\cdots (\partial_{x^n})^{i_n}\phi$ and using the Sobolev inequality)  
 \begin{equation}
  \label{30_Aug_10:34} 
\forall t\geq t_0,\;\;  \|\phi(t,\cdot)\|_{L^\infty(\mR^n)} \lesssim a^{-\frac{n}{2}}\log a.
 \end{equation}
(One can show that this bound is sharp; see Appendix C.) 

\noindent 
This completes the proof of Theorem~\ref{theorem_3}.

\section{Decay in the cosmological region of the $\mathrm{RNdS}$ spacetime}
\label{Section_RNdS}

\noindent 
The Reissner-Nordstr\"om-de Sitter (RNdS) spacetime $(M,g)$ is a solution to the Einstein-Maxwell 
equations with a positive cosmological constant, and it represents a pair\footnote{We note that 
there is no solution analogous to RNdS but with only one black hole. This is analogous to 
(but much more complicated than, and still not fully understood)  the fact that one cannot have a single 
electric charge on a spherical universe (Gauss's law requires that the total charge must be zero). In fact, the 
fundamental solution of the Laplace equation on the sphere gives a unit positive charge at some point 
and a unit negative charge at the antipodal point. One can have more than two black holes, 
for instance the so-called Kastor-Traschen solution \cite{KT}.} of antipodal charged 
black holes in a spherical\footnote{``Spherical'' here means that the Cauchy hypersurface (that is, ``space'') is an $n$-sphere.} universe which is undergoing accelerated expansion. 
The Reissner-Nordstr\"om-de Sitter metric in $n+1$ dimensions is given by 
$$
g=-\frac{1}{V} dr^2 +V dt^2+r^2 d\Omega^2,
$$
where 
$$
V= \displaystyle r^2+\frac{2M}{r^{n-2}}-\frac{e^2}{r^{n-1}}-1,
$$ and 
 $d\Omega^2$ is the unit round metric on $S^{n-1}$. The constants $M$ and $e$ are proportional to the mass and the charge, respectively, of the 
 black holes, and the cosmological constant is chosen to be 
 $$
 \Lambda=\frac{n(n-1)}{2}
 $$
 by an appropriate choice of units. 
 
 Consider the polynomial 
 $$
 p(r):=r^{n-1} V(r)=r^{n+1}-r^{n-1}+2Mr -e^2.
 $$
 As $p(0)=-e^2<0$ and as $p(r)\stackrel{r\rightarrow\infty}{\longrightarrow} \infty $, it follows that  
 $p$ will have a real root in $(0,+\infty)$, and the largest real root of $p$, which we denote by $r_c$, must be positive. 
 If $r>r_c$, then clearly $p(r)>0$, and so also $V(r)>0$.
 
 It can also be seen that $p$ has at most three distinct positive roots. Suppose, on the contrary, that $p$ has 
 more than three distinct positive roots: $r_1<r_2<r_3<r_4$. Applying Rolle's theorem to 
 $p$ on $[r_i,r_{i+1}]$ ($i=1,2,3$), we conclude that $p'$ must have three distinct roots $r_i'\in (r_i,r_{i+1})$ ($i=1,2,3$). 
  Applying Rolle's theorem to $p'$ on $[r_i',r_{i+1}']$ ($i=1,2$), we conclude that $p''$ must have two distinct roots $r_i''\in (r_i',r_{i+1}')$ ($i=1,2$). 
  But 
  $$
  p''=r^{n-3}n(n+1) \Big(r^2-\frac{(n-1)(n-2)}{n(n+1)}\Big),
  $$
  which has only one positive root, a contradiction. 
  
  The `subextremality' assumption on the RNdS spacetime made in Theorem~\ref{theorem_3}, refers to a nondegeneracy of the 
  positive roots of $p$: we assume that there are exactly three positive roots, $r_-, r_+$ and $r_c$, and  
  $$
  0<r_-<r_+<r_c.
  $$
  These describe the event horizon $r=r_+$, $\phantom{l}$and the Cauchy `inner' horizon  $\;r=r_-$.  
 It can be seen that the subextremality condition then implies $p'(r_c)>0$. (Indeed, $p'(r_c)$ cannot be negative, as otherwise $p$ would acquire a root larger than $r_c$ since 
 $p(r)\stackrel{r\rightarrow \infty}{\longrightarrow }\infty$.
 Also, if $p'(r_c)=0$, then Rolle's theorem implies again that $p'$ would have three positive roots, ones in $(r_-,r_+)$ and $(r_+,r_c)$, 
 and one at $r_c$, which is impossible, as we had seen above.) $p'(r_c)>0$ implies that  $V'(r_c)>0$. 
 We will also assume that 
 $$
 V''(r_c)>0.
 $$
 Our assumptions have the following consequence, which will be used in our proof of Theorem~\ref{theorem_4}.

 \begin{lemma}[Global redshift] 
  $V'(r)>0$ for all  $r\geq r_c$.
 \end{lemma}
 \begin{proof} We have 
 $$
 V'(r)=\frac{rp'(r)-(n-1)p(r)}{r^n}=\frac{2r^{n+1}+2(3-n) Mr +(n-1)e^2}{r^n}=:\frac{q(r)}{r^n}.
  $$
  As $V'(r_c)>0$, we have $q(r_c)>0$. Also $V''(r_c)>0$ and so $V'$ is increasing near $r_c$. But 
  then $q(r)=r^n V'(r)$ is also increasing near $r_c$, and in particular, $q'(r_c)\geq 0$. 
  Let us suppose that there exists an $r_*>r_c$ such that $V'(r_*)=0$, and let $r_*$ be the smallest such root.  Then $q(r_*)=0$ too. 
  We note that 
  $$
  q'=2(n+1)r^n+2(3-n)M,
  $$
  and so $q'$ can have only one nonnegative root, namely 
   $
  \left(\frac{(n-3)}{n+1}M\right)^{\frac{1}{n}}\geq 0. 
  $ 
  \begin{itemize} 
   \item[$1^\circ$] $r_*$ is a repeated root of $q$. Then $q'(r_*)=0$.
   
    If in addition $q'(r_c)=0$, then we arrive at a contradiction, since 
    $q'$ then has two positive roots (at $r_c$ and at $r_*$), which is impossible.  
   
    If $q'(r_c)>0$, then we arrive at a contradiction as follows. As $q$ is increasing near $r_c$, and 
    since $q(r_c)>0=q(r_*)$, it follows by the intermediate value theorem that there is some $r_{c}' \in (r_c,r_*)$ 
    such that $q(r_{c}')=q(r_c)$. But by Rolle's theorem applied to $q$ on $[r_c,r_{c}']$, there must exist 
    an $r_*'\in (r_c, r_{c}')$ such that $q'(r_*')=0$. Again $q'$ acquires two zeros (at $r_*$ and at $r_*'$), which is 
    impossible. 
    
    \item[$2^\circ$] $r_*$ is a simple root of $q$. But as $q(r)\stackrel{r\rightarrow\infty}{\longrightarrow }\infty$, 
    it follows that there must be at least one more root $r_{**}>r_{*}$ of $q$. By Rolle's theorem applied to $q$ on $[r_*, r_{**}]$, it follows that  $q'(r_{**}')=0$ 
    for some $r_{**}' \in (r_*,r_{**})$.

    If in addition $q'(r_c)=0$, then we arrive at a contradiction, since 
    $q'$ then has two positive roots (at $r_c$ and at $r_{**}'$), which is impossible.  
    
    If $q'(r_c)>0$, then, as in the last paragraph of $1^\circ$ above, there exists an $r_{*}'\in (r_c,r_c')\subset (r_c,r_*)$ 
    such that $q'(r_*')=0$. Thus $q'$ again gets two positive roots (at $r_*'$ and at $r_{**}'$), which is impossible.
  \end{itemize}
This shows that our assumption the $V'$ is zero beyond $r_c$ is incorrect. 
 \end{proof}

 \noindent 
 The hypersurfaces of constant $r$ are spacelike cylinders 
 with a future-pointing unit normal vector field 
 $$
 N=V^{\frac{1}{2}} \frac{\partial }{\partial r},
 $$
 and volume element 
 $$
 dV_n=V^{\frac{1}{2}} r^{n-1} dt d\Omega.
 $$
 
 \medskip 
 
 \noindent 
 The global structure of a maximal spherically symmetric extension of this metric 
 can be depicted by a conformal Penrose diagram shown below, repeated periodically; see for example  \cite{CNO0}. 
 
 \medskip

\begin{figure}[h]
     \center 
      \includegraphics[width=12.6 cm]{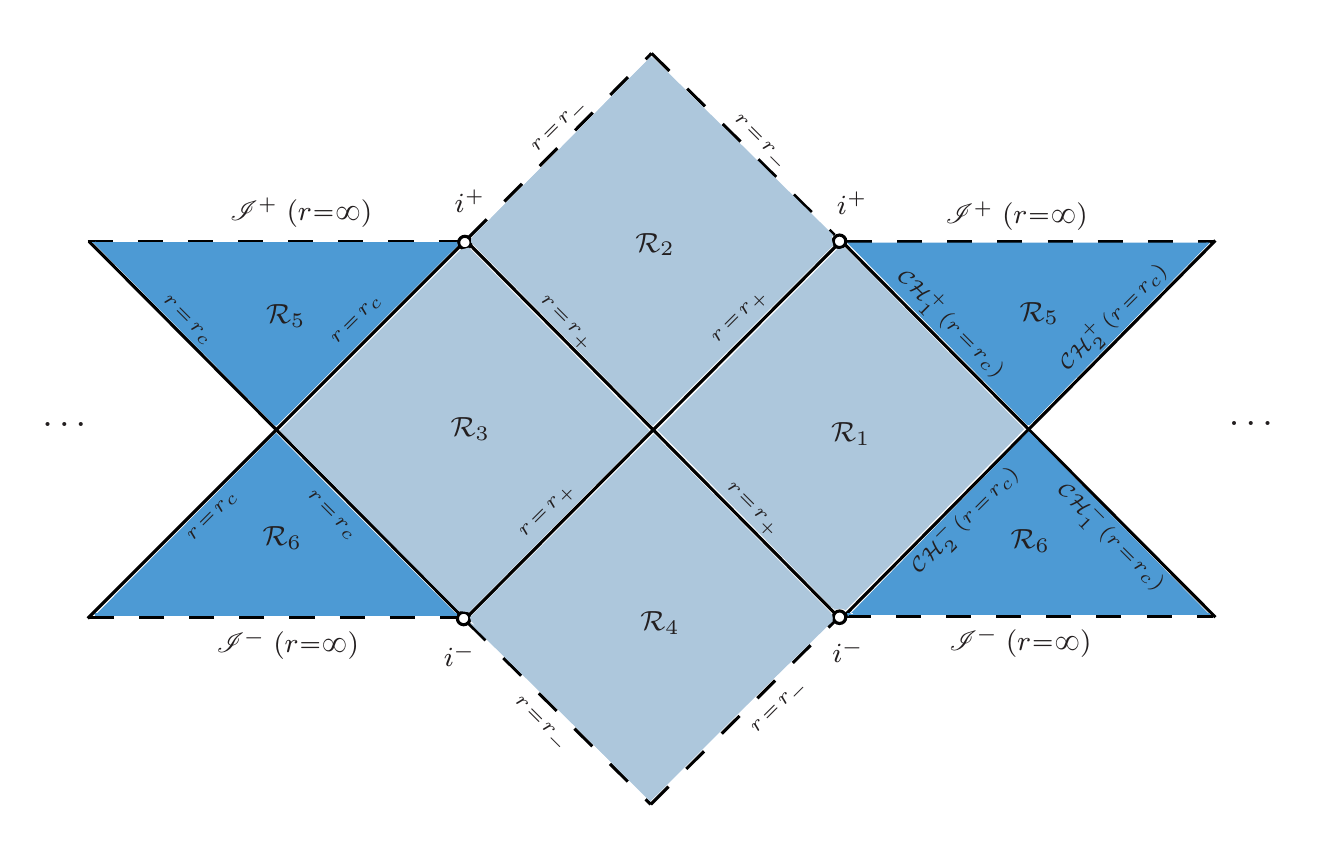}
     \caption{Conformal diagram of the Reissner-Nordstr\"om-de Sitter spacetime.}
      \label{31_Aug_15:57}
 \end{figure}

 \medskip
 
 \noindent 
 We are interested in the behaviour of the solution to the  Klein-Gordon equation in the cosmological region $\calR_5$ of this spacetime 
 (see Figure~\ref{31_Aug_15:57}), 
 bounded by the cosmological horizon branches $\calC\calH_1^+$, $\calC\calH_2^+$, the future null infinity ${\mathscr{I}}^+$, and 
 the point $i^+$. In particular, we want to obtain estimates for the decay rate of $\phi$ as $r\rightarrow \infty$. 
 We guess the decay rates simply by substituting $r$ instead of $e^t$ in the estimates we had obtained for 
 the decay rate of $\phi$ with respect to $t$ in the case of the de Sitter universe in flat FLRW form from the previous Section~\ref{Section_FLRW}. 
 
 We will prove the following result. 
 
 \goodbreak
 
\begin{theorem}
$\;$\label{theorem_4}

\noindent 
Suppose that 
\begin{itemize}
 \item $\epsilon>0$, 
 \item $m\in \mR$, 
 \item $M>0$, 
 \item $e> 0$,
 \item $n> 2$, 
 \item $(M,g)$ is the $(n+1)$-dimensional subextremal Reissner-Nordstr\"om-de Sitter solution given by the metric 
   $$
   g=-\frac{1}{V} dr^2+Vdt^2 +r^2 d\Omega^2,
   $$
   where 
   $$
   V= r^2+\frac{2M}{r^{n-2}} -\frac{e^2}{r^{n-1}}-1,
   $$
   and $d\Omega^2$ is the metric of the unit $(n-1)$-dimensional sphere $S^{n-1}$, 
 \item $k>\frac{n}{2}+2$, and 
 \item $\phi$ is a smooth solution to $\square_g \phi-m^2 \phi=0$ such that 
$$
\|\phi\|_{H^{k}(\calC\calH_1^+)}<+\infty \;\;\;\;\textrm{ and }\;\;\;\;\|\phi\|_{H^{k}(\calC\calH_2^+)}<+\infty,
$$
where $\calC\calH_1^+\simeq \calC\calH_2^+\simeq \mR\times S^{n-1}$ are the two components of the 
future cosmological horizon, parameterised by the flow parameter of the global Killing vector field $\frac{\partial}{\partial t}$. 
\end{itemize}
Then there exists a $r_0$ large enough so that for all $r\geq r_0$, 
$$
\| \phi(r,\cdot)\|_{L^\infty(\mR\times S^{n-1})}\lesssim 
\left\{ \begin{array}{ll} 
r^{-\frac{n}{2}+\epsilon} & \textrm{if }\;\; |m|> \frac{n}{2},\phantom{a^{\sqrt{\frac{n^2}{4}-m^2}}}\\
r^{-\frac{n}{2}+\sqrt{\frac{n^2}{4}-m^2}\;\!+\epsilon} &\textrm{if }\;\;|m|\leq\frac{n}{2}.
\end{array}\right.
$$
 \end{theorem}

 \smallskip 
 
 \subsection{Preliminary energy function} $\;$
 
 \smallskip 
 
 \noindent 
 For a $\phi$ defined in the cosmological region $\calR_5$, we define 
 $$
 \phi':=\frac{\partial \phi }{\partial r}\;\;\textrm{ and }\;\;
 \dot{\phi}:=\frac{\partial \phi}{\partial t}.
 $$
 We will also use the following notation:
 $$
 \begin{array}{ccl}
 \cnab \phi && \textrm{gradient of }\phi\textrm{ on }S^{n-1}\textrm{ with respect to the unit round metric},\\
 |\cnab\phi|&& \textrm{norm with respect to the unit round metric},\\
 \cLap\phi && \textrm{Laplacian of }\phi\textrm{ on }S^{n-1}\textrm{ with respect to the unit round metric},\\
 \cg && \textrm{determinant of the unit round metric}.\phantom{S^{n-1}\cnab \phi }
 \end{array}
 $$
 Suppose that $\phi$ satisfies the Klein-Gordon equation $\square_g\phi-m^2 \phi=0$. 
 Recall that the energy-momentum tensor associated with $\phi$ is given by 
 $$
 T_{\mu \nu}=\partial_\mu \phi \partial_\nu \phi -\frac{1}{2} g_{\mu\nu} (\partial_\alpha \partial^\alpha \phi+m^2\phi^2).
 $$
 Thus 
 \begin{eqnarray*}
  T(N,N)
  &=&
  \left( \phi'^2-\frac{1}{2}\frac{(-1)}{V} \Big(\phi'^2(-V)+\dot{\phi}^2 \frac{1}{V}+\frac{1}{r^2} |\cnab\phi|^2 +m^2\phi^2\Big)\right)V\\
  &=& \frac{1}{2}\left( V\phi'^2 +\frac{1}{V}\dot{\phi}^2 +\frac{1}{r^2} |\cnab\phi|^2 +m^2\phi^2\right).
 \end{eqnarray*}
  Define 
 $$
 X:=\frac{V^{\frac{1}{2}}}{r^{n-1}} N=\frac{V^{\frac{1}{2}}}{r^{n-1}} V^{\frac{1}{2}}\frac{\partial}{\partial r}=\frac{V}{r^{n-1}}\frac{\partial}{\partial r}.
 $$
 We define the energy 
 \begin{eqnarray*}
 E(r)&:=&\int_{\mR\times S^{n-1}} T(X,N) dV_n\phantom{\frac{1}{2} \int_{\mR\times S^{n-1}} }\\
 &=&
 \frac{1}{2} \int_{\mR\times S^{n-1}}  \left( V^2 \phi'^2 +\dot{\phi}^2 +\frac{V}{r^2} |\cnab \phi|^2+m^2V\phi^2\right)dt d\Omega.
 \end{eqnarray*}

\subsection{The auxiliary function $\psi$ and its PDE}$\;$

\smallskip 

\noindent 
The Klein-Gordon equation $\square_g\phi-m^2\phi=0$ can be rewritten as:
\begin{eqnarray*}
 && \frac{1}{\sqrt{-g}}\;\!\partial_\mu (\sqrt{-g} \;\partial^\mu \phi)-m^2 \phi=0,
 \\
 &\Leftrightarrow& \frac{1}{r^{n-1}\sqrt{\cg}}\;\! \partial_\mu\left( r^{n-1}\sqrt{\cg} \;g^{\mu\nu} \partial_\nu \phi\right)-m^2\phi=0.
 \end{eqnarray*}
This becomes 
$$
\frac{1}{r^{n-1}\!\sqrt{\cg}} \;\!\bigg( \partial_r \Big(r^{n-1} \!\sqrt{\cg} \; (-V) \partial_r \phi\Big) 
+\partial_t \Big( r^{n-1} \!\sqrt{\cg} \;\frac{1}{V} \partial_t \phi\Big) 
 +\frac{r^{n-1} \!\sqrt{\cg}}{r^2} \;\cLap \phi \bigg) -m^2\phi=0
 $$
 that is, 
 \begin{eqnarray*}
 && 
 -(V\phi')'-\frac{(n-1)}{r} V \phi'+\frac{\ddot{\phi}}{V}+\frac{1}{r^2} \cLap \phi -m^2 \phi=0,\\
 &\Leftrightarrow&
 \phi'' +\frac{(n-1)}{r}\phi' +\frac{V'}{V} \phi' -\frac{\ddot{\phi}}{V^2} -\frac{1}{r^2 V} \cLap \phi +\frac{m^2}{V}\phi=0.
\end{eqnarray*}
Define 
$$
\psi:=r^{\kappa}\phi,
$$
where 
$$
\kappa=\left\{\begin{array}{cl} \displaystyle
               \frac{n}{2}  &\textrm{ if } |m|\geq \displaystyle \frac{n}{2},\\ [.3cm]
               \displaystyle\frac{n}{2}-\sqrt{\frac{n^2}{4}-m^2}& \textrm{ if }|m|\leq  \displaystyle \frac{n}{2}.
              \end{array}\right.
$$
Then, using the PDE for $\phi$, it can be shown that 
\begin{equation}
 \label{RNdS_KG_aux}
\psi''+\left(\frac{V'}{V}+\frac{n-1}{r}-\frac{2\kappa}{r}\right)\psi'-\frac{\ddot{\psi}}{V^2}
-\frac{1}{r^2 V} \cLap \psi +\theta \psi=0,
\end{equation}
where 
$$
\theta:=\frac{m^2}{V}+\frac{\kappa}{r}\left(\frac{1}{r}-\frac{V'}{V}\right)-\frac{\kappa}{r^2} (n-1-\kappa).
$$

\smallskip 

\subsection{The case $|m|>\frac{n}{2}$}$\;$

\smallskip 

\noindent 
Then $\kappa=\frac{n}{2}$, and \eqref{RNdS_KG_aux} becomes 
\begin{equation}
 \label{RNdS_KG_aux_m_big}
 \psi''+\left(\frac{V'}{V}-\frac{1}{r}\right)\psi'-\frac{\ddot{\psi}}{V^2}
-\frac{1}{r^2 V} \cLap \psi +\theta \psi=0,
\end{equation}
where 
$$
\theta:=-\frac{n}{2}\left(\frac{n}{2}-1\right)\frac{1}{r^2} +\frac{m^2}{V}+\frac{n}{2r}\left( \frac{1}{r}-\frac{V'}{V}\right).
$$
We will use an energy function to obtain the required decay of $\psi$ for large $r$, 
and in order to do so, we will need to keep careful track of the limiting behaviour of the 
various functions appearing in the expression for $\theta$ and the coefficients of the PDE \eqref{RNdS_KG_aux_m_big}. 
We will do this step-by-step in a sequence of lemmas. 

\begin{lemma}
Given any $\epsilon>0$, there exists an $r_0$ large enough so that for all $r\geq r_0$, 
$$
\frac{2+\epsilon}{r}\geq \frac{V'}{V}\geq \frac{2-\epsilon}{r}.
$$
\end{lemma}
\begin{proof}
 This follows immediately from
 $$
  \lim_{r\rightarrow \infty} r\frac{V'}{V}
  \!=\!
  \lim_{r\rightarrow \infty}r\cdot 
  \frac{2r-\frac{2(n-2)M}{r^{n-1}}+\frac{e^2(n-1)}{r^n}}{r^2+\frac{2M}{r^{n-2}}-\frac{e^2}{r^{n-1}}-1}
  \!=\!\lim_{r\rightarrow \infty}\frac{2-\frac{2(n-2)M}{r^{n}}+\frac{e^2(n-1)}{r^{n+1}}}{1+\frac{2M}{r^{n}}-\frac{e^2}{r^{n+1}}-\frac{1}{r^2}}
  \!=\!2.
  $$
  \end{proof}
  
\begin{lemma}
There exists $r_0$ large enough so that for $r\geq r_0$, we have $\theta>0$.   
\end{lemma}

\noindent 
(We note that the proof uses the fact that $|m|>\frac{n}{2}$, and so this result is specific to this subsection.) 

\begin{proof}
We have 
$$
\lim_{r\rightarrow \infty}\frac{r^2}{V}=\lim_{r\rightarrow \infty}\frac{1}{1+\frac{2M}{r^n}-\frac{e^2}{r^{n+1}}-\frac{1}{r^2}}=1,
$$
and so there exists a $r_0'$ such that 
$$
\displaystyle 
\frac{r^2}{V}\geq 1-\epsilon
$$
for $r\geq r_0'$. Also, by the previous lemma, there exists a $r_0>r_0'$ such that 
$$
\frac{V'}{V}\leq \frac{2+\epsilon}{r} 
$$
for all $r\geq r_0$. Then we have for $r>r_0$ that 
\begin{eqnarray*}
 \theta &=& \frac{1}{r^2}\left(
 -\frac{n}{2}\left(\frac{n}{2}-1\right) +m^2\frac{r^2}{V}+\frac{n}{2}\Big(1-\frac{V'}{V}r\Big)\right)\\
 &\geq & \frac{1}{r^2}\left(
 -\frac{n^2}{4}+\frac{n}{2} +m^2(1-\epsilon)+\frac{n}{2}\Big(1-\frac{(2+\epsilon)}{r}r\Big)\right)
 \\
 &=& \frac{1}{r^2}\left(\delta-\epsilon\Big(\delta+\frac{n}{2}+\frac{n^2}{4}\Big)\right),
\end{eqnarray*}
where 
 $
\delta:=m^2-\frac{n^2}{4}>0.
$ 

\medskip 

\noindent 
Taking $\epsilon$ at the outset small enough so as to satisfy 
 $
0<\epsilon<\frac{\delta}{\delta+\frac{n}{2}+\frac{n^2}{4}},
$ 
we see that $\theta>0$ for $r\geq r_0$.
\end{proof}

\smallskip 

\noindent 
Define the energy 
$$
\calE(r):=\frac{1}{2}\int_{\mR\times S^{n-1}} \left(\psi'^2+\frac{1}{V^2}\dot{\psi}^2+\frac{1}{r^2 V} |\cnab \psi|^2+\theta \psi^2\right) dt d\Omega.
$$

\medskip 

\noindent 
(We assume for the moment that this is finite for a sufficiently large $r_0$. 
Later on, in the subsection on redshift estimates, we will see how our initial finiteness of Sobolev norms 
of $\phi$ on the two branches $\calC \calH_1^+$, $\calC \calH_2^+$ of the cosmological horizon  guarantees this.)

\medskip

\noindent We now proceed to find an expression for $\calE'(r)$, and to  simplify it, we will use \eqref{RNdS_KG_aux_m_big}, and the divergence theorem, 
to get rid of the terms involving $\ddot{\psi}$ and $\cLap \psi$, the spherical Laplacian of $\psi$: 
{\small 
\begin{eqnarray*}
\calE'(r)\!\!\!\!\!&=&\!\!\! \!\! \int_{\mR\times S^{n-1}} \!\! \left( \psi' \psi'' +\frac{1}{2} \Big(\frac{1}{V^2}\Big)' \dot{\psi}^2 +
\frac{1}{V^2} \dot{\psi}\dot{\psi}'+\frac{1}{2} \Big( \frac{1}{r^2 V}\Big)'|\cnab \psi|^2\right.\\
&&\left. \phantom{aaaaaaaa} +\frac{1}{r^2V}\langle \cnab \psi,(\cnab \psi)'\rangle 
+\frac{\theta'}{2} \psi^2 +\theta \psi \psi' \right) dt d\Omega\\
&=&\!\!\!\!\! \int_{\mR\times S^{n-1}}\!\!  \left( \psi' \bigg(-\Big( \frac{V'}{V}-\frac{1}{r} \Big)\psi' +\frac{1}{V^2}\ddot{\psi}+\frac{1}{r^2V}\cLap \psi-\cancel{\theta \psi}\bigg) \right. \\
&&\phantom{aaaaaaaa} +\frac{1}{2} \Big(\frac{1}{V^2}\Big)' \dot{\psi}^2 +
\frac{1}{V^2} \dot{\psi}\dot{\psi}'+\frac{1}{2} \Big( \frac{1}{r^2 V}\Big)'|\cnab \psi|^2\\
&&\left. \phantom{aaaaaaaa} +\frac{1}{r^2V}\langle \cnab \psi,(\cnab \psi)'\rangle 
+\frac{\theta'}{2} \psi^2 +\cancel{\theta \psi \psi'} \right) dt d\Omega\\
&=&\!\!\!\!\! \int_{\mR\times S^{n-1}}\!\!  \left( -\Big( \frac{V'}{V}\!-\!\frac{1}{r} \Big)\psi'^2 +\cancel{\frac{1}{V^2}\ddot{\psi}\psi'}
+\cancel{  \frac{1}{V^2} \dot{\psi}\dot{\psi}' }+\!
\frac{1}{r^2V}\left( \bcancel{\psi' \cLap \psi}+\bcancel{ \langle \cnab \psi, (\cnab \psi)'\rangle} \right)  \right. \\
&&\phantom{aaaaaaaa} \left. +\frac{1}{2} \Big(\frac{1}{V^2}\Big)' \dot{\psi}^2 +\frac{1}{2} \Big( \frac{1}{r^2 V}\Big)'|\cnab \psi|^2
+\frac{\theta'}{2} \psi^2  \right) dt d\Omega.
\end{eqnarray*}
}
We note that in the above,  getting rid of the spherical Laplacian by using the divergence theorem is allowed because the compact sphere $S^{n-1}$ has no boundary. 
For the second time derivative, however, there is a boundary at infinity (with two connected components), namely 
$$
\lim_{t\rightarrow +\infty}\int_{S^{n-1}} \dot{\psi} \psi' d\Omega-\lim_{t\rightarrow -\infty} \int_{S^{n-1}} \dot{\psi} \psi' d\Omega,
$$
which can be seen to be equal to $0$, by  Lemma~\ref{cylinder_technical_lemma_surface_integral} 
from Appendix B.
\noindent $\!\!$Thus
{\small $$
\calE'(r)=
\int_{\mR\times S^{n-1}} 
\left(-\Big(\frac{V'}{V}-\frac{1}{r}\Big)\psi'^2+\frac{1}{2}\Big(\frac{1}{V^2}\Big)'\dot{\psi}^2+
\frac{1}{2}\Big(\frac{1}{r^2V}\Big)' |\cnab \psi|^2 +\frac{\theta'}{2}\psi^2
\right) dt d\Omega.
$$}

\goodbreak 

\noindent 
Let $\epsilon>0$ be given. Then there exists an $r_0$ large enough such that:

\medskip 

\noindent (a) $\displaystyle \frac{V'}{V}-\frac{1}{r}\geq \frac{2-\epsilon}{r}-\frac{1}{r}=\frac{1-\epsilon}{r}$,

\medskip 

\noindent (b) $\displaystyle \Big(\frac{1}{V^2}\Big)'=-2\frac{V'}{V}\frac{1}{V^2}\leq -2\frac{(2-\epsilon)}{r}\frac{1}{V^2}$, 

\medskip 

\noindent (c) $\displaystyle \Big(\frac{1}{r^2V}\Big)'=-\frac{1}{r^2V}\Big(\frac{2}{r}+\frac{V'}{V}\Big)\leq -\frac{1}{r^2V}\Big(\frac{2}{r}+\frac{2-\epsilon}{r}\Big)=-\frac{1}{r^2 V}
\frac{(4-\epsilon)}{r}$,

\medskip 

\noindent (d) $\displaystyle \frac{\theta'}{\theta}\!
=\!\frac{1}{r} \Big( \frac{ -\frac{n}{2}(\frac{n}{2}-1)(-2)\!-\frac{m^2V'}{r^3 V^2} \!-\frac{n}{2r^5}(\frac{1}{r}-\frac{V'}{V}) +\frac{n}{2r^4}(-\frac{1}{r^2}-\frac{V''V-V'^2}{V^2}) }{
 -\frac{n}{2}(\frac{n}{2}-1)+\frac{m^2}{r^2 V} +\frac{n}{2r^3}(\frac{1}{r}-\frac{V'}{V})  } \Big)$
 
 $\phantom{aaa} \displaystyle \leq \frac{1}{r}(-2+\epsilon).$

 \medskip 
 
\noindent 
Hence, using (a)-(d) above, we obtain 
{\small 
\begin{eqnarray*}
 \calE'(r)\!\!\!
 &=&\!\!\!
\int_{\mR\times S^{n-1}}\!\!
\left(-\Big(\frac{V'}{V}-\frac{1}{r}\Big)\psi'^2+\frac{1}{2}\Big(\frac{1}{V^2}\Big)'\dot{\psi}^2+
\frac{1}{2}\Big(\frac{1}{r^2V}\Big)' |\cnab \psi|^2 +\frac{\theta'}{2}\psi^2
\!\right) dt d\Omega\\
&\leq &\!\!\!
\int_{\mR\times S^{n-1}} 
\left(-\frac{(1-\epsilon)}{r}\psi'^2+\frac{1}{2}\frac{(-2)(2-\epsilon)}{rV^2}\dot{\psi}^2+
\frac{1}{2}\frac{(-1)}{r^2V}\frac{(4-\epsilon)}{r}|\cnab \psi|^2 \right.\\
& & \left.\phantom{aaaaaaa}+\!\frac{1}{2}\frac{1}{r}(-2+\epsilon)\theta \psi^2
\!\right) dt d\Omega\\
&=&\!\!\!-\frac{1}{r}\frac{1}{2}
\int_{\mR\times S^{n-1}} 
\left(2(1-\epsilon)\psi'^2 + 2(2-\epsilon)\frac{1}{V^2}\dot{\psi}^2 + 
(4-\epsilon)\frac{1}{r^2V}|\cnab \psi|^2 \right.\\
&& \left. \phantom{\frac{1}{V^2}aaaaaaala}+(2-\epsilon)\theta \psi^2
\right) dt d\Omega \\
&\leq&\!\!\!  -\frac{2(1-\epsilon)}{r}\calE(r).
\end{eqnarray*}}

\noindent 
Using Gr\"onwall's inequality (see e.g. \cite[Appendix~B(j)]{Eva}), we obtain 
$$
\calE(r)\leq \;\calE(r_0) e^{\int_{r_0}^r -\frac{2(1-\epsilon)}{r} dr }\; =\calE(r_0)\Big(\frac{r}{r_0}\Big)^{-2(1-\epsilon)}
\lesssim r^{-2+2\epsilon}.
$$
Thus 
$$
\int_{\mR\times S^{n-1}} \theta \psi^2 dt d\Omega \;\leq\;  2\calE(r) \;\lesssim \; r^{-2+2\epsilon}, 
$$
and so 
$$
\int_{\mR\times S^{n-1}}  \psi^2 dt d\Omega 
\; \lesssim \;
\frac{r^{2\epsilon}}{r^2 \theta} 
\; \lesssim \; 
\frac{r^{2\epsilon}}{r^2\frac{1}{r^2}} \; = \;  r^{2\epsilon}. 
$$
Hence 
$$
\|\psi(r,\cdot)\|_{L^2(\mR\times S^{n-1})} \lesssim r^\epsilon.
$$
Consequently, 
$$
\|\phi(r,\cdot)\|_{L^2(\mR\times S^{n-1})} \lesssim r^{-\frac{n}{2}+\epsilon}.
$$
Recall that $S^{n-1}$ admits $\frac{n(n-1)}{2}$ independent Killing vectors, given by 
$$
L_{ij}=x^i \frac{\partial}{\partial x^j}-x^j \frac{\partial}{\partial x^i},
$$
for $i<j$ (under the usual embedding $S^{n-1}\subset \mR^n$). As 
$\frac{\partial}{\partial t}$ and $L_{ij}$ are Killing vector fields, it follows that 
$\dot{\phi}$ and $L_{ij}\cdot \phi$ are also solutions to $\square_g \phi-m^2 \phi=0$. 
Commuting with the Killing vector fields $\frac{\partial}{\partial t}$ and $L_{ij}$,  
 if we assume at the moment\footnote{This will be proved later in the subsection on redshift estimates.} that at $r_0$ we have 
 $
\|\phi(r_0,\cdot)\|_{H^k(\{r=r_0\})} <+\infty,
$ 
then we also obtain for all $r\geq r_0$ that 
 $
\|\phi(r,\cdot)\|_{H^{k'}(\mR\times S^{n-1})} \lesssim r^{-\frac{n}{2}+\epsilon},
$ 
where $k'=k-2>\frac{n}{2}$. By the Sobolev inequality\footnote{ 
The part of the Sobolev embedding theorem concerning inclusion in H\"older spaces  
holds for a complete Riemannian manifold with a positive injectivity radius and a bounded sectional curvature; see e.g. \cite[\S3.3, Thm.3.4]{Heb} or \cite[Ch.2]{Aub}.},  
 $
\|\phi(r,\cdot)\|_{L^\infty(\mR\times S^{n-1})} \lesssim r^{-\frac{n}{2}+\epsilon}.
$

This completes the proof of Theorem~\ref{theorem_4} in the case when $|m|>\frac{n}{2}$ 
(provided we show the aforementioned finiteness of energy, which will be carried out in Subsection~\ref{Subsection_redshift_estimates} on redshift estimates). 

\smallskip 

\subsection{The case $|m|\leq \frac{n}{2}$}$\;$

\smallskip 

\noindent 
Let $\epsilon'>0$ be given. Define 
$$
\widetilde{\calE}(r)=\frac{1}{2} \int_{\mR \times S^{n-1}} 
\left( \psi'^2 +\frac{1}{V^2} \dot{\psi}^2+\frac{1}{r^2 V} |\cnab \psi|^2 +\frac{\epsilon'}{r^2} \psi^2\right) dt d\Omega.
$$
We now proceed to find an expression for $\widetilde{\calE}'(r)$, and we will simplify it using  \eqref{RNdS_KG_aux} and the divergence theorem, 
in order to get  rid of the terms involving $\ddot{\psi}$ and the spherical Laplacian of $\psi$: $\phantom{\ddot{\psi}}$
{\small 
\begin{eqnarray*}
\!\!\!\!\!\!\!
&&\!\!\!\!\!\!\!\!\!\!\!\!\widetilde{\calE}'(r)\phantom{\int_{\mR\times S^{n-1}} \!\! \boldsymbol{\left(} \frac{1}{2} \Big(\frac{1}{V^2}\Big)'\right) }
\\
\!\!\!\!\!&=&\!\!\! \!\! \int_{\mR\times S^{n-1}} \!\! \left( \psi' \psi'' +\frac{1}{2} \Big(\frac{1}{V^2}\Big)' \dot{\psi}^2 +
\frac{1}{V^2} \dot{\psi}\dot{\psi}'+\frac{1}{2} \Big( \frac{1}{r^2 V}\Big)'|\cnab \psi|^2\right.\\
&&\left. \phantom{aaaaaaaa} +\frac{1}{r^2V}\langle \cnab \psi,(\cnab \psi)'\rangle 
-\frac{\epsilon'}{r^3} \psi^2 +\frac{\epsilon'}{r^2} \psi \psi' \boldsymbol{\right)} dt d\Omega\\
\!\!\!\!\!&=&\!\!\!\!\! \int_{\mR\times S^{n-1}}\!\!  \left( \psi' \bigg(-\Big( \frac{V'}{V}+\frac{n-1}{r}-\frac{2\kappa}{r} \Big)\psi' +\cancel{\frac{\ddot{\psi}}{V^2}}+\bcancel{
\frac{1}{r^2V}\cLap \psi}
-\theta \psi \bigg) \right.\\
&& \phantom{aaaaaaaa} +\frac{1}{2} \Big(\frac{1}{V^2}\Big)' \dot{\psi}^2 +
\cancel{\frac{1}{V^2} \dot{\psi}\dot{\psi}'}+\frac{1}{2} \Big( \frac{1}{r^2 V}\Big)'|\cnab \psi|^2 
+\bcancel{\frac{1}{r^2V}\langle \cnab \psi,(\cnab \psi)'\rangle} \\
&& \phantom{aaaaaa}\left.\phantom{\Bigg(} 
-\frac{\epsilon'}{r^3} \psi^2 +\frac{\epsilon'}{r^2} \psi \psi' \right) dt d\Omega\\
\!\!\!\!\!&=&\!\!\!\!\! \int_{\mR\times S^{n-1}}\!\!  \left( \!-\Big( \frac{V'}{V}\!+\!\frac{n-1}{r}\!-\!\frac{2\kappa}{r} \Big)\psi'^2 
\!+\!\frac{1}{2} \Big(\frac{1}{V^2}\!\Big)' \dot{\psi}^2 \!+\! \frac{1}{2} \Big( \frac{1}{r^2 V}\!\Big)'|\cnab \psi|^2 \!-\!\frac{\epsilon'}{r^3} \psi^2\!\right) dt d\Omega \\
&& +\left( \frac{\epsilon'}{r^2} -\theta\right) \int_{\mR\times S^{n-1}}\!\! \psi \psi'  dt d\Omega.
\end{eqnarray*}
}
Again, for getting rid of the spherical Laplacian, we use the divergence theorem, noting that the sphere $S^{n-1}$ has no boundary. 
For handling the second time derivative, as before, we note that there is a boundary at infinity (with two connected components), 
which can be seen to be equal to $0$, by  Lemma~\ref{cylinder_technical_lemma_surface_integral} 
from Appendix B.

\medskip 

\noindent 
Thus  
{\small 
\begin{eqnarray*}
 \widetilde{\calE}'(r)\!\!\!\!&\!=\!&\!\!\!\!
\int_{\mR\times\! S^{n-1}}\!\! \!\left(\! -\Big( \frac{V'}{V}\!+\!\frac{n\!-\!1}{r}\!-\!\frac{2\kappa}{r}\Big)\psi'^2 
\!+\!\frac{1}{2} \Big(\frac{1}{V^2}\!\Big)' \dot{\psi}^2 \!+\!\frac{1}{2}\Big(\frac{1}{r^2 V}\!\Big)'|\cnab \psi|^2
\!\!-\!\frac{\epsilon'}{r^3} \psi^2\!\right) dt d\Omega\\
&&+\left( \frac{\epsilon'}{r^2}-\theta\right) \int_{\mR\times S^{n-1}} \psi \psi' dt d\Omega. 
\end{eqnarray*}}
\noindent Now there exists an $r_0$ large enough such that for  all $r\geq r_0$,  we have:
\begin{itemize}

\smallskip 

 \item[(i)] $\displaystyle \frac{V'}{V}+ \frac{n-1}{r}-\frac{2\kappa}{r}\geq \frac{2-\epsilon'}{r}+\frac{n-1}{r}-\frac{2\kappa}{r}=\frac{1-\epsilon'+(n-2\kappa)}{r}
 \geq \frac{1-\epsilon'}{r}$,
 
 \smallskip 
 
 \noindent using $n-2\kappa\geq 0$.
 
 \medskip 
 
 \item[(ii)] $\displaystyle \Big(\frac{1}{V^2}\Big)'\leq -\frac{2(2-\epsilon')}{r}\cdot \frac{1}{V^2}$.
 
 \medskip 
 
 \item[(iii)] $ \displaystyle \Big( \frac{1}{r^2 V}\Big)'\leq -\frac{1}{r^2 V} \Big( \frac{2}{r}+ \frac{2-\epsilon'}{r}\Big)$.
 
 \medskip 
\end{itemize}
Using (i), (ii) and (iii), it can be seen that 
{\small 
\begin{eqnarray*}
 \widetilde{\calE}'(r)\!\!\!\!&\!\leq\!&\!\!\!\!
 \int_{\mR\times S^{n-1}}\!\! \!\left(\! -\frac{(1\!-\!\epsilon')}{r}\psi'^2 
\!-\!\frac{1}{2} \frac{2(2\!-\!\epsilon')}{rV^2}\dot{\psi}^2 
\!-\!\frac{1}{2}\frac{1}{r^2 V}\Big(\frac{2}{r}\!+\!\frac{2\!-\!\epsilon'}{r}\Big) |\cnab \psi|^2 
\!-\!\frac{\epsilon'}{r^3} \psi^2\!\right) dt d\Omega
\\
&&+\left( \frac{\epsilon'}{r^2}-\theta\right) \int_{\mR\times \!S^{n-1}} \psi \psi' dt d\Omega\\
\!\!\!\!&\!\leq\!&\!\!\!\!
-\frac{1}{r}\frac{1}{2} \int_{\mR\times S^{n-1}}\!\! \!\left(\! 2(1\!-\!\epsilon')\psi'^2 
\!+\!2(2\!-\!\epsilon')\frac{1}{V^2}\dot{\psi}^2 
\!+\!(4-\epsilon')\frac{1}{r^2 V} |\cnab \psi|^2 
\!+\!\frac{2\epsilon'}{r^2} \psi^2\!\right) dt d\Omega\\
&&+\left( \frac{\epsilon'}{r^2}-\theta\right)\int_{\mR\times S^{n-1}} \psi \psi' dt d\Omega.
\end{eqnarray*}}

\noindent 
Hence 
$$
\widetilde{\calE}'(r)\leq -\frac{2(1-\epsilon')}{r}\widetilde{\calE}(r)+\Big( \frac{\epsilon'}{r^2}-\theta\Big)\int_{\mR\times S^{n-1}} \psi \psi' dt d\Omega.
$$
We have 
\begin{eqnarray*}
\theta&=&\frac{m^2}{V}+\frac{\kappa}{r}\Big(\frac{1}{r}-\frac{V'}{V}\Big)-\frac{\kappa}{r^2}(n-1-\kappa)\\
&=&\frac{1}{r^2} \left( \frac{m^2}{\frac{V}{r^2}}+\kappa\Big(1-\frac{V'}{V}r\Big) -\kappa(n-1-\kappa)\right).
\end{eqnarray*}
As $\displaystyle \frac{V}{r^2}\stackrel{r\rightarrow\infty}{\longrightarrow} 1$ and 
$\displaystyle \frac{V'}{V}r\stackrel{r\rightarrow\infty}{\longrightarrow} 2$, it follows that 
$$
r^2\theta \;\stackrel{r\rightarrow\infty}{\longrightarrow}\; \frac{m^2}{1}+\kappa(1-2)-\kappa(n-1-\kappa)=m^2-\kappa n+\kappa^2=0.
$$
Thus, given $\epsilon'>0$, there exists an $r_0$ large enough such that for $r\geq r_0$, 
$|r^2\theta|<\epsilon'$, that is, 
$$
|\theta|<\frac{\epsilon'}{r^2}.
$$
So 
\begin{eqnarray*}
 \widetilde{\calE}'(r)&\leq & -\frac{2(1-\epsilon')}{r}\widetilde{\calE}(r)+\Big( \frac{\epsilon'}{r^2}-\theta\Big)\int_{\mR\times S^{n-1}} \psi \psi' dt d\Omega\\
 &\leq & -\frac{2(1-\epsilon')}{r}\widetilde{\calE}(r)+\Big( \frac{\epsilon'}{r^2}+\frac{\epsilon'}{r^2}\Big)\left|\int_{\mR\times S^{n-1}} \psi \psi' dt d\Omega\right|.
\end{eqnarray*}
The Cauchy-Schwarz inequality applied to the last integral gives 
\begin{eqnarray*}
\left|\int_{\mR\times S^{n-1}} \psi \psi' dt d\Omega\right|
&\leq & 
\sqrt{\int_{\mR\times S^{n-1}}  \psi^2 dt d\Omega}\;\cdot \; \sqrt{\int_{\mR\times S^{n-1}} \psi'^2 dt d\Omega}\\
&\leq & \sqrt{\frac{2r^2}{\epsilon'}\widetilde{\calE}(r)} \;\cdot \;\sqrt{2\;\!\widetilde{\calE}(r)}=\frac{2r}{\sqrt{\epsilon'}} \widetilde{\calE}(r).
\end{eqnarray*}
So we obtain 
$$
\widetilde{\calE}'(r)\leq -\frac{2(1-\epsilon')}{r}\widetilde{\calE}(r)+\frac{2\epsilon'}{r^2}\frac{2r}{\sqrt{\epsilon'}} \widetilde{\calE}(r)
=(-2+2\epsilon'+4\sqrt{\epsilon'})\frac{1}{r} \widetilde{\calE}(r).
$$
Application of Gr\"onwall's inequality yields 
$$
\widetilde{\calE}(r)\leq \widetilde{\calE}(r_0)e^{\int_{r_0}^r (-2+2\epsilon'+4\sqrt{\epsilon'})\frac{1}{r}dr} =
\frac{\widetilde{\calE}(r_0)}{r_0^{-2+2\epsilon'+4\sqrt{\epsilon'}}} r^{-2+2\epsilon'+4\sqrt{\epsilon'}}.
$$
So 
\begin{eqnarray*}
\int_{\mR\times S^{n-1}} \psi^2 dt d\Omega &=& \frac{2r^2}{\epsilon'}\frac{1}{2} \int_{\mR\times S^{n-1}} \frac{\epsilon'}{r^2} \psi^2 dt d\Omega \leq  \frac{2r^2}{\epsilon'} \widetilde{\calE}(r)\\
&\leq & \frac{2r^2}{\epsilon'}\frac{\widetilde{\calE}(r_0)}{r_0^{-2+2\epsilon'+4\sqrt{\epsilon'}}} r^{-2+2\epsilon'+4\sqrt{\epsilon'}}
=  \frac{2\widetilde{\calE}(r_0)}{\epsilon' r_0^{-2+2\epsilon'+4\sqrt{\epsilon'}}} r^{2\epsilon'+4\sqrt{\epsilon'}}.
\end{eqnarray*}
Thus 
$$
\|\psi(r,\cdot)\|_{L^2(\mR\times S^{n-1})} \leq  \sqrt{\frac{2\widetilde{\calE}(r_0)}{\epsilon'}} \frac{1}{r_0^{-1+\epsilon'+2\sqrt{\epsilon'}}} r^{\epsilon'+2\sqrt{\epsilon'}},
$$
and so 
$$
\|\phi(r,\cdot)\|_{L^2(\mR\times S^{n-1})} \leq \sqrt{\frac{2\widetilde{\calE}(r_0)}{\epsilon'}} \frac{1}{r_0^{-1+\epsilon'+2\sqrt{\epsilon'}}} r^{-\kappa +\epsilon'+2\sqrt{\epsilon'}}.
$$
Given $\epsilon>0$, arbitrarily small, we can choose $\epsilon'=\epsilon'(\epsilon)>0$ small enough so that $\epsilon'+2\sqrt{\epsilon'}<\epsilon$ at the outset, so that 
$$
\|\phi(r,\cdot)\|_{L^2(\mR\times S^{n-1})} \lesssim r^{-\kappa+\epsilon}.
$$
Again assuming at the moment that at $r_0$ we have 
$$
\|\phi(r_0,\cdot)\|_{H^k(\{r=r_0\})} <+\infty,
$$
and by commuting with the Killing vector fields $\frac{\partial}{\partial t}$ and $L_{ij}$, then we also obtain for all $r\geq r_0$ that 
$$
\|\phi(r,\cdot)\|_{H^{k'}(\mR\times S^{n-1})} \lesssim r^{-(\frac{n}{2}-\sqrt{\frac{n^2}{4}-m^2}\;\!)+\epsilon},
$$
where $k'=k-2>\frac{n}{2}$. By the Sobolev inequality, this yields 
$$
\|\phi(r,\cdot)\|_{L^\infty(\mR\times S^{n-1})} \lesssim r^{-(\frac{n}{2}-\sqrt{\frac{n^2}{4}-m^2}\;\!)+\epsilon}.
$$
This completes the proof of Theorem~\ref{theorem_4} in the case when $|m|\leq \frac{n}{2}$ 
(provided we show the finiteness of energy, which will be carried out in the subsection on redshift estimates below). 

\medskip 

\subsection{Redshift estimates}$\;$\label{Subsection_redshift_estimates}

\smallskip 

\noindent 
The last step is to use redshift estimates to transfer finiteness of the energies along the branches $\calC\calH_1^+$ and $\calC\calH_2^+$ of the cosmological horizon  
to finiteness at $r=r_0$, justifying the finiteness of the energies $\calE(r_0)$ and $\widetilde{\calE}(r_0)$ assumed in the previous two subsections.

\smallskip 

\noindent 
Define the new coordinate $u$ by 
$$
u=t+\int_{r_*}^r \frac{1}{V} dr,
$$
where $r_*>r_c$ is arbitrary, but fixed. Then 
$$
du=dt+\frac{1}{V}dr.
$$
The Reissner-Nordstr\"om-de Sitter metric can be rewritten using the coordinates $(u,r,\cdots)$, instead of the old $(t,r,\cdots)$-coordinates, as follows 
\begin{eqnarray*}
 g&=& -\frac{1}{V} dr^2+Vdt^2+r^2 d\Omega^2=V\Big(-\frac{1}{V^2} dr^2+dt^2\Big) +r^2d\Omega^2\\
 &=&-V\Big(\frac{1}{V}dr +dt\Big) \Big(\frac{1}{V} dr -dt\Big)+r^2d\Omega^2 \\
 &=&-Vdu\Big(-du +\frac{2}{V}dr\Big)+r^2 d\Omega^2=Vdu^2-2du dr +r^2d\Omega^2.
\end{eqnarray*}
The matrix of the metric in the $(u,r,\cdots)$-coordinate system is 
$$
[g_{\mu\nu}]=\left[\begin{array}{rr|c} V& -1 & \\-1& 0 &\\\hline & & \ast \end{array}\right].
$$
Since 
$$
\det \left[\begin{array}{rr} V & -1 \\ -1 & 0 \end{array}\right]=-1,
$$
this coordinate system extends across the cosmological horizon $r=r_c$ (where $V=0$). The hypersurfaces of 
constant $u$ are null and transverse to the cosmological horizon. Thus only one of the branches of the cosmological 
horizon, namely $\calC\calH_1^+$,  is covered by the $(u,r,\cdots)$-coordinates. (In order to cover the other branch $\calC\calH_2^+$, 
where $u=-\infty$, we can introduce 
$$
v:=-t+\int_{r_*}^r\frac{1}{V}dr,
$$
and use the $(v,r,\cdots)$-coordinate chart.)

%

 \begin{figure}[h]
     \center 
      \includegraphics[width=9 cm]{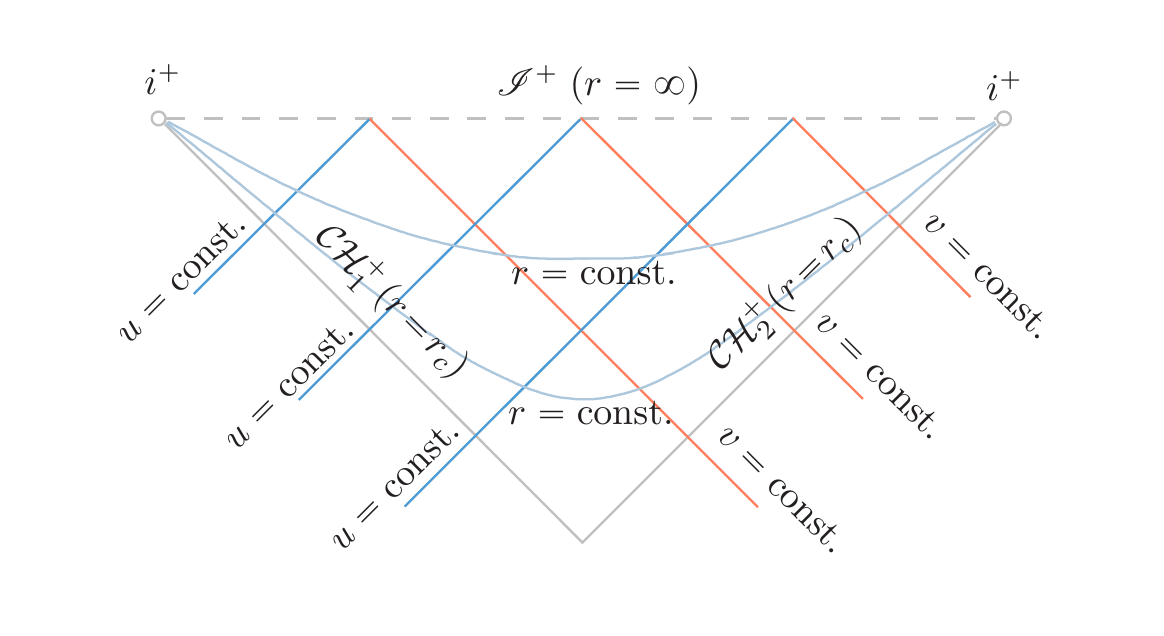}
 \end{figure}

\noindent 
We will only consider $\calC\calH_1^+$ in the remainder of this subsection, since $\calC\calH_2^+$ can be treated in an analogous manner. 

\medskip 

\noindent 
The Killing vector field 
$$
K=\frac{\partial}{\partial u}=\frac{\partial}{\partial t}
$$
is well-defined across $\calC\calH_1^+$, and is null on the cosmological horizon $\calC\calH_1^+$, 
even though the $t$-coordinate is not defined there. Consider the vector field in the $(u,r,\cdots)$-coordinate chart given by 
$$
Y=\left( \frac{\partial}{\partial r}\right)_{u}.
$$
The subscript $u$ means that the integral curves of $Y$ in the $(u,r,\cdots)$-coordinate chart have a constant $u$-coordinate. 
Then we have $du(Y)=0$ and $dr(Y)=1$, and so in the old $(t,r,\cdots)$-coordinate chart, the vector field $Y$ can be expressed as 
$$
Y=\frac{\partial }{\partial r}-\frac{1}{V}\frac{\partial}{\partial t}.
$$
Let the vector field $X$ be defined by 
$$
X=\frac{V^{\frac{1}{2}}}{r^{n-1}}N=\frac{V}{r^{n-1}}\frac{\partial}{\partial r}
$$
in the old $(t,r,\cdots)$-coordinate chart. To find the expression for $X$ in the 
$(u,r,\cdots)$-coordinate chart induced basis vectors, we first find 
$$
N=-\frac{\textrm{grad }r}{|\textrm{grad }r|}
$$
in the $(u,r,\cdots)$-coordinate chart induced basis vectors. Since 
$$
\left[\begin{array}{rr} V & -1 \\ -1 & 0 \end{array}\right]^{-1}
=
\frac{1}{-1} \left[\begin{array}{rr} 0 & 1 \\ 1 & V \end{array}\right]
=
\left[\begin{array}{rr} 0 & -1 \\ -1 & -V \end{array}\right],
$$
we have 
$$
\langle \textrm{grad }r ,\textrm{grad } r \rangle=\langle dr ,dr\rangle=-V.
$$
If $\omega:=-\frac{1}{\sqrt{V}} dr$, then
$$
N=g^{\mu \nu } \omega_\nu= \frac{1}{\sqrt{V}} \Big( \frac{\partial}{\partial u}+V\frac{\partial}{\partial r}\Big).
$$
So 
$$
X=\frac{\sqrt{V}}{r^{n-1}} N=\frac{1}{r^{n-1}}\Big(\frac{\partial}{\partial u}+V\frac{\partial}{\partial r}\Big).
$$
The energy 
\begin{eqnarray*}
 E(r)&=& \int_{\mR\times S^{n-1}} T(X,N) dV_n
\\
&=&
\frac{1}{2}\int_{\mR\times S^{n-1}} \Big(V^2\phi'^2 +\dot{\phi}^2+\frac{V}{r^2} |\cnab\phi|^2 +m^2 V\phi^2\Big) dt d\Omega
\\
&\stackrel{r\rightarrow r_c}{\longrightarrow}& 
\frac{1}{2}\int\displaylimits_{\small\substack{ \mR\times S^{n-1}\\(\calC\calH_1^+)}} (K\cdot \phi )^2 du d\Omega 
+
\frac{1}{2}\int\displaylimits_{\small \substack{\mR\times S^{n-1} \\ (\calC\calH_2^+)}} (K\cdot \phi )^2 dv d\Omega  
\end{eqnarray*}
(since $V(r_c)=0$). So $E(r)$ `loses control' of the transverse and angular derivatives as $r\rightarrow r_c$. To remedy this problem, 
we define a new energy $\widetilde{E}$, by adding $Y$ to $X$, obtaining 
$$
\widetilde{E}(r):=E(r)+\int_{\mR\times S^{n-1}} T(Y,N) dV_n.
$$
In the old $(t,r,\cdots)$-coordinates, $N=\sqrt{V}\;\!\partial_r$, and so 
\begin{eqnarray*}
 T(Y,N)&=& T(\partial_r,N)-\frac{1}{V}T(\partial_t,N)=\frac{1}{\sqrt{V}}T(N,N)-\frac{1}{\sqrt{V}}T(\partial_t,\partial_r)
 \\
 &=&\frac{1}{\sqrt{V}}\left(T(N,N)-T(\partial_t,\partial_r)\right).
\end{eqnarray*}
We have 
$$
T(\partial_t,\partial_r)=\dot{\phi}\;\!\phi'.
$$
So 
{\small 
\begin{eqnarray*}
 \widetilde{E}(r)\!\!\!\!&\!=\!&\!\!\! E(r)\!+\!\int_{\mR\times \!S^{n-1}}\!\! \!\left(\!\frac{1}{\sqrt{V}}\frac{1}{2}\Big(V\phi'^2\!+\!\frac{\dot{\phi}^2}{V}\!+\!\frac{|\cnab \phi|^2}{r^2}
 \!+\!m^2\phi^2\Big)
\! -\!\frac{1}{\sqrt{V}}\dot{\phi}\phi' \!\right)\sqrt{V}r^{n-1} dt d\Omega\\
 &\!=\!&\!\!\!
 E(r)\!+\!\int_{\mR\times S^{n-1}} \frac{1}{2}\left(V \Big(\phi'-\frac{1}{V}\dot{\phi}\Big)^2\!+\!\frac{1}{r^2} |\cnab \phi|^2
 \!+\!m^2\phi^2\right)r^{n-1} dt d\Omega\\
 &\!=\!&\!\!\!
 E(r)\!+\!\int_{\mR\times S^{n-1}} \frac{1}{2}\left( V (Y\cdot \phi)^2\!+\!\frac{1}{r^2} |\cnab \phi|^2
 \!+\!m^2\phi^2\right)r^{n-1} dt d\Omega.
\end{eqnarray*}}
We now have 
\begin{eqnarray*}
\widetilde{E}(r_c)&=& E(r_c)
+\frac{r_c^{n-3}}{2} \int\displaylimits_{\small\substack{ \mR\times S^{n-1}\\(\calC\calH_1^+)}} |\cnab\phi|^2 du d\Omega 
+\frac{r_c^{n-3}}{2} \int\displaylimits_{\small\substack{ \mR\times S^{n-1}\\(\calC\calH_2^+)}} |\cnab\phi|^2 dv d\Omega \\
&&\phantom{E(r_c)} 
+\frac{m^2r_c^{n-1}}{2} \!\!\!\!\!\int\displaylimits_{\small\substack{ \mR\times S^{n-1}\\(\calC\calH_1^+)}} \phi^2 du d\Omega 
+\frac{m^2r_c^{n-1}}{2} \!\!\!\!\int\displaylimits_{\small\substack{ \mR\times S^{n-1}\\(\calC\calH_2^+)}} \phi^2 dv d\Omega,
 \end{eqnarray*}
so that using $\widetilde{E}$ instead of $E$ allows us to regain some control of the angular derivatives as $r\rightarrow r_c$. 
We note that $\widetilde{E}(r_c)$ is equivalent to 
$$
\|\phi\|^2_{H^1(\calC\calH_1^+)} +\|\phi\|^2_{H^1(\calC\calH_2^+)}.
$$
We will now compute the deformation tensor $\Xi$ corresponding to the multiplier $Y$. We have 
\begin{itemize}
\item $\displaystyle \Big[-\frac{1}{V} \partial_t,\partial_r\Big]=-\frac{V'}{V} \partial_t$,
\item $\displaystyle \calL_{\partial_r} g=\frac{V'}{V} dr^2 -\frac{2}{V} dr\calL_{\partial_r} dr +V'dt^2+2rd\Omega^2
= \frac{V'}{V} dr^2  +V'dt^2+2rd\Omega^2$,
\item $\displaystyle \calL_{-\frac{1}{V}\partial_t} g=-\frac{2}{V} dr \calL_{-\frac{1}{V}\partial_t} dr +2V dt\calL_{-\frac{1}{V}\partial_t} dt 
=2V dt \;\Big(- \frac{V'}{V^2} \Big) dr$.
\end{itemize}
Hence 
\begin{eqnarray*}
 \Xi &=& \frac{1}{2} \calL_Y g = \frac{1}{2} \calL_{\partial_r - \frac{1}{V} \partial_t }g 
 = \frac{1}{2} \calL_{\partial_r }g+\frac{1}{2} \calL_{- \frac{1}{V} \partial_t }g\phantom{\frac{1}{V'}}\\
 &=& 
 \frac{1}{2}\frac{V'}{V^2}dr^2 +\frac{V'}{2} dt^2 +r d\Omega^2 +\frac{V'}{V}dt dr
 \\
 &=& \frac{1}{2} V' \Big( dt +\frac{1}{V} dr \Big)^2 +r d\Omega^2 = \frac{1}{2} V' du^2 +r d\Omega^2.
\end{eqnarray*}
We have 
$$
 du(\partial_r)=\Big(dt+\frac{1}{V}\Big)\partial_r=\frac{1}{V},\;\; \textrm{ and } \;\; du(\partial_t)=1, 
$$
and on the other hand, 
$$
-g(Y,\partial_r)=-g\Big(\partial_r-\frac{1}{V}\partial_t,\;\! \partial_r\Big)=\frac{1}{V} ,\textrm{ and } 
-g(Y,\partial_t)= -g\Big(\partial_r-\frac{1}{V}\partial_t,\;\! \partial_t\Big)=1,
$$
showing that 
$$
du=-g(Y,\cdot).
$$
Also, we recall that 
\begin{eqnarray*}
T_{\mu \nu}
&=&
\partial_\mu\phi \partial_\nu\phi -\frac{g_{\mu\nu}}{2}\left(\partial_\alpha \phi \partial^\alpha\phi +m^2\phi^2\right) 
\\
&=&
\left( d\phi \otimes d\phi \!-\!\frac{1}{2}\langle d\phi,d\phi\rangle g-\frac{1}{2}m^2\phi^2 g\right)_{\mu \nu}.
\end{eqnarray*}
It follows that 
\begin{eqnarray}
 \nonumber  \!\!\!\!\!\!\!\!\!\!
 T^{\mu\nu}\Xi_{\mu\nu}\!\!\!&\!=\!&\!\!\!\Xi^{\mu \nu}T_{\mu\nu}\phantom{\frac{1}{2}}\\
 \label{9_Aug_1618}
 &\!=\!&\!\!\! \frac{1}{2} V'(Y\cdot\phi)^2+\frac{1}{r^3} |\cnab \phi|^2 -\frac{n-1}{2r} \langle d\phi,d\phi\rangle -\frac{n-1}{2r} m^2\phi^2.
\end{eqnarray}
 
\noindent 
We have $d\phi(\partial_\mu)=\partial_\mu\phi$. So 
\begin{eqnarray*}
\langle d\phi,d\phi\rangle &=& g^{\alpha\beta} (d\phi)_\alpha (d\phi)_\beta 
=g^{\alpha\beta} (\partial_\alpha \phi)( \partial_\beta\phi)=(\partial_\alpha \phi)( \partial^\alpha \phi)\phantom{\frac{1}{V}}\\
&=&\phi'^2 (-V)+\dot{\phi}^2\frac{1}{V}+\frac{1}{r^2} |\cnab \phi|^2.
\end{eqnarray*}
Also, 
\begin{eqnarray*}
 -2(K\cdot\phi)(Y\cdot\phi)-V(Y\cdot\phi)^2
 &=& -2\dot{\phi}\Big(\phi'-\frac{1}{V}\dot{\phi}\Big)-V\Big(\phi'-\frac{1}{V} \dot{\phi}\big)^2
 \\
 &=& \frac{1}{V} \dot{\phi}^2 -V\phi'^2.
\end{eqnarray*}
So 
\begin{equation}
 \label{9_8_2019_1556} 
\langle d\phi,d\phi \rangle
=
-2(K\cdot\phi)(Y\cdot\phi)-V(Y\cdot\phi)^2+\frac{1}{r^2} |\cnab\phi|^2.
\end{equation}
Combining \eqref{9_Aug_1618} and \eqref{9_8_2019_1556}, we obtain 
{\small 
\begin{eqnarray*}
 T^{\mu \nu}\Xi_{\mu \nu}
 \!\!\!\!&\!=\!&\!\!\!\!
 \Big( \frac{V'}{2}\!+\!\frac{n\!-\!1}{2r} V\Big) (Y\!\cdot\!\phi)^2 \!+\!\frac{n\!-\!1}{r} (K\!\cdot\! \phi) (Y\!\cdot\! \phi) \!-\!\frac{n\!-\!3}{2r^3}|\cnab \phi|^2\!-\!\frac{n\!-\!1}{2r}m^2 \phi^2 
 \\
 \!\!\!\!&\!=\!&\!\!\!\! 
 \frac{V'}{2} \Big( Y\!\cdot \!\phi\!+\!\frac{n\!-\!1}{rV'} (K\!\cdot\! \phi)\Big)^2\!-\!\frac{(n\!-\!1)^2}{2r^2 V'} (K\!\cdot\!\phi)^2
  \!-\! \frac{n\!-\!3}{2r^3}|\cnab \phi|^2\!-\!\frac{n\!-\!1}{2r}m^2 \phi^2 .
\end{eqnarray*}}

\vspace{-0.2cm} 

\noindent 
Now as $V'(r)> 0$ for $r\geq r_c$ (global redshift), it follows that the first summand in the last expression is nonnegative, and so we obtain the inequality 
\begin{equation}
 \label{9_aug_16:29} 
T^{\mu \nu}\Xi_{\mu \nu}\geq -\frac{(n-1)^2}{2r^2 V'} (K\cdot\phi)^2- \frac{n-3}{2r^3}|\cnab \phi|^2-\frac{n-1}{2r}m^2 \phi^2 .
\end{equation}
Now suppose that $r_0$ is fixed. As $r^2 V'(r)>0$ for all $r\in [r_c,r_0]$, we have (by the extreme value theorem)
$$
\min_{r\in [r_c,r_0]}r^2 V'(r)>0.
$$
Thus 
$$
-\frac{(n-1)^2}{2r^2 V'}\geq -\frac{(n-1)^2}{\min\limits_{r\in [r_c,r_0]}r^2 V'(r)} =:-C_1(r_0).
$$

\smallskip 

\noindent 
Similarly, for $r\in [r_c,r_0]$, $\displaystyle \frac{1}{r^3}\leq \frac{1}{r_c^3}$, and so 
$$
-\frac{n-3}{2r^3}  \geq -\frac{n-3}{2r_c^3} =:-\widetilde{C}_2(r_c).
$$
Also,  for $r\in [r_c,r_0]$,
$$
-\frac{(n-1)}{2r}m^2  \geq -\frac{(n-1)}{2r_c}m^2 =:-C_3(r_c).
$$
Also, if we set 
$$
\Pi:=\frac{1}{2}\calL_X g,
$$
then from Step 1 of the proof of Theorem~\ref{theorem_2b} (see in particular the inequality \eqref{15:19:5:Nov:2019} on page \pageref{15:19:5:Nov:2019}), 
we have for $r>r_c$ that 
\begin{eqnarray*}
T^{\mu\nu}\Pi_{\mu\nu}&\geq&  -\left( 1+\frac{n}{2} + \frac{e^2}{2r_c^{n-1}}\right)\frac{1}{r^{n+2}}  |\cnab \phi|^2 \\
&\geq &-\underbrace{ \left( 1+\frac{n}{2} + \frac{e^2}{2r_c^{n-1}}\right)\frac{1}{r_c^{n+2}}}_{=:\widetilde{\widetilde{C}}_2(r_c)}|\cnab \phi|^2.
\end{eqnarray*}
Set $C_2(r_c):=\widetilde{C}_2(r_c)+\widetilde{\widetilde{C}}_2(r_c)$. We have 
\begin{eqnarray*}
T^{\mu\nu}\Pi_{\mu\nu}\!\!\!\!&+&\!\!\!T^{\mu\nu} \Xi_{\mu \nu}\\
&\geq& -\widetilde{\widetilde{C}}_2(r_c)|\cnab\phi|^2 -C_1(r_0) (K\cdot \phi)^2 -\widetilde{C}_2(r_c) |\cnab \phi|^2 -C_3(r_c) \phi^2\\
&=& -C_1(r_0) (K\cdot \phi)^2 -C_2(r_c) |\cnab \phi|^2 -C_3(r_c) \phi^2\\
&\geq & -\underbrace{\max\{C_1(r_0),\;\! C_2(r_c),\;\! C_3(r_c)\}}_{=:C(r_c,r_0)>0}\cdot \left( (K\cdot \phi)^2+|\cnab \phi|^2 +\phi^2\right).
\end{eqnarray*}
For $r_1\in (r_c,r_0)$, and  $T>0$, define 
 $
\calD=\{r=r_1\}\;\!\cap\;\! \{-T\leq t\leq T\}.
$ 
We now apply the divergence theorem, with the current $J$ corresponding to the multiplier $X+Y$, in the region 
 $
\calT=D^+(\calD)\;\! \cap \;\! \{r\leq r_0\}. 
$ 
 Noticing that the flux across the future null boundaries is less than or equal to $0$, we obtain, 
after passing the limit $T\rightarrow \infty$, that 
\begin{equation}
 \label{9_8_19:15}
\widetilde{E}(r_1)-\widetilde{E}(r_0)\geq -\!\int_{r_1}^{r_0} \!\int_{\mR\times S^{n-1}}\!\!C(r_0,r_c)\left( (K\cdot \phi)^2+|\cnab \phi|^2 +\phi^2\right)r^{n-1}dtd\Omega dr.
\end{equation}
But 
\begin{eqnarray*}
 \widetilde{E}(r)&=& \frac{1}{2} \int_{\mR\times S^{n-1}}\left(V^2 \phi'^2 +\dot{\phi}^2+\frac{V}{r^2} |\cnab \phi|^2+m^2V\phi^2\right)dtd\Omega
 \\
 &&\!\!\!+
 \frac{1}{2}\! \int_{\mR\times S^{n-1}}\left(V (Y\cdot \phi)^2+\frac{1}{r^2} |\cnab \phi|^2+m^2\phi^2\right)r^{n-1} dt d\Omega.
\end{eqnarray*}
In particular,
\begin{eqnarray*}
 \int_{r_1}^{r_0} \int_{\mR\times S^{n-1}}\!(K\cdot \phi)^2 r^{n-1} dt d\Omega dr \!&\leq &\! \int_{r_1}^{r_0} \int_{\mR\times S^{n-1}}\dot{\phi}^2 r_0^{n-1} dt d\Omega dr\\
 &\leq &\int_{r_1}^{r_0} 2\widetilde{E}(r)r_0^{n-1} dr=2r_0^{n-1} \int_{r_1}^{r_0} \widetilde{E}(r)dr.
\end{eqnarray*}
Also, 
\begin{eqnarray*}
 \int_{r_1}^{r_0} \int_{\mR\times S^{n-1}}|\cnab \phi|^2 r^{n-1} dt d\Omega dr &\leq & \int_{r_1}^{r_0} \int_{\mR\times S^{n-1}} \frac{|\cnab \phi|^2}{r^2} r^{n-1} r^2 dt d\Omega dr\\
 &\leq &\int_{r_1}^{r_0} 2\widetilde{E}(r)r_0^{2} dr=2r_0^{2} \int_{r_1}^{r_0} \widetilde{E}(r)dr.
\end{eqnarray*}
Finally, 
\begin{eqnarray*}
 \int_{r_1}^{r_0} \int_{\mR\times S^{n-1}} \phi^2 r^{n-1} dt d\Omega dr &\leq & \int_{r_1}^{r_0} \int_{\mR\times S^{n-1}} m^2 \phi^2 r^{n-1}\frac{1}{m^2} dt d\Omega dr\\
 &\leq &\int_{r_1}^{r_0} 2\widetilde{E}(r)\frac{1}{m^2} dr=\frac{2}{m^2} \int_{r_1}^{r_0} \widetilde{E}(r)dr.
\end{eqnarray*}
Using the above three estimates, it follows from \eqref{9_8_19:15} that 
\begin{eqnarray}
 \nonumber 
 \widetilde{E}(r_1)-\widetilde{E}(r_0)&\geq& -\int_{r_1}^{r_0} C(r_0,r_c)\left( 2r_0^{n-1}+2r_0^2+\frac{2}{m^2}\right)\widetilde{E}(r)dr \\
 \label{9_august_19_16}
& =&-\int_{r_1}^{r_0} k(r_0,r_c) \widetilde{E}(r)dr,
\end{eqnarray}
where  $\displaystyle 
k(r_0,r_c):=C(r_0,r_c)\left( 2r_0^{n-1}+2r_0^2+\frac{2}{m^2}\right).
$

\smallskip 

\noindent 
Now suppose that $r_2$ is such that $
r_c<r_1<r_2<r_0.
$

\smallskip 

\noindent 
If we redo all of the above steps in order to obtain \eqref{9_august_19_16}, 
but with $r_2$ replacing $r_0$, we obtain 
\begin{equation}
\label{nau_august_19:39}
\widetilde{E}(r_1)-\widetilde{E}(r_2)\geq -\int_{r_1}^{r_2} k(r_2,r_c) \widetilde{E}(r)dr,
\end{equation}
where $\displaystyle 
 k(r_2,r_c)=C(r_2,r_c)\left( 2r_2^{n-1}+2r_2^2+\frac{2}{m^2}\right).
$

\smallskip 

\noindent 
But 
\begin{eqnarray*}
\!\! k(r_2,r_c)\!\!&=&\!\!C(r_2,r_c)\left( 2r_2^{n-1}+2r_2^2+\frac{2}{m^2}\right)\\
 &\leq &\!\!C(r_2,r_c)\left( 2r_0^{n-1}+2r_0^2+\frac{2}{m^2}\right)\\
 &=&\!\!\max\{ C_1(r_2),\;\! C_2(r_c),\;\! C_3(r_c)\} \cdot \left( 2r_0^{n-1}+2r_0^2+\frac{2}{m^2}\right).
\end{eqnarray*}
We have 
$$
C_1(r_2)=\frac{(n-1)^2}{\min\limits_{r\in [r_c,r_2]}r^2 V'(r)}
\leq 
\frac{(n-1)^2}{\min\limits_{r\in [r_c,r_0]}r^2 V'(r)}=C_1(r_0),
$$
since $[r_c,r_2]\subset [r_c,r_0]$. Hence 
\begin{eqnarray*}
 \!\!k(r_2,r_c)\!\!&\leq &\!\!\max\{ C_1(r_2),\;\! C_2(r_c),\;\! C_3(r_c)\}\! \cdot \!\left( 2r_0^{n-1}+2r_0^2+\frac{2}{m^2}\right)\\
 &\leq &\!\!\max\{ C_1(r_0),\;\! C_2(r_c),\;\! C_3(r_c)\} \!\cdot \!\left( 2r_0^{n-1}+2r_0^2+\frac{2}{m^2}\right)=k(r_0,r_c).
\end{eqnarray*}
So from \eqref{nau_august_19:39}, we get
$$
\widetilde{E}(r_1)-\widetilde{E}(r_2)\geq -\int_{r_1}^{r_2} k(r_2,r_c) \widetilde{E}(r)dr\geq  -\int_{r_1}^{r_2} k(r_0,r_c) \widetilde{E}(r)dr.
$$
Consequently, for all $r_2\in [r_1,r_0)$, $\displaystyle  \widetilde{E}(r_2)\leq \widetilde{E}(r_1)+\int_{r_1}^{r_2} k(r_0,r_c) \widetilde{E}(r) dr $.

\smallskip 

\noindent 
By the integral form of Gr\"onwall's inequality (see e.g. \cite[Thm.~1.10]{Tao}), we obtain for all $r_2\in [r_1,r_0)$ that 
$$
\widetilde{E}(r_2)\leq \widetilde{E}(r_1) e^{\int_{r_1}^{r_2} k(r_0,r_c)dr } =\widetilde{E}(r_1) e^{k(r_0,r_c)\cdot (r_2-r_1)}.
$$
Passing the limit as $r_2\nearrow r_0$ yields
$$
\widetilde{E}(r_0) \leq \widetilde{E}(r_1)e^{ k(r_0,r_c)\cdot (r_0-r_1)},
$$
and this holds for all $r_1\in (r_c, r_0)$. Now passing the limit as $r_1\searrow r_c$, we obtain
$$
\widetilde{E}(r_0)\leq \widetilde{E}(r_c)e^{ k(r_0,r_c)\cdot (r_0-r_c)}.
$$
Consequently, 
\begin{eqnarray*}
 E(r_0)\leq \widetilde{E}(r_0) &\leq & \widetilde{E}(r_c) e^{k(r_0,r_c)\cdot (r_0-r_c)} \\
 &\lesssim & \widetilde{E}(r_c) \lesssim \|\phi\|^2_{H^1(\calC\calH_1^+)}+\|\phi\|^2_{H^1(\calC\calH_2^+)}<+\infty.
\end{eqnarray*}
Commuting with the Killing vector fields $\frac{\partial}{\partial t}$ and $L_{ij}$, we see that the hypothesis from Theorem~\ref{theorem_4}, namely, 
$$
\|\phi\|_{H^k(\calC\calH_1^+)}<+\infty \;\; \;\textrm{ and }\; \;\;\|\phi\|_{H^k(\calC\calH_2^+)}<+\infty,
$$
for some $k>\frac{n}{2}+2$,  yields also that 
$$
\|\phi\|_{H^k(\{r=r_0\})}\lesssim 
\|\phi\|_{H^k(\calC\calH_1^+)}+\|\phi\|_{H^k(\calC\calH_2^+)}<+\infty.
$$
We now show that this justifies the assumption used in the previous two subsections. 
For simplicity, we only consider just one of the energies 
$$
\calE(r)=\frac{1}{2}\int_{\mR\times S^{n-1}} \left( \psi'^2+\frac{1}{V^2} \dot{\psi}^2 +\frac{1}{r^2 V} |\cnab \psi |^2+\theta \psi^2\right) dt d\Omega.
$$
(The proof of the finiteness of  $\widetilde{\calE}(r_0)$ is entirely analogous.) 
As $\psi=r^{\kappa} \phi$, we obtain 
finiteness of the last summand, namely 
\begin{eqnarray*}
\phantom{aaa}\int_{\mR\times S^{n-1}} \theta(r_0) \left(\psi(r_0,\cdot)\right)^2 dt d\Omega
&=& \theta(r_0) r_0^{2\kappa} \int_{\mR\times S^{n-1}}\left(\phi(r_0,\cdot)\right)^2 dt d\Omega\\
&\leq &\theta(r_0) r_0^{2\kappa} \|\phi(r_0,\cdot)\|_{H^1(\{r=r_0\})}^2<+\infty. \phantom{\int_{S^{n}}}
\end{eqnarray*}
We have 
\begin{eqnarray*}
\int_{\mR\times S^{n-1}} \left(\phi'(r_0,\cdot)\right)^2 dt d\Omega 
&=&
\frac{1}{(V(r_0))^2}\int_{\mR\times S^{n-1}} \left(V(r_0,\cdot) \phi'(r_0,\cdot)\right)^2 dt d\Omega \\
&\leq& \frac{2E(r_0)}{(V(r_0))^2} <+\infty.
\end{eqnarray*}
Since $\psi'(r_0,\cdot)=\kappa r_0^{\kappa-1} \phi(r_0,\cdot)+r_0^\kappa \phi'(r_0,\cdot)$, and 
as $\phi(r_0,\cdot)\in H^1(\{r=r_0\})$, we have $\psi(r_0,\cdot)\in L^2(\mR\times S^{n-1})$, that is, 
$ \displaystyle 
\int_{\mR\times S^{n-1}}  \left(\psi'(r_0,\cdot)\right)^2 dtd\Omega <+\infty .
$ 

\smallskip 

\noindent 
We also have 
\begin{eqnarray*}
\int_{\mR\times S^{n-1}} \frac{1}{\left(V(r_0)\right)^2} \left(\dot{\psi}(r_0,\cdot)\right)^2 dt d\Omega 
&=&
\frac{r_0^{2\kappa}}{\left(V(r_0)\right)^2}\int_{\mR\times S^{n-1}}  \left(\dot{\phi}(r_0,\cdot)\right)^2dtd\Omega 
\\
&\leq& \frac{r_0^{2\kappa}}{\left(V(r_0)\right)^2}\|\phi(r_0,\cdot)\|_{H^1(\{r=r_0\})}^2<+\infty.
\end{eqnarray*}
Finally, 
\begin{eqnarray*}
\int_{\mR\times S^{n-1}} \frac{1}{r_0^2 V(r_0)} |\cnab \psi(r_0,\cdot)|^2  dt d\Omega
&=&
\frac{r_0^{2\kappa}}{r_0^2 V(r_0)}\int_{\mR\times S^{n-1}}  |\cnab \phi(r_0,\cdot)|^2  dt d\Omega
\\
&\lesssim &
\|\phi(r_0,\cdot)\|_{H^1(\{r=r_0\})}^2<+\infty.\phantom{\int_{S^{n}}\frac{1}{V^2}}
\end{eqnarray*}
Thus each summand in the expression for $\calE(r_0)$ is finite.

\noindent 
This completes the proof of Theorem~\ref{theorem_4}.

\section{Decay in $\mathrm{RNdS}$ when $m=0$, the wave equation}
\label{Section_thm_2b}

\noindent 
In \cite[Theorem~2]{CNO}, the following result was shown:

\begin{theorem}$\;$\label{theorem_CNO_2}

\noindent 
Suppose that 
\begin{itemize}
 \item $\delta>0$,  
 \item $M>0$, 
 \item $e\geq 0$,
 \item $n> 2$, 
 \item $(M,g)$ is the $(n+1)$-dimensional subextremal Reissner-Nordstr\"om-de Sitter solution given by the metric 
   $$
   g=-\frac{1}{V} dr^2+Vdt^2 +r^2 d\Omega^2,
   $$
   where 
   $$
   V= r^2+\frac{2M}{r^{n-2}} -\frac{e^2}{r^{n-1}}-1,
   $$
   and $d\Omega^2$ is the metric of the unit $(n-1)$-dimensional sphere $S^{n-1}$, 
 \item $k>\frac{n}{2}+2$, and 
 \item $\phi$ is a smooth solution to $\square_g \phi=0$ such that 
$$
\|\phi\|_{H^{k}(\calC\calH_1^+)}<+\infty \;\;\;\textrm{ and }\;\;\;\|\phi\|_{H^{k}(\calC\calH_2^+)}<+\infty,
$$
where $\calC\calH_1^+\simeq \calC\calH_2^+\simeq \mR\times S^{n-1}$ are the two components of the 
future cosmological horizon, parameterised by the flow parameter of the global Killing vector field $\frac{\partial}{\partial t}$. 
\end{itemize}
Then there exists a $r_0$ large enough so that for all $r\geq r_0$, 
$$
\| \partial_r\phi(r,\cdot)\|_{L^\infty(\mR\times S^{n-1})}\lesssim r^{-3+\delta}.
$$
\end{theorem}

\noindent Using a method similar to the one we used to show Rendall's conjecture in Theorem~\ref{theorem_2}, we 
can improve the almost-exact bound of $r^{-3+\delta}$ to $r^{-3}$. 

\begin{remark} 
As observed in \cite[Remark~1.4]{CNO}, this decay rate bound of $r^{-3}$ for $\partial_r  \phi$  is in fact  the decay rate one would expect in light of Rendall's conjecture. 
Indeed, for freely falling observers in 
the cosmological region, one has 
$$
r(\tau)\sim e^\tau\sim a(\tau), 
$$
where $\tau$ is the proper time, and $a(\tau)$ is the radius of a comparable de Sitter universe in flat FLRW form, giving 
$$
\partial_r \phi\sim \frac{\partial_\tau \phi}{\partial_\tau r}\sim \frac{e^{-2\tau}}{e^\tau}\sim \frac{1}{r^3}.
$$
\end{remark}

\goodbreak

\noindent Thus our improved version of Theorem~\ref{theorem_CNO_2} is the following result.

\begin{theorem}$\;$\label{theorem_2b}

\noindent 
Suppose that 
\begin{itemize}  
 \item $M>0$, 
 \item $e\geq 0$,
 \item $n> 2$, 
 \item $(M,g)$ is the $(n+1)$-dimensional subextremal Reissner-Nordstr\"om-de Sitter solution given by the metric 
   $$
   g=-\frac{1}{V} dr^2+Vdt^2 +r^2 d\Omega^2,
   $$
   where 
   $$
   V= r^2+\frac{2M}{r^{n-2}} -\frac{e^2}{r^{n-1}}-1,
   $$
   and $d\Omega^2$ is the metric of the unit $(n-1)$-dimensional sphere $S^{n-1}$, 
 \item $k>\frac{n}{2}+2$, and 
 \item $\phi$ is a smooth solution to $\square_g \phi=0$ such that 
$$
\|\phi\|_{H^{k}(\calC\calH_1^+)}<+\infty \;\;\;\textrm{ and }\;\;\;\|\phi\|_{H^{k}(\calC\calH_2^+)}<+\infty,
$$
where $\calC\calH_1^+\simeq \calC\calH_2^+\simeq \mR\times S^{n-1}$ are the two components of the 
future cosmological horizon, parameterised by the flow parameter of the global Killing vector field $\frac{\partial}{\partial t}$. 
\end{itemize}
Then there exists a $r_0$ large enough so that for all $r\geq r_0$, 
$$
\| \partial_r\phi(r,\cdot)\|_{L^\infty(\mR\times S^{n-1})}\lesssim r^{-3}.
$$
\end{theorem}

\begin{proof}$\;$ 

\noindent {\bf Step 1:} We will first establish the following estimates: 
there exists an $r_0$ large enough such that  for all $r\geq r_0$, 
\begin{eqnarray*}
 \|\ddot{\phi}(r,\cdot)\|_{L^\infty(\mR,S^{n-1})}\!\! & \lesssim & \!\! 1,\\
 \|\cLap \phi(r,\cdot)\|_{L^\infty(\mR,S^{n-1})} \!\! & \lesssim & \!\! 1.
\end{eqnarray*}
We will follow \cite[\S3.2]{CNO} in order to obtain the bounds above,
which will be needed in Step 2 of our proof below. 
We repeat this preliminary step here from \cite[\S3.2]{CNO} for the 
sake of completeness and for the convenience of the reader.
 
 Suppose that $\phi$ satisfies the wave equation $\square_g\phi=0$. 
 Recall that the energy-momentum tensor associated with $\phi$ is given by 
 $$
 T_{\mu \nu}=\partial_\mu \phi \partial_\nu \phi -\frac{1}{2} g_{\mu\nu} \partial_\alpha \partial^\alpha \phi.
 $$
 Thus 
 \begin{eqnarray*}
  T(N,N)
  &=&
  \left( \phi'^2\!-\!\frac{1}{2}\frac{(-1)}{V} \Big(\phi'^2(-V)\!+\!\dot{\phi}^2 \frac{1}{V}\!+\!\frac{1}{r^2} |\cnab\phi|^2 \Big)\right)V
  \\
  &=&
   \frac{1}{2}\left( V\phi'^2 \!+\!\frac{1}{V}\dot{\phi}^2 \!+\!\frac{1}{r^2} |\cnab\phi|^2\right).
 \end{eqnarray*}
  Define 
 $\displaystyle 
 X:=\frac{V^{\frac{1}{2}}}{r^{n-1}} N=\frac{V^{\frac{1}{2}}}{r^{n-1}} V^{\frac{1}{2}}\frac{\partial}{\partial r}=\frac{V}{r^{n-1}}\frac{\partial}{\partial r}.
 $
 
 \medskip 
 
 \noindent 
 The current $J$ is given by   
 $
 J_\mu:=T_{\mu \nu} X^\nu.
 $
 We define the energy 
 $$
 E(r)\!:=\!\int_{\mR\times S^{n-1}} \!T(X,N) dV_n
 \!=\!
 \frac{1}{2} \int_{\mR\times S^{n-1}}  \!\left( V^2 \phi'^2 \!+\!\dot{\phi}^2 \!+\!\frac{V}{r^2} |\cnab \phi|^2\right)dt d\Omega.
 $$
 The deformation tensor $\Pi$ associated to the multiplier $X$ is given by 
 $$
 \Pi=\frac{1}{2}\calL_X g=
 -\frac{1}{V} dr \calL_Xdr +\frac{V}{2r^{n-1}} \frac{V'}{V^2} dr^2 +\frac{V'V}{2r^{n-1}}dt^2 +\frac{V}{r^{n-1}}rd\Omega^2.
 $$
 It can be shown that 
 $\displaystyle 
 \calL_X dr =\left(\frac{V'}{r^{n-1}}-\frac{(n-1) V}{r^n} \right)dr.
 $ 
 Thus 
 \begin{eqnarray*}
  \Pi\!\!\!
  &=& \!\!\!
  -\frac{1}{V} \left(\frac{V'}{r^{n-1}}\!-\!\frac{(n-1) V}{r^n} \right)dr^2 \!+\!\frac{V'}{2r^{n-1}V} dr^2 \!+\!\frac{V'V}{2r^{n-1}}dt^2 \!+\!\frac{V}{r^{n-1}}rd\Omega^2\\
  &=&\!\!\!
  \underbrace{\frac{V'}{2r^{n-1}}\left(-\frac{1}{V}dr^2+Vdt^2\right)}_{=:\Pi^{(1)}} 
  +
  \underbrace{\phantom{\bigg(}\frac{n-1}{r^n} dr^2\phantom{\bigg)}}_{=:\Pi^{(2)}} 
  +
  \underbrace{\phantom{\bigg(}\frac{V}{r^{n-2}} d\Omega^2\phantom{\bigg)}}_{\Pi^{(3)}} .
 \end{eqnarray*}
 We have 
 \begin{eqnarray*}
  T^{\mu\nu}\Pi^{(1)}_{\mu \nu}
  \!\!\!\!&\!=\!&\!\!\!
  \frac{V'}{2r^{n-1}}\left(-\frac{1}{2}V\phi'^2\!-\!\frac{1}{2V} \dot{\phi}^2\!-\!\frac{1}{2r^2}|\cnab \phi|^2 \right.
  \\
&&  \phantom{aaaaaa}+\left. \frac{1}{V}\left(\dot{\phi}^2 \!-\!\frac{1}{2}V\Big(\phi'^2(-V)+\dot{\phi}^2\frac{1}{V}\!+\!\frac{1}{r^2}|\cnab\phi|^2\Big)\right)\right)\\
  &=& \frac{V'}{2r^{n-1}}\left(-\frac{1}{r^2}|\cnab \phi|^2\right).
 \end{eqnarray*}
Also, 
\begin{eqnarray*}
T^{\mu\nu}\Pi^{(2)}_{\mu \nu}
  &=& 
  \frac{(n-1)}{2r^n}\left(V^2 \phi'^2 +\dot{\phi}^2 +\frac{V}{r^2} |\cnab \phi|^2 \right).
  \end{eqnarray*}
  Finally, 
  \begin{eqnarray*}
  T^{\mu\nu}\Pi^{(3)}_{\mu \nu}
  \!\!\!&\!=\!&\!\!\! 
  \frac{V}{r^{n-2}} \left( \frac{1}{r^4} |\cnab \phi|^2 \!-\!\frac{(n-1)}{2r^2} \left( \phi'^2(-V)\!+\!\dot{\phi}^2\frac{1}{V}\!+\!\frac{1}{r^2}|\cnab \phi|^2 \right)\right)\\
  \!\!\!&\!=\!&\!\!\!
  \frac{V}{r^{n+2}} |\cnab \phi|^2 +\frac{(n-1)}{2r^n} \left( V^2 \phi'^2 \!-\!\dot{\phi}^2\!-\!\frac{V}{r^2} |\cnab \phi|^2\right).
  \end{eqnarray*}
  Consequently, the full bulk term is 
  $$
   \nabla_\mu J^\mu = 
   T^{\mu \nu} \Pi_{\mu \nu} 
   =
   \frac{(n-1)}{r^n}V^2\phi'^2 +|\cnab \phi|^2 \left(\frac{V}{r^{n+2}}-\frac{V'}{2r^{n+1}}\right).
  $$
  Using the expression for $V$, we compute 
  $$
  \frac{V}{r^{n+2}}-\frac{V'}{2r^{n+1}}
  =
  -\frac{1}{r^{n+2}}\left(1+\frac{(n+1)e^2}{2r^{n-1}}-\frac{nM}{r^{n-2}}\right). 
  $$
  For $r>r_c$, we have $V(r)>0$, and so 
  \begin{eqnarray*}
   1+\frac{(n+1)e^2}{2r^{n-1}}-\frac{nM}{r^{n-2}}&=& 1-\frac{n}{2}\left( r^2+\frac{2M}{r^{n-2}}-\frac{e^2}{r^{n-1}}-1\right) -\frac{n}{2}+\frac{e^{2}}{2r^{n-1}}
   \\
   &=&1-\frac{n}{2} V(r)-\frac{n}{2}+\frac{e^{2}}{2r^{n-1}} \\
   &\leq& 1-0-\frac{n}{2}+\frac{e^{2}}{2r^{n-1}}\leq 1+\frac{n}{2}+\frac{e^{2}}{2r_c^{n-1}}=:C.
  \end{eqnarray*}
  Hence 
   \begin{equation}
   \label{15:19:5:Nov:2019}
    \displaystyle 
   \nabla_\mu J^\mu = 
   T^{\mu \nu} \Pi_{\mu \nu} 
   \geq -\frac{C}{r^{n+2}} |\cnab \phi|^2 .
   \end{equation}
   For each $T<0$, define the set 
   $$
   \calC:=\{r=r_0\} \;\! \cap \;\! \{-T\leq t\leq T\}.
   $$
   Also, consider the region 
    $
   \calS:=D^+(\calC)\;\!\cap \;\! \{r\leq r_1\}.
   $ 
%

\begin{figure}[h]
    \center
      \includegraphics[width=9 cm]{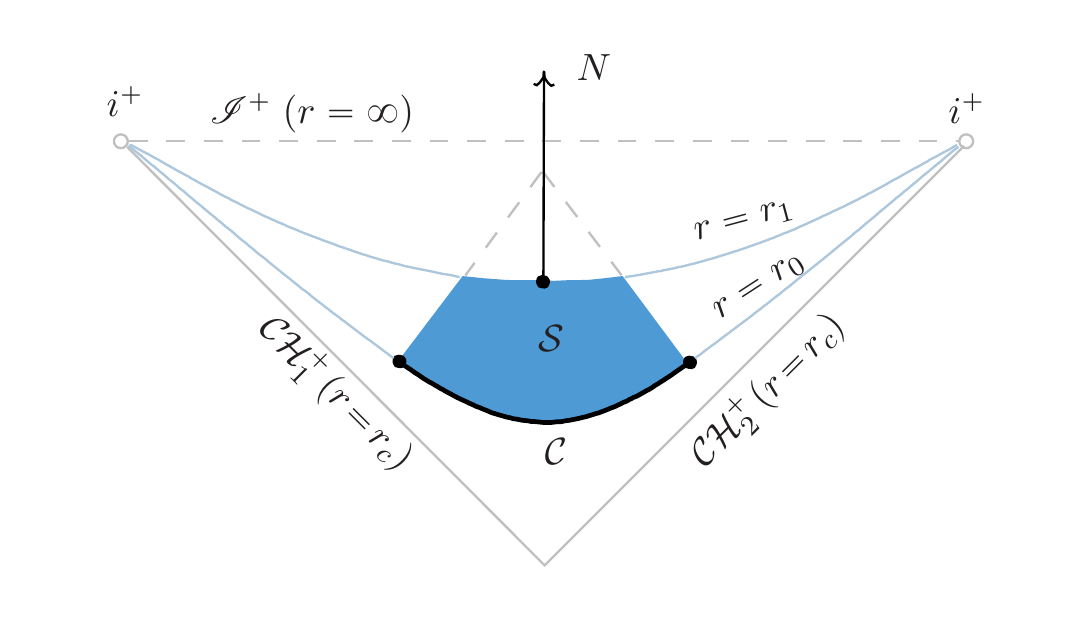}
 \end{figure}

\noindent We will apply the divergence theorem to the current $J$ on the region $\calS $. 
As the flux across the future null boundaries is nonpositive, we have 
\begin{eqnarray*}
&&
\int_{r_0}^{r_1}\int_{\mR\times S^{n-1}} \left( -\frac{C}{r^{n+2}} |\cnab \phi|^2 \right)dtd\Omega dr
\\
&\leq & 
\int_{\calS} (\nabla_\mu J^\mu) \epsilon= 
\underbrace{ \int_{\substack{\tiny\textrm{null}\\ \tiny \textrm{part}}} J \intprod \epsilon}_{\leq 0} +\underbrace{\int_{\{r=r_1\}} J\intprod \epsilon}_{-E(r_1)} 
+\underbrace{\int_{\{r=r_0\}} J\intprod \epsilon}_{E(r_0)} .
\end{eqnarray*}
So 
$\displaystyle 
E(r_0)-E(r_1)\geq -\int_{r_0}^{r_1} \int_{\mR\times S^{n-1}} \frac{C}{r^{3}} |\cnab \phi|^2 dtd\Omega dr.
$ 
We have 
$$
\int_{\mR\times S^{n-1}} |\cnab \phi|^2 dt d\Omega =\int_{\mR\times S^{n-1}} \frac{V}{r^2}|\cnab \phi|^2 \frac{r^2}{V} dt d\Omega \leq 2E(r)\cdot \frac{r^2}{V} =\frac{2r^2}{V}E(r).
$$
Since
$$
\lim_{r\rightarrow \infty}\frac{2r^2}{V}=\lim_{r\rightarrow \infty}\frac{2}{1+\frac{2M}{r^n}-\frac{e^2}{r^{n+1}}-\frac{1}{r^2}}=\frac{2}{1},
$$
there exists an $r_0$ large enough such that for all $r\geq r_0$, $|\frac{2r^2}{V}-2|<1$, and in particular, 
 $
(0<) \frac{2r^2}{V}<3.
$ 
Hence 
$$
 E(r_1) \!\leq \! E(r_0) \!+\!\int_{r_0}^{r_1} \int_{\mR\times S^{n-1}} \!\left( \frac{C}{r^{3}} |\cnab \phi|^2 \right)dtd\Omega dr
 \!\leq \! E(r_0)\!+\!\int_{r_0}^{r_1} \frac{3C}{r^3} E(r) dr .
$$
Using Gr\"onwall's inequality (see for e.g. \cite[Thm.~1.10]{Tao}), we obtain 
$$
E(r_1)\leq E(r_0) e^{\int_{r_0}^{r_1} \frac{3C}{r^3} dr}\lesssim C(r_0) E(r_0), 
$$
as was also noted in \cite[Eq.(71)]{CNO}. Thus, we have in particular that 
there exists an $r_0$ large enough such that  for all $r\geq r_0$, 
$$
 \int_{\mR \times S^{n-1}} \dot{\phi}^2 dtd\Omega \lesssim 1,\quad  \textrm{ and } \quad 
 \int_{\mR \times S^{n-1}} |\cnab \phi|^2 dtd\Omega \lesssim 1.
$$
Commuting with the Killing vector fields $\frac{\partial}{\partial t}$ and $L_{ij}$, 
we obtain (after the transferral of the 
finiteness of the energies along the branches $\calC\calH_1^+$ and $\calC\calH_2^+$ of the cosmological horizon  
to finiteness at $r=r_0$, and an application of Sobolev's inequality) that 
for all $r\geq r_0$, 
$$
\|\ddot{\phi}(r,\cdot)\|_{L^\infty(\mR,S^{n-1})} \lesssim  1,\quad \textrm{ and }\quad 
\|\cLap \phi(r,\cdot)\|_{L^\infty(\mR,S^{n-1})}    \lesssim  1.
$$

\medskip

\noindent {\bf Step 2:} 
In this step, we will write the wave equation in new coordinates which `equalises' the magnitude of the coefficient weights for the $r$ and $t$ coordinates in the matrix of the metric. 

\medskip 

\noindent 
To this end, we define 
 $\displaystyle  
\rho= \int_{r_0}^r \frac{1}{V(r)} dr.
$ 
Then 
 $
\displaystyle 
\frac{d\rho}{dr}=\frac{1}{V(r)}\textrm{ and }  V(r) \displaystyle \frac{d}{dr}=\frac{d}{d\rho}.
$

\medskip 

\noindent 
With a slight abuse of notation, we write 
 $
V(\rho):=V(r(\rho)).
$ 
We have 
\begin{eqnarray*}
g&=& -\frac{1}{V} dr^2 + V dt^2  +r^2d\Omega^2\\
&=& -Vd\rho^2+Vdt^2+(r(\rho))^2 d\Omega^2.\phantom{\frac{1}{V}}
\end{eqnarray*}
The wave equation $\square_g \phi=0$ can be rewritten as
$\partial_\mu(\sqrt{-g} \;\!\partial^\mu \phi)=0$, which becomes
 $
\partial_\mu (Vr^{n-1}  \partial^\mu \phi)=0.
$ 
Separating the differential operators with respect to the $\rho,t,\cdots$ coordinates, we obtain 
$$
\partial_\rho (r^{n-1} \partial_\rho \phi)=r^{n-1} \ddot{\phi} +V r^{n-3} \cLap \phi .
$$
 Integrating from $\rho_0:=\rho(r_0)=0$ to $\rho=\rho(r)$, we obtain 
 $$
 r^{n-1} \partial_\rho \phi -r_0^{n-1} \left. (\partial_\rho \phi)\right|_{\rho=\rho_0}
 =
 \int_{0}^\rho \left( r^{n-1} \ddot{\phi}+V r^{n-3} \cLap \phi \right) d\rho,
 $$
 and so 
 $$
 r^{n-1} V \partial_r \phi =r_0^{n-1} V(r_0) \left. (\partial_r \phi)\right|_{r=r_0} + \int_{0}^\rho \left( r^{n-1} \ddot{\phi}+V r^{n-3} \cLap \phi \right) d\rho,
 $$
that is, 
 $$
 \partial_r \phi =\Big( \frac{r_0}{r}\Big)^{n-1} \frac{V(r_0)}{V(r)} \left. (\partial_r \phi)\right|_{r=r_0}+ \frac{1}{r^{n-1} V}
 \int_{0}^\rho \left( r^{n-1} \ddot{\phi} +V r^{n-3} \cLap \phi \right) d\rho.
 $$
Hence 
\begin{eqnarray*}
 \!\!\!\!&&\|\partial_r \phi(r,\cdot)\|_{L^\infty(\mR\times S^{n-1})} \phantom{\Big(\frac{a(t_0)}{a(t)}\Big)^{n-2}}
 \\
 &\leq &\!\!\!\!\!\Big( \frac{r_0}{r}\Big)^{n-1} \frac{V(r_0)}{V(r)} \| (\partial_r \phi)(r_0,\cdot)\|_{L^\infty(\mR\times S^{n-1})} \\
 && \!\!\!\!\!+ \frac{1}{r^{n-1} V}\!
\int_{\rho_0}^\rho \!\left( r^{n-1} \|\ddot{\phi}(r,\cdot)\|_{L^\infty(\mR\times S^{n-1})} \!+\!V r^{n-3} \|\cLap \phi(r,\cdot)\|_{L^\infty(\mR\times S^{n-1})} \right)\! d\rho.
\end{eqnarray*}
Using the fact that $V\sim r^2$ for $r\geq r_0$, with $r_0$ large enough, and the estimates from Step 1 above, we obtain 
\begin{eqnarray*}
 \|\partial_r \phi(r,\cdot)\|_{L^\infty(\mR\times S^{n-1})} 
 &\lesssim &\frac{A}{r^{n+1}} +\frac{B}{r^{n+1}} \int_{\rho_0}^\rho  \left(r(\rho)\right)^{n-1}d\rho \phantom{\frac{1}{V}}\\
 &\lesssim & \frac{A}{r^{n+1}} +\frac{B}{r^{n+1}} \int_{r_0}^r  r^{n-1}\frac{1}{V(r)} dr \\
 &\lesssim &\frac{A}{r^{n+1}} +\frac{B'}{r^{n+1}} \int_{r_0}^r  r^{n-3} dr. 
\end{eqnarray*}
Recalling that $n>2$, we have  
$$
\|\partial_r \phi(r,\cdot)\|_{L^\infty(\mR\times S^{n-1})} \lesssim \frac{A}{r^{n+1}}+\frac{B'}{r^{n+1}}\frac{1}{(n-2)} (r^{n-2}-r_0^{n-2})\lesssim \frac{1}{r^3}.
$$
This completes the proof of Theorem~\ref{theorem_2b}.
\end{proof}

\goodbreak

\section{Appendix A: Fourier modes (de Sitter in flat FLRW form)}

\noindent  In this appendix, we give the details of the Fourier modal analysis that 
motivates the specific estimates given in Theorem~\ref{theorem_3},  
starting with   spatially periodic solutions to the Klein-Gordon equation. 

Let $\mT^n=\mR^n/(2\pi \mZ)^n$. Suppose that the `spatially periodic' $\phi:\mR\times \mT^n\rightarrow \mR$ 
satisfies the Klein-Gordon equation \eqref{KG_FLRW}. 
Writing 
$$
\phi=\sum_{\bbk \in \mZ^n} c_{\bbk}(t) e^{i \langle \bbk ,\bbx\rangle},
$$
\eqref{KG_FLRW} yields 
$$
-\ddot{c}_{\bbk}-\frac{n\dot{a}}{a}\dot{c}_\bbk+\frac{1}{a^2} \delta^{pq} i k_p ik_q c_{\bbk}-m^2c_{\bbk}=0,
$$
that is, 
$$
\ddot{c}_{\bbk}+\frac{n\dot{a}}{a}\dot{c}_\bbk+\frac{k^2}{a^2} c_{\bbk}+m^2c_{\bbk}=0,
$$
where 
$$
k^2:=\langle \bbk,\bbk\rangle.
$$
So 
\begin{eqnarray*}
 \frac{d}{dt}(a^n \dot{c}_\bbk)
 &=&
 na^{n-1} \dot{a} \dot{c}_\bbk +a^n \ddot{c}_\bbk\\
 &=& 
 na^{n-1} \dot{a} \dot{c}_\bbk +a^n \left( -\frac{n\dot{a}}{a}\dot{c}_\bbk-\Big(\frac{k^2}{a^2}+m^2\Big) c_{\bbk}\right)\\
 &=& 
 -a^n \left( \frac{k^2}{a^2}+m^2\right) c_\bbk,
\end{eqnarray*}
that is, 
\begin{equation}
\label{KG_Fourier}
\frac{d}{dt}(a^n \dot{c}_\bbk)+a^{n-2} (k^2+m^2 a^2) c_\bbk=0.
\end{equation}
Let 
$$
\tau=\int \frac{1}{a(t)} dt.
$$
Then 
$
\frac{d}{dt}=\frac{1}{a} \frac{d}{d\tau},
$ 
and so \eqref{KG_Fourier} becomes (with $ \frac{d}{d\tau}=:{\phantom{\cdot}}^{\prime}\;$) 
$$
\frac{1}{a}\left(a^n \frac{1}{a} c_\bbk'\right)'+a^{n-1}(k^2+m^2 a^2) c_\bbk=0,
$$
that is, 
\begin{equation}
 \label{KG_Fourier_conformal}
 (a^{n-1}c_\bbk')'+a^{n-1}(k^2+m^2 a^2) c_\bbk=0.
\end{equation}
Defining $d_\bbk$ by 
 $
c_\bbk =: a^{-\frac{n-1}{2}} d_\bbk,
$ 
we have 
$$
c_\bbk'=a^{-\frac{n-1}{2}} d_\bbk'-\frac{n-1}{2} a^{-\frac{n-1}{2}-1} a' d_\bbk.
$$
\eqref{KG_Fourier_conformal} yields 
$$
\left( a^{\frac{n-1}{2}} d_\bbk'-\frac{n-1}{2} a^{\frac{n-1}{2}-1} a'd_\bbk\right)'+a^{\frac{n-1}{2}}(k^2+m^2 a^2) d_\bbk=0,
$$
that is,
\begin{equation}
 \label{KG_Fourier_conformal_2}
d_\bbk''+\left( k^2 +m^2 a^2 -\frac{(n-1)}{2} \frac{a''}{a} -\frac{(n-1)(n-3)}{4} \Big(\frac{a'}{a}\Big)^2 \right)d_\bbk=0.
\end{equation}
Now if $a(t)=e^t$, then we may take $\tau=-e^{-t}$, so that $-t=\log (-\tau)$, that is, $-\tau=e^{-t}$. 
We remark that relative to our earlier use of conformal coordinates in \eqref{Aug_30_15:18} on page \pageref{Aug_30_15:18}, we are taking $t_0=+\infty$ for simplicity.
We then have 
$$
 a=e^{-\log (-\tau)}=-\frac{1}{\tau},\quad \quad 
 a'=\frac{1}{\tau^2},\quad \quad 
 a''=-\frac{2}{\tau^3}.
$$
Hence \eqref{KG_Fourier_conformal_2} becomes 
$$
d_\bbk''\!+\!\left(k^2 \!+\!\frac{m^2}{\tau^2} \!-\!\frac{n-1}{2}\Big( \!-\!\frac{2}{\tau^3}\Big)\Big(\!-\!\frac{\tau}{1}\Big)
\!-\!\frac{(n-1)(n-3)}{4}\Big(\frac{1}{\tau^2}\frac{(-\tau)^2}{1}\Big)\right)d_\bbk\!=\!0,
$$
that is,
\begin{equation}
 \label{KG_Fourier_conformal_a=exp_t}
 d_\bbk''+\left( k^2-\frac{\mu}{\tau^2}\right) d_\bbk=0,
\end{equation}
where 
$$
\mu:=n-1+\frac{(n-1)(n-3)}{4}-m^2.
$$
The general solution to this equation is\footnote{See for example \cite[p.95]{Wat}. For the relevant notation, 
see also \cite[pages 82,100,101]{Wat}.} given by 
\begin{equation}
\label{KG_Fourier_conformal_a=exp_t_gen_soln}
C_1 \sqrt{\tau} \;J_\nu (|k| \tau) + C_2 \sqrt{\tau}\;Y_\nu(|k|\tau),
\end{equation}
where $\nu$ satisfies 
$$
\nu^2=\frac{1}{4}+\mu=\frac{n^2}{4}-m^2.
$$
Here $J_\nu$ denotes the Bessel function of the first kind, 
$$
J_\nu (z)=\sum_{m=0}^\infty \frac{(-1)^m}{m!\;\!\Gamma(m+\nu+1)} \Big(\frac{z}{2}\Big)^{2m+\nu},
$$
and $Y_\nu$ is the Bessel function of the second kind, 
$$
Y_\nu(z)=\frac{J_\nu (z)\cos(\nu \pi)-J_{-\nu}(z)}{\sin(\nu \pi)},
$$
where the right hand side is replaced by its limiting value if $\nu$ is an integer.
Without loss of generality, in the solution \eqref{KG_Fourier_conformal_a=exp_t_gen_soln}, 
we may only consider $\nu$ such that $\textrm{Re}(\nu)\geq 0$. 

We note that as $t\rightarrow\infty$, $-\tau=e^{-t}\searrow 0$, and so $\tau\nearrow 0$. 
We now use the asymptotic expansions of $J_\nu(z)$ and $Y_\nu (z)$ as $z\nearrow 0$ 
(see e.g. \cite[9.1.7-9]{AS}):
\begin{itemize}
 \item[$\underline{1}^\circ$] If $\nu\neq 0$ (that is, $m\neq \pm\frac{n}{2}$), then as $\tau\nearrow 0$, we have 
 \begin{eqnarray*}
  J_\nu(|k|\tau)&=&C(-\tau)^\nu +O(|\tau|),\\
  Y_\nu(|k|\tau)&=&A (-\tau)^\nu+B(-\tau)^{-\nu}+C(-\tau)^{2-\nu}+O(|\tau|).
 \end{eqnarray*}
 So as $\tau\nearrow 0$ or $t\rightarrow \infty$, we have 
 \begin{eqnarray*}
  d_\bbk&=& A(-\tau)^{\frac{1}{2}+\nu}+B(-\tau)^{\frac{1}{2}-\nu} +C(-\tau)^{\frac{5}{2}-\nu}+O(|\tau |)\\
  &=&A e^{(-\frac{1}{2}-\nu)t}+Be^{(-\frac{1}{2}+\nu)t}+Ce^{(-\frac{5}{2}+\nu)t}+O(e^{-t}).
 \end{eqnarray*}
 But $c_\bbk=a^{-\frac{n}{2}+\frac{1}{2}} d_\bbk=e^{(-\frac{n}{2}+\frac{1}{2})t}d_\bbk$, and so 
 $$
 c_\bbk=Ae^{(-\frac{n}{2}-\nu)t}+Be^{(-\frac{n}{2}+\nu)t}+Ce^{(-\frac{n}{2}-2+\nu)t}+O(e^{-\frac{n+1}{2}t})
 $$
as $t\rightarrow+\infty$. We recall that $\textrm{Re}(\nu)\geq 0$, and so keeping only the dominating 
term, we have 
$$
|c_\bbk|= C'e^{-\left(\frac{n}{2}-\textrm{Re}(\nu)\right)t}+O(e^{-\frac{n+1}{2}t})
$$
as $t\rightarrow+\infty$. 
Thus we expect $\phi$ to satisfy 
$$
\|\phi(t,\cdot)\|_{L^\infty(\mR^n)}\lesssim 
\left\{ \begin{array}{ll} 
a^{-\frac{n}{2}} & \textrm{if } \;|m|> \frac{n}{2},\\
a^{-\frac{n}{2}+\sqrt{\frac{n^2}{4}-m^2}} &\textrm{if }\;|m|<\frac{n}{2}.
\end{array}\right.
$$
\item[$\underline{2}^\circ$] If $\nu=0$ (that is, $m=\pm \frac{n}{2}$), then as $\tau\nearrow 0$, we have 
 \begin{eqnarray*}
  J_\nu(|k|\tau)&=&C + O(|\tau|),\\
  Y_\nu(|k|\tau)&=&C\log (-\tau)+O(|\tau|).
 \end{eqnarray*}
 This implies that 
 $$
 d_\bbk= (A+Bt)e^{-\frac{1}{2}t} +O(e^{-t})
 $$
 as $t\rightarrow +\infty$.
 Hence 
 $$
 |c_\bbk|=(A+Bt)e^{-\frac{n}{2} t}+O(e^{-\frac{n+1}{2}t}).
 $$
 Thus we expect $\phi$ to satisfy 
$$
\|\phi(t,\cdot)\|_{L^\infty(\mR^n)}\lesssim  
a^{-\frac{n}{2}}\log a  \;\;\textrm{ if } \; m=\pm \frac{n}{2}.
$$
\end{itemize}
Summarising, $\phi$ is expected to have the decay 
$$
\|\phi(t,\cdot)\|_{L^\infty(\mR^n)}\lesssim 
\left\{ \begin{array}{ll} 
a^{-\frac{n}{2}} &\displaystyle  \textrm{if } |m|> \frac{n}{2},\phantom{a^{\sqrt{\frac{n^2}{4}-m^2}}}\\[0.2cm]
a^{-\frac{n}{2}}\log a &\displaystyle \textrm{if }|m|=\frac{n}{2},\phantom{a^{\sqrt{\frac{n^2}{4}-m^2}}}\\[0.2cm]
a^{-\frac{n}{2}+\sqrt{\frac{n^2}{4}-m^2}} &\displaystyle\textrm{if }|m|<\frac{n}{2}.
\end{array}\right.
$$
This motivates the decay estimates in Theorem~\ref{theorem_3}.

\goodbreak 

\section{Appendix B}

\noindent In this section, we prove the technical result we had used in the proof of Theorem~\ref{theorem_3}, in Section~\ref{Section_FLRW}. 

\smallskip 

\begin{lemma}
\label{technical_lemma_surface_integral}
If $f,g\in H^1(\mR^n)$, then 
 $\displaystyle 
\lim_{r\rightarrow +\infty}\int_{S_r} fg d\sigma_r=0.
$ 
\end{lemma}
\begin{proof} By the Cauchy-Schwarz inequality, 
$$
\left|\int_{S_r} fg d\sigma_r \right|^2\leq \int_{S_r} |f|^2 d\sigma_r \cdot \int_{S_r} |g|^2 d\sigma_r,
$$
and so, it is enough to show that 
$$
\lim\limits_{r\rightarrow\infty} \int\limits_{S_r}|f|^2 d\sigma_r=0.
$$
Suppose this does not hold. Then there exists an increasing sequence $(r_k)_{k}$ such that 
$r_k\stackrel{k\rightarrow \infty}{\longrightarrow}\infty$, and there exists an $\epsilon>0$ 
such that for each $k$, 
$$
\int_{S_{r_k}} |f|^2 d\sigma_{r_k}>\epsilon.
$$
(The plan is to use the trace theorem to fatten these $S_{r_k}$-slices to `annuli' $A_k$ and obtain 
$\|f\|_{H^1(A_k)}^2>\widetilde{\epsilon}>0$ for all $k$, giving rise to the contradiction that 
$$
\infty>\|f\|_{H^1(\mR^n)}^2 >\sum_k \|f\|^2_{H^1(A_k)} >\sum_{k} \tilde{\epsilon} =+\infty.
$$
So we will construct a subsequence $(r_{k_m})_m$ of $(r_k)_k$ and a sequence $(\delta_m)_m$ of positive 
numbers such that $r_{k_1}<r_{k_1}+\delta_1<r_{k_2}<r_{k_2}+\delta_2<r_{k_3}<\cdots$, and such that for the 
`annuli' $A_m:=\{\bbx: r_{k_m}<|\bbx|<r_{k_m}+\delta_m\}$, we have $\|f\|_{H^1(A_m)}^2>\widetilde{\epsilon}$. 
We will need to keep track of the constants in the trace theorems on our annuli $A_m$, and we will use the following 
\cite[p.41]{Gri}.)  

\begin{theorem}
\label{theorem_Grisvard}
Let $\Omega$ be a bounded open subset of $\mR^n$ with a Lipschitz boundary $\Gamma$. Then for $f\in H^1(\Omega)$ and for all $\epsilon\in (0,1)$, 
$$
\|f\|^2_{L^2(\partial\Omega)}\leq \frac{\|\bmu\|_{C^1(\overline{\Omega})}}{\delta} \left( \epsilon^{1/2} \|\nabla f\|_{L^2(\Omega)}^2 +(1+\epsilon^{-1/2})\|f\|_{L^2(\Omega)}^2\right),
$$
where $\bmu\in C^1(\overline{\Omega},\mR^n)$ is such that $\bmu\cdot \bbn \geq \delta$ on $\partial \Omega$, $\bbn$ being the outer normal vector. 
\end{theorem}

\noindent 
If $\Omega$ is an annulus $A=\{\bbx:r<\|\bbx\|<R\}$ (which is clearly bounded, open, and also it has the 
Lipschitz boundaries which are the two spheres $S_r$ and $S_R$), then with $\bmu(\bbx)=\bbx$, we have 
$$
\bmu\cdot \bbn=\|\bbx\|=\left\{\begin{array}{cc} R \textrm{ on } S_R,\\
                                r \textrm{ on } S_r 
                               \end{array}\right\}\geq r=:\delta.
$$
Also, if we take $\epsilon=1/4$, then 
$$
\|f\|_{L^2(S_r)}^2\leq \|f\|_{L^2(\partial A)}^2 \leq 3\frac{\|\bmu\|_{C^1(\overline{A})}}{r}\|f\|_{H^1(A)}^2.
$$
As 
$$
\|\bmu\|_{C^1(\overline{A})} =\max\limits_{\overline{A}} \|\bmu\|+ \max\limits_{\overline{A}} |\nabla\cdot\bmu|=R+n,
$$ 
we obtain 
$$
\|f\|_{L^2(S_r)}^2\leq 3\frac{R+n}{r} \|f\|_{H^1(A)}^2.
$$
Now we will construct $(r_{k_m})_{m}$ and $(\delta_m)_m$. 

\medskip 

\noindent 
We choose $k_1$ such that $r_{k_1}>n$. Let $\delta_1$ be such that $0<\delta_1<r_{k_1}-n$. 
Then for the annulus 
$$
A_1:=\{ \bbx: r_{k_1}<\|\bbx\|<r_{k_1}+\delta_1\},
$$
 we have 
$$
\|f\|^2_{H^1(A_1)} \geq \frac{r_{k_1}/3}{(r_{k_1}+\delta_1)+n} \|f\|^2_{L^2(S_{r_{k_1}})} 
\geq \frac{1/3}{1+\frac{\delta_1+n}{r_{k_1}}}\epsilon>\frac{1/3}{1+1}\epsilon=\frac{\epsilon}{6}=:\widetilde{\epsilon}.
$$
Now suppose $r_{k_1},\cdots, r_{k_m}, \delta_1,\cdots, \delta_m$ possessing the desired properties have been constructed. 
Choose $k_{m+1}$ such that $r_{k_{m+1}}>r_{k_m}+\delta_m$. Let $\delta_{m+1}$ be such that 
$0<\delta_{m+1}<r_{k_{m+1}} -n$.

\smallskip 

\noindent Then for the annulus 
$$
A_{m+1}:=\{\bbx: r_{k_{m+1}}<\|\bbx\|<r_{k_{m+1}}+\delta_{m+1}\},
$$ 
we have 
$$
\|f\|_{H^1(A_{m+1})}^2 >\frac{r_{k_{m+1}}/3}{(r_{k_{m+1}}+\delta_{m+1})+n} \|f\|_{L^2(S_{r_{k_{m+1}}})}
\geq \frac{1/3}{1+\frac{\delta_{m+1}+n}{r_{k_{m+1}}}} \epsilon>\frac{\epsilon}{6}=\widetilde{\epsilon}.
$$
This completes the induction step. 

\smallskip 

\noindent 
So we have arrived at the contradiction that 
$$
+\infty >\|f\|_{H^1(\mR^n)}^2 \geq \sum_m \|f\|^2_{H^1(A_m)} 
\geq \sum_m \widetilde{\epsilon} =+\infty.
$$
This shows that our original assumption was incorrect, and so 
$$
\lim_{r\rightarrow \infty}\int_{S_r} |f|^2 d\sigma_r=0,
$$
completing the proof of our lemma. 
\end{proof}

\noindent An analogous result also holds for the cylinder $\mR\times S^{n-1}$. This was used in the proof of 
our Theorem~\ref{theorem_4}. 

\begin{lemma}
\label{cylinder_technical_lemma_surface_integral}
If $n\geq 3$ and $f,g\in H^1(\mR\times S^{n-1})$, then 
$$
\displaystyle 
\lim_{t\rightarrow +\infty}\int_{S^{n-1}} fg d\Omega=0=\lim_{t\rightarrow -\infty} \int_{S^{n-1}} fg d\Omega.
$$
\end{lemma}
\begin{proof}(Sketch.) The proof is based on the same idea as the above, but is somewhat simpler, 
 since the radius of $S^{n-1}$ doesn't change,  
and the constants one has in the trace theorem for a `cylindrical band' of the form
$(a,b) \times  S^{n-1}$ already work, as opposed to having to keep careful track, via Theorem~\ref{theorem_Grisvard}, of the 
constants in the earlier case when the radii of the $S_{r}^{n-1}$ were changing. Procceding in the same way as in the previous lemma, we assume that 
$$
\neg \left(\lim_{t\rightarrow +\infty}\int_{S^{n-1}} |f|^2 d\Omega=0\right),
$$
and so there exists an $\epsilon>0$ and a sequence $(t_k)_{k\in \mN}$ such that $\lim\limits_{k\rightarrow \infty}t_k=+\infty$, and
$$
\lim_{k\rightarrow +\infty}\int_{S^{n-1}} |f(t_k,\cdot)|^2 d\Omega > \epsilon.
$$
In order to fatten the  
`circle' $\{t_k\}\times S^{n-1}$ to a cylindrical band of the form $I=(t_k,t_k+\delta)\times S^{n-1}$, while keeping the $L^2$-norm of $f$ 
on the band uniformly (in $k$) bigger than a fixed positive quantity, one can use the inequality 
$$
\|f(t_k,\cdot)\|_{L^2(S^{n-1})}\leq C \|f\|_{H^1(I\times S^{n-1})}.
$$
This follows from  \cite[Prop. 4.5, p.287]{Tay}, by taking $\Omega=[t_k,t_k+\delta] \times S^{n-1}$.  
The rest of the proof is the along the same lines. 
\end{proof}

\goodbreak 

\section{Appendix C: Sharpness of bound when $|m|=\frac{n}{2}$ in Theorem~\ref{theorem_3} }$\;$

\noindent In this appendix, we will show the sharpness of the bound from Theorem~\ref{theorem_3} we had obtained for the decay of the solution to 
the Klein-Gordon equation in the de Sitter universe in flat FLRW form, when $|m|=\frac{n}{2}$. Let us recall  this bound:
$$
\forall t\geq t_0,\;\;  \|\phi(t,\cdot)\|_{L^\infty(\mR^n)} \lesssim a^{-\frac{n}{2}}\log a.
$$
\noindent 
If $|m|=\frac{n}{2}$, then with $\psi:=a^{\frac{n}{2}}\phi$, we had seen that 
$$
\ddot{\psi} -\frac{1}{a^2} \Delta \psi=0.
$$
We will now construct a solution $\psi$ that satisfies 
$$
 \|\psi(t,\cdot)\|_{L^2(\mR^n)}\sim A+Bt \;\textrm{ as } t\rightarrow \infty,
 $$
 showing that 
 $$
 \|\phi(t,\cdot)\|_{L^2(\mR^n)} \sim (A+Bt) a^{-\frac{n}{2}} \;\textrm{ as } t\rightarrow \infty,
 $$
 and so the bound 
 $$
 \|\phi(t,\cdot)\|_{L^2(\mR^n)} \lesssim (A+Bt) a^{-\frac{n}{2}} \; \textrm{ for all large }t
 $$
 cannot be improved. 
 
 \smallskip 
 
 \noindent 
 We want 
 \begin{equation}
 \label{KG_FLRW_m=n/2}
 \ddot{\psi}-\frac{1}{e^{2t}} \Delta \psi=0.
 \end{equation}
 Taking the Fourier transform with respect to only the (spatial) $\bbx$-variable, and denoting 
 $$
 \widehat{\psi}(t,\bxi):=\int_{\mR^n} \psi(t,\bbx) e^{i\langle \bxi, \bbx\rangle } d^n\bbx,
 $$
 \eqref{KG_FLRW_m=n/2} becomes
 \begin{equation}
 \label{FT_KG_FLRW_m=n/2}
 \frac{\partial^2}{\partial t^2} \widehat{\psi} (t,\bxi)+\frac{\|\bxi\|^2}{e^{2t}} \widehat{\psi}(t,\bxi)=0,
 \end{equation}
 which is a family of ordinary differential equations in $t$, parameterised by $\bxi \in \mR^n$. 
 For a fixed $\bxi\in \mR^n$, the general solution to the ODE \eqref{FT_KG_FLRW_m=n/2} is given by 
 $$
 \widehat{\psi}(t,\bxi)=C_1(\bxi)\cdot  J_0(\|\bxi\| e^{-t})+C_2(\bxi)\cdot  Y_0(\|\bxi\|e^{-t}),
 $$
 where 
 \begin{itemize} 
  \item $J_0$ is the Bessel function of first kind and of order $0$, and 
  \item $Y_0$ is the Bessel function of second kind and of order $0$.
 \end{itemize}
 In order to construct our $\psi$, we will make special choices of $C_1$ and $C_2$. 
 
 \smallskip 
 
 \noindent 
 We recall \cite[(9.1.7-8)]{AS} that 
 \begin{eqnarray*}
  J_0(z)&\sim& 1,\\
  Y_0(z) & \sim& \frac{2}{\pi} \log z
 \end{eqnarray*}
 as $z\searrow 0$ ($z\in \mR$). 
 
 \noindent 
 Now as $t\rightarrow \infty$, $e^{-t}\searrow 0$, and so from the above limiting behaviour of $J_0$ and $Y_0$, 
 we obtain that as $t\rightarrow\infty$, 
 \begin{eqnarray*}
  \widehat{\psi}(t,\bxi)&\sim &C_1(\bxi) \cdot 1+C_2(\bxi) \cdot \left(\frac{2}{\pi} \log (\|\bxi\|e^{-t})\right)\\
  &=&
  C_1(\bxi)+\frac{2}{\pi} C_2(\bxi) \log \|\bxi\|-\frac{2}{\pi} t\cdot  C_2(\bxi).
 \end{eqnarray*}
By Plancherel's identity (see e.g. \cite[Prop. 3.2]{Tay}) 
$$
\|\widehat{\psi}(t,\cdot)\|_{L^2(\mR^n)}=\|\psi(t,\cdot)\|_{L^2(\mR^n)}.
$$
Since we want the linear behaviour in $t$ of $\|\psi(t,\cdot)\|_{L^2(\mR^n)}$,  
we keep $C_2$ nonzero, but may take $C_1\equiv 0$. Then as $t\rightarrow \infty$, 
$$
\widehat{\psi}(t,\bxi)=C_2(\bxi) \cdot Y_0(\|\bxi\|e^{-t}).
$$
In order to have $\widehat{\psi}(t,\cdot)$ (and so also $\psi(t,\cdot)$) in 
$L^2(\mR^n)$ for all $t$, we choose $C_2$ to have a sufficiently fast decay. 

\medskip 

\noindent 
We recall \cite[\S9.2.2]{AS} that 
$$
Y_0(z)= \sqrt{\frac{2}{\pi z}} \left(\sin \Big(z-\frac{\pi}{4}\Big)+O\Big(\frac{1}{|z|}\Big)\right) 
$$
as $z\rightarrow \infty$ ($z\in \mR$). So we have 
$$
Y_0(\|\bxi\| e^{-t})
=
\sqrt{\frac{2}{\pi \|\bxi\|e^{-t}}} \left(\sin \Big(\|\bxi\|e^{-t}-\frac{\pi}{4}\Big)+O\Big(\frac{1}{\|\bxi\|e^{-t}}\Big)\right) 
$$
as $\|\bxi\|\rightarrow +\infty$ (and $t$ is kept fixed). So to arrange $\widehat{\psi}(t,\cdot)\in L^2(\mR^n)$ for all $t$, we may take 
$$
C_2(\bxi):=\frac{\|\bxi\|}{(\|\bxi\|^2+1)^{1+\frac{n}{4}}}.
$$
(Also this choice makes 
$$
\bxi \mapsto C_2(\bxi)\log \|\bxi\|\in L^2(\mR^n),
$$
which will be needed below.) 

\medskip 

\noindent 
Then $\widehat{\psi}(t,\cdot)\in L^2(\mR^n)$ for all $t$. Also, as $t\rightarrow \infty$, 
$$
\widehat{\psi}(t,\bxi)
\sim 
\frac{2}{\pi} \bigg( \underbrace{\frac{\|\bxi\|}{(\|\bxi\|^2+1)^{1+\frac{n}{4}}} \log\|\bxi\|}_{=:f\in L^2(\mR^n)} 
-t \underbrace{\frac{\|\bxi\|}{(\|\bxi\|^2+1)^{1+\frac{n}{4}}}}_{=:g\in L^2(\mR^n)} \bigg),
$$
and 
$$
\|\widehat{\psi}(t,\cdot)\|_{L^2(\mR^n)} 
\geq 
\frac{2}{\pi}\Big(t\underbrace{\|g\|_{L^2(\mR^n)}}_{\neq 0}-\|f\|_{L^2(\mR^n)}\Big)\geq 0
$$
for large $t$.

\end{document}